\documentclass[11pt,leqno]{article}
\usepackage{amssymb,amsfonts}%amsLatex
\usepackage{amsmath,amsthm,amsxtra}
\usepackage[all]{xy}

\newcommand{\st}[1]{\ensuremath{^{\scriptstyle \textrm{#1}}}}

\newcommand\bigcheck[1]{#1 \raise1ex\hbox{$\hspace{-1ex}{}^\vee$}}
\newcommand\sucheck[1]{#1 \raise0.5ex\hbox{$\hspace{-1ex}{}^\vee$}}

\newcommand{\alphaparenlist}{% changes enumerate 1st level to (a)...(z)
  \renewcommand{\theenumi}{\alph{enumi}}%
  \renewcommand{\labelenumi}{(\theenumi)}%
}
\newcommand{\arabiclist}{% changes enumerate 1st level to 1. ...9.
  \renewcommand{\theenumi}{\arabic{enumi}}%
  \renewcommand{\labelenumi}{\theenumi.}%
}
\newcommand{\romanparenlist}{% changes enumerate 1st level to (i)...(ix)
  \renewcommand{\theenumi}{\roman{enumi}}%
  \renewcommand{\labelenumi}{(\theenumi)}%
}
\newcommand{\romanlistii}{% changes enumerate 2nd level to i. ... x.
  \renewcommand{\theenumii}{\roman{enumii}}%
  \renewcommand{\labelenumii}{(\theenumii)}%
}

\newcommand{\ad}{\mathop{\rm ad}\,}

\newcommand{\Aut}{{\rm Aut}}

\newcommand{\ch}{{\rm ch}}

\newcommand{\End}{\mathop{\rm End }}
\newcommand{\Gr}{{\rm Gr}}
\newcommand{\im}{\mathop{\rm im  \, }}

\newcommand{\mult}{{\rm mult}}
\renewcommand{\ne}{\mathop{\rm ne}\,}
\newcommand{\new}{\mathop{\rm new}\,}

\newcommand{\re}{\mathop{\rm re  \, }}

\newcommand{\Res}{\mathop{\rm Res  \, }}
\newcommand{\Sh}{{\rm Sh}}
\renewcommand{\sl}{s\ell}
\newcommand{\str}{{\rm str}}

\newcommand{\tr}{\rm tr \, }
\newcommand{\tw}{\rm tw \, }
\newcommand{\vac}{|0\rangle}

\newcommand{\C}{\mathcal{C}}

\renewcommand{\O}{\mathcal{O}}

\newcommand{\CC}{\mathbb{C}}
\newcommand{\NN}{\mathbb{N}}
\newcommand{\QQ}{\mathbb{Q}}
\newcommand{\RR}{\mathbb{R}}
\newcommand{\ZZ}{\mathbb{Z}}

\newcommand{\fg}{\mathfrak{g}}
\newcommand{\fh}{\mathfrak{h}}

\newcommand{\fn}{\mathfrak{n}}

\renewcommand{\hat}{\widehat}
\makeatletter
\renewcommand\section{\@startsection {section}{1}{\z@}%
                                   {-3.5ex \@plus -1ex \@minus -.2ex}%
                                   {2.3ex \@plus.2ex}%
                                   {\normalfont\large\bfseries}}
\renewcommand\subsection{\@startsection{subsection}{2}{\z@}%
                                     {-3.25ex\@plus -1ex \@minus -.2ex}%
                                     {0ex \@plus .0ex}%
                                     {\normalfont\normalsize\bfseries}}

\@addtoreset{equation}{section}
\makeatother

\newtheorem{theorem}{Theorem}[section]
\newtheorem{definition}{Definition}[section]
\newtheorem{lemma}{Lemma}[section]
\newtheorem{corollary}{Corollary}[section]
\newtheorem{proposition}{Proposition}[section]
\newtheorem*{lemma*}{Lemma}

\theoremstyle{remark}
\newtheorem{remark}{Remark}[section]

\newtheorem{examples}{Examples}[section]

\makeatletter
\def\@maketitle{\newpage
 \null
 \vskip 2em
 \begin{center}%
  \vskip 3em
  {\Large\bf \@title \par}%
  \vskip 1.5em
  {\normalsize
   \lineskip .5em
   \begin{tabular}[t]{c}\@author
   \end{tabular}\par}%
  \vskip 2em

 \end{center}%
 \par
 \vskip 2.5em}
\makeatother
\begin{document}

\title{On rationality of $W$-algebras}

\author{Victor G. Kac\thanks{Department of Mathematics, M.I.T.,
    Cambridge, MA 02139,
    USA.~~kac@math.mit.edu}~~\thanks{Supported in part by NSF
    grant   DMS-0501395.}~~and Minoru Wakimoto\thanks{~~wakimoto@r6.dion.ne.jp}
~~\thanks{Supported in part by Department of Mathematics MIT,
Clay Mathematical Institute, grant-in-aid for scientific research
A-14204003, and the 21st century COE program at RIMS.}\\[4ex]
Dedicated to Bertram Kostant on his 80\st{th} birthday.
}

\maketitle

\noindent{\textbf{Absract:}}  We study the problem of
classification of triples ($\fg ,f, k$), where $\fg$ is a simple Lie algebra, 
$f$ its nilpotent element and $k \in \CC$, for which the simple $W$-algebra 
$W_k (\fg ,f)$ is rational.

\vspace{2ex}

\section{Introduction}
\label{sec:intro}

A vertex algebra $V$, used to construct a rational conformal field theory,
must satisfy at least the following three conditions:

(a) V has only finitely many irreducible representations
$\{M_j \}_ {j \in J}$,

(b) the normalized characters  
$\chi_j(\tau) = \tr_{M_j}e^{2\pi i \tau (L_0 - c/24)}$ 
converge to holomorphic 
functions on the complex upper half-plane $\CC^+$,

(c) the functions $\{\chi_j(\tau)\}_ {j \in J}$ span an $SL_2 (\ZZ)$-invariant
space.

A vertex algebra $V$ is called rational if it satisfies these three 
properties. Recall that any semisimple vertex operator algebra,
i.e. a vertex operator algebra for which any representation is completely
reducible, is rational (see \cite{Z}, \cite{DLM}).
However, it is unclear how to verify this condition for vertex algebras,
considered in the present paper.    
%An important problem in the theory of vertex algebras is the
%classification of semisimple vertex algebras, that is the vertex
%algebras for which any representation is completely reducible.
%Such vertex algebras have a number of nice properties.  For
%example they have only finitely many irreducible representations
%and their characters are modular functions, spanning
%a finite-dimensional 
%$SL_2 (\ZZ)$-invariant space 
%\cite{Z}, \cite{DLM}.  
%Consequently, semisimple
%vertex algebras can be used to construct rational conformal field
%theories.

It is well known that lattice vertex algebras, associated to even
positive definite lattices, are rational (and semisimple as well).  
A simple Virasoro
vertex algebra is rational (and semisimple as well) iff its central 
charge is of the
form $c=1-\frac{6(p-p')^2}{pp'}$, where $p,p'$ are relatively
prime integers, greater than~$1$ (which are central charges of
the so called minimal models \cite{BPZ}). A simple affine vertex
algebra 
$V_k (\fg)$, 
attached to a simple Lie algebra $\fg$
is rational (and semisimple as well) iff its level~$k$ is a non-negative
integer.  (Simplicity of a vertex operator algebra is a necessary, but by
far not sufficient, condition of semisimplicity.)

It follows from 
\cite{KW1}, Theorems 3.6 and 3.7, and \cite{KW2}, Remark 4.3(a), 
that for a rational $k$ of the form
\begin{equation}  
\label{eq:0.1}  
k =- h^{\vee} + \frac{p}{u}, \hbox{  where  }  
(p,u) =1 \, , \, u \ge 1 \, , \, (u,\ell)=1 (\hbox{resp.}=\ell)\, , 
p \ge h^\vee (\hbox{resp.} \ge h)\, , \, 
\end{equation}
where $h$ is the Coxeter number, $h^{\vee}$ is the dual Coxeter number and 
$\ell (=1,2 \hbox{ or } 3)$ 
is the ``lacety'' of $\fg$, the normalized character of the vertex algebra
$V_k(\fg)$ is a modular function (conjecturally, these are all $k$ with
this property). It was shown in \cite{KW1}, Theorem 3.6, that for these $k$
with $(u,\ell)=1$ the affine Lie algebra $\hat{\fg}$ 
has a finite set of irreducible highest weight modules
$\{ M_j \}_{j \in J}$ (called admissible), whose regularized normalized 
characters 
$\chi_j(\tau,z) = \tr_{M_j}e^{2\pi i \tau (L_0 - c/24)+z}$,
$z \in \fg$, span an $SL_2 (\ZZ)$-invariant space (if $ (u,\ell)=\ell$,
then one has only $\Gamma_0 (\ell)$-invariance, see \cite{KW2}, Remark 4.3(a)). 
Conjecturally, these $\hat{\fg}$-modules extend to 
$V_k (\fg)$ and are all of its irreducible
modules (this conjecture was proved in \cite{AM} for $\fg = \sl_2$).   
Thus, $V_k (\fg)$ for $k$ of the form (\ref{eq:0.1}) with $(u,\ell)=1$
satisfy the properties 
(a) and (c) of rationality, but property (b) fails for some $j \in J$ since
$\chi_j(\tau,z)$ may have a pole at $z=0$, unless $k$ is a non-negative
integer.
  
In the present paper we study the problem of rationality of
simple $W$-algebras $W_k (\fg , f)$, which is a family of vertex
algebras, depending on $k \in \CC$, attached to a simple Lie
algebra~$\fg$ and a nilpotent element $f$ of $\fg$ (rather its
conjugacy class) \cite{KRW}, \cite{KW3}.  
More precisely, we need to analyze for which
triples $(\fg , f,k)$ the $W_k (\fg ,f)$-modules, obtained by the
quantum Hamiltonian reduction from admissible modules of
level~$k$ over the affine Lie algebra~$\hat{\fg}$, have convergent
characters, as the modular invariance property is preserved by this reduction. 
%which span a finite-dimensional $SL_2 (\ZZ)$-invariant space,
%i.e.,~the above-mentioned necessary conditions of semisimplicity hold.  
We call such $f$ an exceptional nilpotent and such $k$ an
exceptional level.  

%Since conjecturally only admissible characters of $\hat{\fg}$
%have modular invariance property,
% (provided that $k$ is not the critical level $-h^\vee$), 
%and the quantum Hamiltonian
%reductions preserve this property, we expect that no other
%triples $(\fg ,f,k)$ correspond to semisimple $W$-algebras.  We
%also expect that the above-mentioned conditions are sufficient in
%our context.

The most well studied case of $W$-algebras is that corresponding
to the principal nilpotent element $f$ (a special case of which
for $\fg = \sl_2$ is the Virasoro vertex algebra).  It follows
from \cite{KW1} and \cite{FKW} that for principal $f$ the
exceptional levels $k$ are given by
\begin{equation}  
\label{eq:0.2}  
k =- h^{\vee} + \frac{p}{u}, \hbox{  where  }  (p,u) =1 \, , 
\, p \ge h^\vee \, , \, u \ge h \, , \, (u,\ell)=1 \, .
\end{equation}
%
%where $h$ is the Coxeter number of $\fg$.  
We expect that for the principal nilpotent $f$,
(\ref{eq:0.2}) 
are precisely the values of $k$, for which $W_k
(\fg ,f) $ is a semisimple vertex algebra.

Surprisingly, beyond the principal nilpotent, there are very few
exceptional nilpotents.  We conjecture that there exists an order
preserving map of the set of non-principal
exceptional nilpotent orbits of $\fg$ to the set of positive integers,
relatively prime to~$\ell$ and smaller than $h$, such that 
the corresponding integer~$u$ is the only denominator of 
%$k+h^{\vee}$
an exceptional level~$k=-h^\vee +p/u$, where $p \ge h^\vee$ and $(u,p)=1$.

Note that $f=0$ is an exceptional nilpotent, corresponding to $u
=1$, since in this case $W_k (\fg ,0)$ is the simple affine
vertex algebra of level $k \in \ZZ_+$.

We prove the above conjecture for $\fg \simeq \sl_n$.  In this
case the above map is bijective to the set
$\{ 1,2,\ldots , h-1 \}$,  and the exceptional
nilpotent, corresponding to the positive integer $u \leq h=n$, is
given by the partition $n=u+\cdots + u+s$, where $0 \le s<u$.

For an arbitrary simple $\fg$ we give a geometric description of the 
exceptional pairs $(k,f)$ in terms of $\fg$ and its adjoint group.

We are grateful to K. Bauer, A. Elashvili and A. Premet
for very useful discussions on nilpotent orbits, and to ESI and IHES
for their hospitality. The results of the paper were reported at the Weizmann 
Institute in January 2007 and at a conference in Varna in June 2007.

%We checked this conjecture for low rank classical Lie algebras and $G_2$, 
%all cases of rank~$\le 4$, except
%for $F_4$, and found out that in all these cases, except for
%$B_2$ and $B_3$, 
%and found that the above map is not bijective in general.
%We give a precise form of this conjecture for $so_{2n+1}$ and $sp_{2n}$.

\section{Admissible modules over affine Lie algebras.}
\label{sec:1}

\subsection{Description of the vacuum admissible weights.}~~
\label{sec:1.1}
Let $\hat{\fg}$ be a (non-twisted) affine Lie algebra, associated
to a simple finite-dimensional Lie algebra $\fg$ \cite{K1}.  Recall
that
\begin{displaymath}
  \hat{\fg} = \fg [t,t^{-1}]\oplus \CC K \oplus \CC D 
\end{displaymath}
with commutation relations:
\begin{eqnarray*}
  [at^m,bt^n] &=& [a,b]t^{m+n} + m \,\delta_{m,-n} (a|b)K\, , \\
 {} [D, at^m] &=& mat^m \, , \, [K,\hat{\fg}] =0 \, , 
\end{eqnarray*}
where $(\, . \, | \, . \, )$ denotes the symmetric invariant
bilinear form on $\fg$, normalized by the condition $(\alpha |
\alpha) =2$ for long roots $\alpha$.

Recall that the bilinear form $(\, . \, | \, . \, )$ extends from
$\fg$ to a symmetric invariant bilinear form on $\hat{\fg}$ by
letting 
\begin{eqnarray*}
  (t^m a|t^n b) = \delta_{m,-n} (a|b) \,, \, 
     (\fg [t,t^{-1}] | \CC K + \CC D) =0 \, , \\
     (K|K) = (D|D) =0 \, , \quad (K|D)=1 \, .
\end{eqnarray*}

Choose a Cartan subalgebra $\fh$ of $\fg$, then
\begin{displaymath}
  \hat{\fh} = \fh + \CC K + \CC D
\end{displaymath}
is an ad-diagonalizable subalgebra of $\hat{\fg}$, called its
Cartan subalgebra.  Since the restriction of the bilinear form
$(\, . \, | \, . \, )$ to $\hat{\fh}$ is non-degenerate, we can
(and often will) identify $\hat{\fh}^*$ with $\hat{\fh}$, using
this form.  Given $\alpha \in \hat{\fh}$ such that $(\alpha |
\alpha) \neq 0$, we denote $\alpha^\vee = 2 \alpha /(\alpha
|\alpha)$, unless otherwise specified.

Let $\Delta \subset \fh^*$ be the set of roots of $\fg$, choose a
subset of positive roots $\Delta_+$, and let $\prod = \{ \alpha_1
, \ldots ,\alpha_r \}$ be the set of simple roots, where $r$ is the
rank of $\fg$.  Let
$\Delta_+^\vee = \{ \alpha^\vee |\alpha \in \Delta_+ \}$ be the set
of positive coroots,  $\prod^\vee = \{ \alpha^\vee_1,\ldots ,
\alpha^\vee_r \}$ the set of simple coroots of $\fg$.  Define, as
usual, $\rho$ and $\rho^\vee$ by:
\begin{displaymath}
(\rho | \alpha^\vee_i) = 1 \, , \quad (\rho^\vee | \alpha_i) =1 
   \, , \quad i=1,\ldots , r \, .
\end{displaymath}

Recall that the sets of roots $\hat{\Delta}$ and coroots
$\hat{\Delta}^\vee$ of $\hat{\fg}$ are
\begin{displaymath}
  \hat{\Delta} = \hat{\Delta}^{\re} \cup \hat{\Delta}^{\im} \, , \,
     \hat{\Delta}^\vee = \hat{\Delta}^{\vee,\re} \cup
        \hat{\Delta}^{\im}\, , \, 
\end{displaymath}
where the sets of real roots and coroots are
\begin{displaymath}
  \hat{\Delta}^{\re} = \{ \alpha +nK | \alpha \in \Delta ,n\in\ZZ\}
    \, , \, \hat{\Delta}^{\vee,\re}=
      \{ \alpha^\vee | \alpha \in 
         \hat{\Delta}^{\re} \}\, ,
\end{displaymath}
and the set of imaginary roots (resp. coroots) is
\begin{displaymath}
  \hat{\Delta}^{\im} = \{ nK  | n \in \ZZ \, , \, n\neq 0 \}\, ,
\end{displaymath}
the subsets of positive roots and coroots being:
\begin{eqnarray*}
  \hat{\Delta}^{\re}_+ &=& \Delta_+ \cup \{ \alpha + nK | 
      \alpha \in \Delta ,n>0 \} \, , \, \hat{\Delta}^{\vee ,\re}_+
%      = \Big\{ \frac{2\alpha}{(\alpha | \alpha)} \Big| 
       =\{\alpha^{\vee}|
        \alpha \in \hat{\Delta}^{\re}_+ \} \, , \\
   \hat{\Delta}^{\im}_+ &=& \{ nK |n>0 \} \, , \, 
    \hat{\Delta}_+ = \hat{\Delta}^{\re}_+ \cup 
       \hat{\Delta}^{\im}_+ \, , \, \hat{\Delta}^\vee_+ =
       \hat{\Delta}^{\vee ,\re}_+ \cup \hat{\Delta}^{\im}_+ \, .
\end{eqnarray*}
Recall that the sets of simple roots (resp. coroots) in
$\hat{\Delta}_+$ (resp. $\hat{\Delta}^\vee_+$) then are:
\begin{displaymath}
  \hat{\prod} = \{ \alpha_0 = K-\theta \, , \, 
  \alpha_1 ,\ldots ,\alpha_r \} \, , \, 
    \hat{\prod}^\vee = \{ \alpha^\vee_0 = K-\theta^\vee \, , \, 
      \alpha^\vee_1 ,\ldots , \alpha^\vee_r \} \, , 
\end{displaymath}
where $\theta$ is the highest root in
$\Delta_+$.  The positive integers 
\begin{displaymath}
h=(\rho^\vee | \theta ) +1
\hbox{   and   } h^\vee = (\rho |\theta^\vee) +1
\end{displaymath}
are called the Coxeter number and the dual Coxeter number of
$\fg$, respectively.

Defining $\hat{\rho} = h^\vee D + \rho$, $\hat{\rho}^\vee =h D +
\rho^\vee \in \hat{\fh}$, we  have the usual formulas:
\begin{displaymath}
  (\hat{\rho} |\alpha_i^\vee) =1 \, , \, 
     (\hat{\rho}^\vee |\alpha_i) =1 \, , \, i= 0,1,\ldots ,r \, .
\end{displaymath}

Given $\lambda \in \hat{\fh}^*$, let
\begin{displaymath}
  \hat{\Delta}^\vee_{\lambda} = \{ \alpha^\vee \in
  \hat{\Delta}^\vee | (\lambda | \alpha^\vee) \in \ZZ \} \, .
\end{displaymath}
Let $\hat{\Delta}^\vee_{\lambda +} = \hat{\Delta}^\vee_\lambda
\cap \hat{\Delta}^\vee_+$ and let $\hat{\prod}^\vee_\lambda$ be the set
of simple roots in $\hat{\Delta}^\vee_{\lambda +}$.

\begin{definition}[\cite{KW1}]

A weight $\lambda \in \hat{\fh}^*$ is called \emph{admissible} if
the following two conditions hold:
\begin{subequations}
\begin{equation}
  \label{eq:1.1a}
 (\lambda + \hat{\rho} |\alpha^\vee ) \notin 
 \{ 0,-1,-2,\ldots \}
    \hbox{  for all  } \alpha^\vee \in \hat{\Delta}^\vee_+ \, ,
    \end{equation}
\begin{equation}
  \label{eq:1.1b}
  \hbox{the  } \QQ\hbox{-span of  } \hat{\Delta}^\vee_{\lambda}
  \hbox{  contains  } \hat{\Delta}^\vee \, .
\end{equation}
\end{subequations}
\end{definition}

We proved in \cite{KW1} that the (normalized) character of an
irreducible highest weight $\hat{\fg}$-module with highest weight
$\lambda$ satisfies a modular invariance property if
$\lambda$ is an admissible weight (and conjectured that for no
other $\lambda$ this property holds).  Admissible weights were
completely classified in \cite{KW1}, and this classification will be
used in the paper.  The following proposition describes the
``vacuum'' admissible weights.

\begin{proposition}
  \label{prop:1.1}

%Let $\lambda \in \hat{\fh}^*$ be such that 
%
%\begin{displaymath}
%  \lambda (v)= 0 \hbox{  for  } v \in \fh \, , \, 
%  \lambda (K) =k \in \CC \, .
%\end{displaymath}
%
%
For $k \in \CC$ the weight $\lambda=kD$ is admissible if and only if
$k$ satisfies the following properties (a) and (b):

\alphaparenlist
\begin{enumerate}
\item %%a
$k+h^\vee = \frac{p}{u}$, where $p,u \in \NN$ , $(p,u)=1$;

\item %%b
one of the following possibilities holds:
\begin{list}{}{}
\item (i)~~$(u,\ell) =1$ and $p \ge h^\vee$ (in this case
  $\hat{\prod}^\vee_\lambda =\{ uK-\theta^\vee\, , \, \alpha^\vee_1
  ,\ldots ,\alpha^\vee_r\}$),

\item (ii)~$u \in \ell\,  \NN$ and $p \ge h$ (in this case
  $\hat{\prod}^\vee_{\lambda} =\{ uK -\theta^\vee_s \, , \,
    \alpha^\vee_1 ,\ldots , \alpha^\vee_r \}$).
\end{list}

\end{enumerate}
Here $\ell$ is the lacety of $\fg $ (i.e.,~$\ell = 1$ for $A-D-E$
types,  $\ell =2$ for $B-C-F$ types, and $\ell =3$ for $G_2$),
and $\theta_s$ is the highest among short roots in $\Delta_+$.

\end{proposition}

\begin{proof}

Condition (\ref{eq:1.1b}) of admissibility implies that $(\lambda
+\hat{\rho} |\alpha_i)\in\QQ$ for all $i=0,1,\ldots ,r$.  This,
together with (\ref{eq:1.1a}) for $\alpha^\vee = nK$ implies  (a).

Next, note that $(\lambda + \hat{\rho} | \alpha^\vee) = (\rho
|\alpha^\vee) \in \ZZ$ if $\alpha \in \Delta^{\vee}$, hence
$\Delta^\vee \subset \Delta^\vee_{\lambda}$.  Since $(\lambda +
\hat{\rho}| nK +\alpha^\vee ) =n(k+h^\vee) + (\rho |
\alpha^{\vee}) \in n \frac{p}{u} +\ZZ$, we see that $(\lambda +\hat{\rho}
|nK +\alpha^{\vee}) \in \ZZ$ iff $n \in u\ZZ$.  Therefore,
$\hat{\Delta}^\vee_\lambda = \hat{\Delta}^{\vee ,\re}_\lambda \cup
\hat{\Delta}^{\vee ,\im}_{\lambda}$, where
\begin{displaymath}
  \hat{\Delta}^{\vee ,\re}_\lambda = \{ nuK +\alpha^{\vee}|n\in \ZZ\, ,
  \,  \alpha^\vee \in \Delta^\vee \} \cap \hat{\Delta}^{\vee ,\re} \, ,
  \,  \hat{\Delta}^{\vee ,\im}_\lambda = \{ nu K | n \in \ZZ \backslash
  \{ 0 \} \} \, .
\end{displaymath}
Thus, the $\lambda$ in question is an admissible weight iff
(\ref{eq:1.1a}) holds, i.e.,
\begin{equation}
  \label{eq:1.2}
  (\lambda + \hat{\rho} |\alpha^{\vee}) \in \NN \hbox{  for all  }
  \alpha \in \hat{\Delta}^\vee_{\lambda, +} \, .
\end{equation}

\noindent{Case (i):}~~$(u,\ell) =1$.  In this case 
  \begin{displaymath}
   \hat{\Delta}_{\lambda}^{\vee ,\re} = \{ nuK +\alpha^\vee | n\in \ZZ\, , \, 
       \alpha \in \Delta_{\hbox{long}}\} \cup \{ n\ell uK 
       +\alpha^\vee |n\in \ZZ \, , \, 
          \alpha \in \Delta_{\hbox{short}}\} \, , 
  \end{displaymath}
and the set of simple roots is $\hat{\prod}^\vee_{\lambda} 
= \{ uK -\theta^\vee \, , \, 
     \alpha^\vee_1 ,\ldots ,\alpha^\vee_r \}$.

Hence condition (\ref{eq:1.1a}) is equivalent to
%
%\begin{displaymath}
  $(\lambda + \hat{\rho} |uK-\theta^\vee) =u (k+h^\vee)-
     (h^\vee -1)\in \NN$, 
%\end{displaymath}
%
i.e., $p-h^\vee \in \ZZ_+$.

\vspace{2ex}
\noindent{Case (ii):}~~$(u,\ell)=\ell$.  In this case
  \begin{displaymath}
    \hat{\Delta}^{\vee ,\re}_\lambda = \{ nuK +\alpha^\vee |
       n \in \ZZ \, , \, \alpha^\vee \in \Delta^\vee \}\, , 
\end{displaymath}
hence $\hat{\prod}^\vee_\lambda = \{ uK-\theta^\vee_s$, $\alpha^\vee_1
,\ldots ,\alpha^\vee_r \}$.  Therefore, using that $(\rho
|\theta^\vee_s) = (\rho^\vee |\theta) =h-1$,
condition~(\ref{eq:1.2}) is equivalent to $(\lambda +\hat{\rho}
|uK-\theta^\vee_s) =u (k+h^\vee) - (h-1) \in \NN$, i.e., $p-h \in
\ZZ_+$.

\end{proof}

\subsection{Principal admissible weights.}~~
\label{sec:1.2}
Recall the definition of the affine and extended affine Weyl
groups.  Let $W \subset \End \fh$ be the Weyl group of $\fg$ and
extend it to $\hat{\fh}$ by letting $w(K) = K$, $w (D) =D$ for
all $w \in W$.  Let $Q^\vee = \sum^r_{i=1} \ZZ \alpha^\vee_i$ be
the coroot lattice of $\fg$, then $P = \{ \lambda \in \fh^* |
(\lambda |\alpha) \in \ZZ$ for all $\alpha \in Q^\vee \}$ is the
weight lattice; note that $Q^\vee \subset P \cap Q^*$, where
$Q^*= \{ \lambda \in \fh |(\lambda |\alpha) \in \ZZ$ for all $\alpha \in 
Q\}$.  Given
$\alpha \in \fh$, define the translation \cite{K1}
\begin{displaymath}
  t_{\alpha} (v) = v+ (v|K) \alpha - (\frac12 |\alpha |^2
     (v|K)+(v|\alpha))K \, , \quad v \in \hat{\fh}\, ,
\end{displaymath}
and for a subset $L \subset \fh$, let $t_L=\{ t_{\alpha} | \alpha \in
L\}$.  The affine Weyl group $\hat{W}$ and the extended affine
Weyl group $\tilde{W}$ are defined as follows:
\begin{displaymath}
  \hat{W} = W \ltimes t_{Q^\vee}\, , \, \tilde{W} 
     = W \ltimes  t_{Q^*} \, , 
\end{displaymath}
so that $\hat{W} \subset \tilde{W}$.  Recall that the group
$\tilde{W}_+ = \{ w \in \tilde{W} | w (\hat{\prod}^\vee ) =
\hat{\prod}^\vee \}$ acts transitively on orbits of $\Aut 
\hat{\prod}^\vee$ and simply transitively on the orbit of $\alpha^\vee_0$,
and that $\tilde{W}=\tilde{W}_+\ltimes \hat{W}$.

An admissible weight $\lambda$ is called \emph{principal} if
$\hat{\Delta}^{\vee}_{\lambda}$ is isomorphic to
$\hat{\Delta}^{\vee}$.  We describe below the set of all
principal admissible weights.

For $u \in \NN$ let
\begin{displaymath}
  \hat{S}_{(u)} = \{ uK -\theta^\vee \, , \quad \alpha^\vee_1 ,
     \ldots , \alpha^\vee_r \} \, .
\end{displaymath}
Given $y \in \tilde{W}$, denote by $\hat{P}_{u,y}$ the set of all
admissible weights $\lambda$, such that
$\hat{\prod}^{\vee}_{\lambda} = 
y (\hat{S}_{(u)})$.

Let $\hat{P}$ (resp. $\hat{P}_+$) $= \{ \lambda \in \hat{\fh} |
(\lambda |\alpha^{\vee}_i) \in \ZZ$ (resp. $\in \ZZ_+$) for all
$i=0,\ldots ,r\}$
%, and $ (\lambda|D)=0\}$ 
denote the sets of all integral (resp. dominant
integral) weights.  Given $k \in \CC$, denote by $\hat{P}^k_+$,
$\hat{P}^k_{u,y}$, etc. the subsets of $\hat{P}_+$,
$\hat{P}_{u,y}$, etc., consisting of all elements of level $k$
(recall that the level of $\lambda$ is $(\lambda |K)$).

\begin{theorem}[\cite{KW1}]
  \label{th:1.1}
  \begin{enumerate}\begin{enumerate}
  \item %%a
$\hat{P}^k_{u,y} \neq \emptyset$ iff the triple $k,u,y$ satisfies
the following three properties:
%
%\begin{subequations}
\begin{list}{}{}
%%%\begin{equation}\label{eq:1.3}%a}
\item (i)~~$ (u,\ell)=1$ (recall that  $\ell =1,2$ or $3$ is the
    lacety of $\fg$),
%%\end{equation}
%
%%\begin{equation}\label{eq:1.3b}
\item (ii)~~$ k+h^{\vee} = \frac{p}{u} $,   where   $ p,u \in \NN \, , \, 
   p \ge h^\vee $   and    $ (p,u) =1 $ ,
%\end{equation}
%
%%\begin{equation}\label{eq:1.3c}
\item  (iii)~~$ y (\hat{S}_{(u)}) \subset \hat{\Delta}^\vee_+ $ .
%\end{equation}
%
%  \end{subequations}
\end{list}

\item %%b
If properties (i)---(iii) hold, then
\begin{displaymath}
  \hat{P}^k_{u,y} = \{ y . (\Lambda + (k+h^\vee -p)D) |
     \Lambda \in \hat{P}^{p-h^\vee}_+ \}\, ,
\end{displaymath}
where $y . \lambda =y (\lambda + \hat{\rho})-\hat{\rho}$ is the usual
shifted action.

\item %%c
The set of all principal admissible weights is $\cup_{k,u,y}
\hat{P}^k_{u,y}$, where $(k,u,y)$ runs over all triples
satisfying (i)---(iii).

\item %%d
Two non-empty sets $\hat{P}^k_{u,y}$ and
$\hat{P}^{k^{\prime}}_{u',y'}$ have a non-empty intersection iff
they coincide, which happens iff $k=k'$, $u=u'$ and $y (\hat{S}_{(u)})
=y' (\hat{S}_{(u)})$.

\end{enumerate}\end{enumerate}

\end{theorem}

Denote by $Pr^k$ the set of all principal admissible weights of
level $k$.  Note that $Pr^k=\hat{P}^k_+ + \CC K$ if $k \in \ZZ_+$.

\begin{proposition}[\cite{FKW}]
  \label{prop:1.2}
Let  $\Lambda \in \hat{P}^k_{u,y}$, where $y=t_{\beta}\bar{y}$,
be a principal admissible weight.  Then the following conditions
are equivalent:

\begin{list}{}{}
\item (i)~~$(\lambda |\alpha) \notin \ZZ$ for all $\alpha \in \Delta^\vee$,

\item (ii)~~$y (\hat{S}_{(u)}) \subset \hat{\Delta}^\vee_+
  \backslash \Delta^\vee_+ $,

\item (iii)~~$0<-(\bar{y}^{-1} (\beta ) |\alpha) < u$ for all
  $\alpha \in \Delta_+$,

\item (iv)~~$(\beta |\alpha) \notin u\ZZ$ for all $\alpha \in \Delta$.

\end{list}

\end{proposition}

A principal admissible weight, satisfying one of the equivalent
properties (i)---(iv) of Proposition~\ref{prop:1.2}, is called
\emph{non-degenerate}.  Given $\bar{y} \in W$, denote by
$\hat{P}^k_{\bar{y}}$ the union of all $\hat{P}^k_{u,y}$ with
$y=t_\beta \bar{y}$, where $\beta$ satisfies the inequalities
(iii) of Proposition~\ref{prop:1.2}.  It follows from this
proposition that the set of all non-degenerate principal admissible
weights of level~$k$, denoted by $Prn^k$, is the union of all
$\hat{P}^k_{\bar{y}}$:
\begin{displaymath}
  Prn^k = \bigcup_{\bar{y} \in W} \hat{P}^k_{\bar{y}} \, .
\end{displaymath}

\begin{proposition}[\cite{KW2},\cite{FKW}]
  \label{prop:1.3}
  \begin{enumerate}\begin{enumerate}
  
    \item %%a
      $\hat{P}^k_{\bar{y}} \neq \emptyset$ iff $k$ satisfies
      conditions (\ref{eq:0.2}).
%
%\begin{equation}
%        \label{eq:1.3}
%      k+h^\vee = \frac{p}{u}, \hbox{   where   } p,u =\NN \, ,\, 
%        (p,u) =1 \, , \, (\ell ,u)=1 \, , \, p\geq h^\vee \, , \,
%        u\geq ?? \, .
%\end{equation}

\item %%b
  Provided that (\ref{eq:0.2}) holds, two sets
  $\hat{P}^k_{\bar{y}}$ and   $\hat{P}^k_{\bar{y}'}$ have a
  non-empty intersection iff they coincide, which happens iff
  $\bar{y}^{-1} \bar{y}' \in W_+$, where $W_+$ is the image of $
 \tilde{W}_+$ under the canonical map $\tilde{W} \to W$.

\item %%c
  For $\Lambda \in \hat{P}^k_{\bar{y}}$ there exists a unique
  $\beta$, such that $\Lambda = (t_\beta \bar{y})$.  $(\Lambda^0 +
  (k+h^\vee -p) D)$, where $\Lambda^0 \in
  \hat{P}^{p-h^{\vee}}_+$, $p=u(k+h^\vee)$.  We let
  \begin{displaymath}
    \varphi_{\bar{y}} (\Lambda ) = (\Lambda^0 \, , \, uD -
    \bar{y}^{-1} (\beta) - \hat{\rho}^\vee)\, .
  \end{displaymath}

This is a bijective map
\begin{displaymath}
  \varphi_{\bar{y}} : \hat{P}^k_{\bar{y}} \to \hat{P}^{p-h^{\vee}}_+
     \times \hat{P}^{\vee u-h}_+ \, , 
\end{displaymath}
the inverse map being
\begin{displaymath}
  \psi_{\bar{y}} (\lambda ,\mu) = \bar{y} . \lambda + (k+h^\vee)
     (D-\bar{y} (\mu + \hat{\rho}^\vee))\, .
\end{displaymath}

\item %%d
Given $p,u \in \NN$, such that $(p,u)=1$, $(\ell ,u)=1$, $p \geq
h^\vee$, $u \geq h$, define the set
\begin{displaymath}
  I_{p,u} = (\hat{P}^{p-h^{\vee}}_+ \times \hat{P}^{\vee u-h}_+)
     / \tilde{W}_+
\end{displaymath}
(where $w (\lambda ,\mu) = (w(\lambda) , w (\mu)), w \in
\tilde{W}_+$).  Then, letting
\begin{displaymath}
  \varphi (\Lambda) = \varphi_{\bar{y}} (\Lambda) \hbox{   if   }
     \Lambda \in \hat{P}^k_{\bar{y}}\, ,
\end{displaymath}
gives a well-defined map $\varphi : Prn^k \to I_{p,u}$.

\item %%e
Elements $\Lambda ,\Lambda' \in Prn^k$ have the same image under
the map $\varphi$ iff $\Lambda' = \bar{w} . \Lambda$ for some
$\bar{w} \in W$.  In particular, all fibers of the map $\varphi$
have cardinality $|W|$.

  \end{enumerate} \end{enumerate}

\end{proposition}

\subsection{Characters of principal admissible
  $\hat{\fg}$-modules, their modular transformations and asymptotics.}~~
\label{sec:1.3}

Introduce the following domain $Y$ and coordinates $(\tau ,z,t)$
on $Y$:
\begin{displaymath}
  Y= \{ \lambda \in \hat{\fh} |\, {\rm Re} (\lambda |K)>0\}
 = \{ (\tau ,z,t):= 2\pi i (-\tau D + z+ tK)|\,
  \tau ,t \in \CC \, , \, \rm Im \tau >0 \, , \, z \in \fh \} \, .
\end{displaymath}
The product
\begin{displaymath}
  R_{\hat{\fg}} (v) = e^{(\hat{\rho}|v)} 
      \prod_{\alpha \in \hat{\Delta}^{\re}_+} (1-e^{-(\alpha |v)})
      \prod_{j \in \NN} (1-e^{-j(K|v)})^r
\end{displaymath}
converges to a holomorphic function on $Y$.  Given $\lambda \in
Y-\hat{\rho}$, consider the series
\begin{displaymath}
  N_{\lambda} (v) =\sum_{w \in \hat{W}_\lambda} \epsilon (w)
     e^{(w(\lambda +\hat{\rho})|v)}\, , 
\end{displaymath}
where $\hat{W}_\lambda$ is the subgroup of $\hat{W}$, generated
by reflections $r_{\alpha^\vee}$ (defined by $r_{\alpha^\vee} (v)
=v - (\alpha |v) \alpha^\vee$) in all real roots from
$\hat{\Delta}^\vee_{\lambda}$.  This series converges to a
holomorphic function on $Y$.

Let $L(\lambda)$ denote the irreducible highest weight
$\hat{\fg}$-module with the highest weight $\lambda \in
\hat{\fh}^*$, and let
\begin{displaymath}
  \ch_{L(\lambda)} (v) = \tr_{L(\lambda)} e^{\it v} \, , {\it v} \in Y \,, 
\end{displaymath}
be its character. For a function $F$ on the domain $Y$ we shall often write
$F(\tau,z,t)$ in place of $F(2\pi i(-\tau D+z+tK))$.

\begin{theorem}[\cite{KW1}]
  \label{th:1.2}
  \begin{enumerate}\begin{enumerate}
  \item %%a
If $\lambda$ is an admissible weight, then
$R_{\hat{\fg}}\ch_{L(\lambda)}$ converges in the domain $Y$ 
to the holomorphic function $N_{\lambda}$. 
%in the domain
%$Y_+=\{ v |\, \Re (\alpha_i |v) >0 \, , \quad i=0, \ldots ,r \}$, which
%extends to the following meromorphic function on $Y$ with at most
%simple poles at the hyperplanes $\alpha =0$, $\alpha \in
%\hat{\Delta}^{\re}$:
%
%\begin{displaymath}
%  \ch_{\lambda} (v) = N_{\lambda} (v) /R_{\hat{\fg}} (v) \, .
%\end{displaymath}

\item %%b
Let $\lambda = y . (\lambda^0 + (k+h^\vee -p)D) \in \hat{P}^k_{u,y}$,
where $\lambda^0 \in \hat{P}^{p-h^\vee}_+$.  Define the following
isometry $\phi_{u,y}$ of $\hat{\fh}$:
\begin{eqnarray*}
  \phi_{u,y} (\alpha^\vee_i) = y^{-1} (\alpha^\vee_i) , 
     \quad i=1,\ldots ,r \, , \,
  \phi_{u,y} (D) = uy^{-1} (D) \, , \, 
     \phi_{u,y} (K) = u^{-1}K \, .
\end{eqnarray*}
\end{enumerate}  \end{enumerate}

\end{theorem}
Then
\begin{equation}
  \label{eq:1.4}
  \ch_{L(\lambda)} (v) 
     = N_{\lambda^0}  (\phi_{u,y}(v))/R_{\hat{\fg}}(v)\, .
\end{equation}

\begin{remark}
  $\hat{W}_\lambda = \phi^{-1}_{u,y} \hat{W} \phi_{u,y}$.  Using
  this, one derives (\ref{eq:1.4}) from (a).
\end{remark}

Define the normalized character
\begin{equation}
\label{eq:1.5}
  \chi_{L(\lambda)}(\tau, z, t) = e^{2\pi i \tau s_\lambda} \ch_{L(\lambda)}
(\tau, z, t)\, ,
\end{equation}
where 
\begin{displaymath}
s_{\lambda} = 
\frac{|\lambda +\hat{\rho}|^2}{2(k+h^\vee)} -
\frac{|\hat{\rho}|^2}{2h^\vee},
\end{displaymath}
and $k$ is
the level of $\lambda$.
% and $\bar{\lambda}$ is the projection of
%$\lambda$ to $\fh \subset \hat{\fh}$.  
Let $D_{\hat{\fg}}(\tau
,z,t)=q^{|\hat{\rho} |^2/2h^\vee} R_{\hat{\fg}}(\tau ,z,t)$ denote the 
denominator
of $\chi_{L(\lambda)} (\tau ,z,t)$.  It satisfies the following
well-known transformation formula \cite{K1}:
\begin{equation}
  \label{eq:1.6}
  D_{\hat{\fg}} \left( -\frac{1}{\tau} \, , \, \frac{z}{\tau}
    \, , \, 
      t- \frac{(z|z)}{2\tau} \right) = (-i)^{|\Delta_+ |}
       (-i\tau)^{r/2} D_{\hat{\fg}} (\tau ,z, t) \, .
\end{equation}

Note that $\chi_{L(\lambda)}=\chi_{L(\lambda +aK)}$ for any $a \in \CC$.
For that reason, from now on we shall consider elements of $Pr^k$
$\mod \, \CC K$. There is a unique representative $\lambda$ in each coset,
such that $\lambda (D)=0$. When writing $(\lambda|\mu)$
for $\lambda, \mu \in Pr^k$, we shall mean the scalar product of these 
representatives, in other words, 
$(\lambda | \mu)=(\bar{\lambda} |\bar{\mu})$, where the bar means
the orthogonal projection of $\hat{\fh}$ to $\fh$. Note that, by the ``strange 
formula''\cite{K1}, we have:
\begin{displaymath}
 s_{\lambda}=\frac{(\bar{\lambda} +2\rho |\bar{\lambda})}
{2(k+h^\vee)}- \frac{1}{24} \frac{k\dim \fg}{k+h^\vee}. 
\end{displaymath}

>From (\ref{eq:1.4}) and
the transformation formula and asymptotics for characters of
integrable $\hat{\fg}$-modules we derived those for principal
admissible modules.

\begin{theorem}[\cite{KW1}]
  \label{th:1.3}
Let $\lambda =y. (\lambda^0 + (k+h^\vee -p)D) \in Pr^k$, where
$k+h^\vee =p/u$, $y=t_{\beta}\bar{y}$, $\beta \in Q^*$,
$\bar{y} \in W$, $\lambda^0  \in \hat{P}^{p-h^\vee}_+$.  Then

\begin{enumerate} \item[] \begin{enumerate}
\item %%a
%\begin{displaymath}
\hspace*{1in}$ \chi_{L(\lambda)} \left( -\frac{1}{\tau} \, , \, \frac{z}{\tau}    \, , \, t-\frac{(z|z)}{2\tau}\right)   = \sum_{\lambda' \in
    Pr^k} a(\lambda ,\lambda') \chi_{L(\lambda')} (\tau ,z,t)$,\\
%\end{displaymath}

where
\begin{eqnarray*}
  a (\lambda ,\lambda') &=& i^{|\Delta_+|} u^{-r} (k+h^\vee)^{-r/2}
     |P/Q^\vee | \epsilon (\bar{y} \bar{y}')
     e^{-2\pi i ((\lambda^0+\rho |\beta')
    +(\lambda^{\prime 0}+\rho |\beta) 
    +  (k+h^\vee)(\beta|\beta'))}\\
     && \times \sum_{w \in W} \epsilon (w) 
        e^{-\frac{2\pi i}{k+h^\vee}(w(\lambda^0 +\hat{\rho})|
          \lambda^{\prime 0}+\hat{\rho})} \, .
\end{eqnarray*}

\item %%b
As $\tau \downarrow 0$, we have for each $z \in \fh$, such that
$(\alpha |z) \notin \ZZ$ for all $\alpha \in \Delta$:
\begin{displaymath}
  \chi_{L(\lambda)}(\tau,-\tau z,0)\, \sim b(\lambda,z)e^{\frac{\pi ig^{(k)}}
{12\tau}}\, ,
\end{displaymath}
where
\begin{eqnarray*}
  b (\lambda ,z) &=& \epsilon (\bar{y}) u^{-r/2}
     a(\lambda^0)\prod_{\alpha \in \Delta_+} \sin
     \frac{\pi (z-\beta |\alpha)}{u}/\sin \pi (z|\alpha)\, , \\
     a(\lambda^0) &=& p^{-r/2} 
        |P/Q^\vee|^{-1/2} \prod_{\alpha \in \Delta_+}
        2\sin \frac{\pi (\lambda^0 +\rho |\alpha)}{p}\, , \\
        g^{(k)} &=& \left( 1-\frac{h^\vee}{pu}\right)\dim \fg \, .
\end{eqnarray*}

    \end{enumerate} \end{enumerate}

\end{theorem}

\section{Characters of $W$-algebras $W_k (\fg ,f)$.}  
\label{sec:2}

\subsection{Construction of simple vertex algebras  $W_k (\fg ,f)$.}
\label{sec:2.1}~~~First, recall the  definition of the vertex algebra  
$W_k (\fg
,f)$ \cite{KRW}, \cite{KW3}.  Let $\fg$ be a finite-dimensional
simple Lie algebra with the normalized invariant bilinear form as
in Sec.~\ref{sec:1.1}, and let $f$ be a nilpotent element of $\fg$.
Include $f$ in an $\sl_2$-triple $f,x,e$, so that $[x,e] =e$,
$[x,f]=-f$ and $[e,f]=x$.  We have the eigenspace decomposition
of $\fg$ with respect to $\ad x$:
\begin{displaymath}
  \fg = \bigoplus_{j \in \frac12 \ZZ} \fg_j \, .
\end{displaymath}
Let $\fg_{\pm} = \oplus_{j>0} \fg_{\pm j}$, and let $A_{\ch} =
\fg_+ + \fg_-$ with odd parity.  The restriction of the bilinear
form $(\, . \, | \, . \, )$ to $A_{\ch}$ is non-degenerate and
skew-supersymmetric (since it is symmetric before the change of
parity).  Denote by $A_{\ne}$ the (even) vector space $\fg_{1/2}$
with the non-degenerate skew-symmetric bilinear form
\begin{displaymath}
  \langle a,b \rangle = (f| [a,b]) \,, \quad a,b \in \fg_{1/2}\, .
\end{displaymath}

Introduce the following differential complex $(\C^k (\fg ,f), \,
d_0)$, where
\begin{displaymath}
  \C^k (\fg ,f) = V^k (\fg ) \otimes F (\fg ,f)
\end{displaymath}
is the tensor product of the universal affine vertex algebra of
level~$k$ and $F (\fg ,f) = F (A_{\ch}) \otimes F (A_{\ne})$ is
the tensor product of vertex algebras of free superfermions,
based on $A_{\ch}$ and $A_{\ne}$, and $d_0$ is an odd derivation
of the vertex algebra $\C^k (\fg ,f)$, such that $d^2_0 =0$.

In order to define the differential $d_0$, choose a basis $\{
u_{\alpha} \}_{\alpha \in S_j}$ of each subspace $\fg_j$ of
$\fg$, and let $S_+ = \cup_{j>0} S_j \, , \quad S=\cup_j S_j$.
Then $ \{ u_{\alpha} \}_{\alpha \in S_+}$ is a basis of $\fg_+$,
and we denote by $\{ u^{\alpha} \}_{\alpha \in S_+}$ the dual
basis of $\fg_-$.  Let $\{ \varphi_{\alpha }, \varphi^{\alpha}
\}_{\alpha \in S_+}$ be the corresponding basis of $A_{\ch}$ and
$\{ \phi_{\alpha} \}_{\alpha \in S_{1/2}}$ the corresponding
basis of $A_{\ne}$.  Consider the following odd field of the
vertex algebra $\C^k (\fg ,f)$:
\begin{eqnarray*}
  d(z) &=& d^{\rm st} (z) + \sum_{\alpha \in S_+} (f|u_{\alpha})
     \varphi^{\alpha} (z) + \sum_{\alpha \in S_{1/2}} :
        \varphi^{\alpha}(z) \phi_{\alpha}(z): \, ,\\
\noalign{\hbox{where}}\\
d^{\rm st}(z) &=& \sum_{\alpha \in S_+}:u_{\alpha} (z)
   \varphi^{\alpha} (z):
  - \frac12 \sum_{\alpha ,\beta ,\gamma \in S_+}
    c^{\gamma}_{\alpha\beta}:\varphi_{\gamma}(z) \varphi^{\alpha}
      (z) \varphi^{\beta}(z):
\end{eqnarray*}
and $c^{\gamma}_{\alpha\beta}$ are the structure constants of
$\fg: [u_{\alpha} , u_{\beta}] = \sum_{\gamma}
c^{\gamma}_{\alpha\beta} u_{\gamma} \,(\alpha ,\beta, \gamma \in
S)$.  Then $d_0 : = \Res_z \, d (z)$ is an odd derivation of all
products of the vertex algebra $\C^k (\fg ,f)$ and $d^2_0 =0$.

The homology of the complex $(\C^k (\fg ,f), \, d_0)$ is a vertex
algebra, denoted by $W^k (\fg ,f)$.

We shall assume that $k+h^\vee \neq 0$.  Then one can construct
the following Virasoro field $L (z)$ of $\C^k (\fg ,f)$, making
the latter a conformal vertex algebra:

\begin{displaymath}
  L(z) = L^{\fg} (z) + \partial_z x(z) + L^{\ch} (z)+L^{\ne}(z)\, .
\end{displaymath}
Here $L^{\fg} (z) = \frac{1}{2(k+h^\vee)} \sum_{\alpha \in S}
:u_{\alpha} (z) u^{\alpha} (z):$ is the Sugawara construction,
and
\begin{eqnarray*}
  L^{\ch} (z) &=& \sum_{\alpha \in S_+} (\alpha(x) :\partial_z
    \varphi_{\alpha} (z) \varphi^{\alpha}(z): + (1-\alpha(x)):\partial_z 
      \varphi^{\alpha} (z) \varphi_{\alpha}(z) :) \, , \\
  L^{\ne} (z) &=& \frac12 \sum_{\alpha \in S_{1/2}} :
     \partial_z \phi^{\alpha}(z) \phi_{\alpha} (z) :\, ,
\end{eqnarray*}
where 
%$m_{\alpha}=j$ if $\alpha \in S_j$ and 
$\{ \phi^{\alpha}
\}$ is a basis of $A_{\ne}$, dual to $\{ \phi_{\alpha}\}$ with
respect to the bilinear form $\langle \, . \, , \, . \,
\rangle$.  

With respect to $L(z)$ the field $d(z)$ is primary of conformal weight~$1$.
In the language of $\lambda$-brackets, this means that $[d_{\lambda}L]=\lambda
d$, hence $d_0 (L) = [d_{\lambda} L] |_{\lambda =0}=0$.  Thus,
the homology class of $L$ defines the Virasoro field
$L(z)=\sum_{n \in \ZZ} L_n z^{-n-2}$ of the vertex algebra $W^k
(\fg ,f)$, making it a conformal vertex algebra.  
Its central charge is given by the following formula
\cite{KW1}:
\begin{equation}
  \label{eq:2.1}
  c(\fg ,f,k) = \dim \fg_0 -\frac12 \dim \fg_{1/2}
     - \frac{12}{k+h^\vee} |\rho - (k+h^\vee) x |^2 \, .
\end{equation}

Above and further on we identify the fields of a vertex algebra,
like $d(z)\, ,\,  L(z)\, , \ldots ,$ with the corresponding states
$d\, , \, L\, , \ldots,$ via the field-state
correspondence: $a=a(z) \vac |_{z=0}$.

Let $\fg^f$ be the centralizer of $f$ in $\fg$.  Then the
$\frac12 \ZZ$-grading, induced on $\fg^f$ from $\fg$ is of the form
\begin{displaymath}
  \fg^f = \sum_{j\in \frac12 \ZZ_+} \fg^f_{-j} \, .
\end{displaymath}
It was proved in \cite{KW3} that the vertex algebra $W^k (\fg,f)$
is strongly generated by $\dim \fg^f$ fields $J^{\{ a\}} (z)$,
where $a$ runs over a basis of $\fg^f$, compatible with the above
$\frac12 \ZZ$-grading, and $J^{\{ a \}} (z)$ has conformal weight
$1+j$ if $a \in \fg_{-j}$.

Therefore the operator $L_0$ is diagonalizable on $W^k (\fg ,f)$
with spectrum in $\frac12 \ZZ_+$ and finite-dimensional
eigenspaces, the $0$\st{th} eigenspace being $\CC\vac$.  Hence
the vertex algebra $W^k (\fg ,f)$ has a unique maximal ideal $I$,
and $W_k (\fg ,f) = W^k (\fg ,f)/I$ is a simple vertex algebra.

%\begin{remark}
%  \label{rem:2.1}  Since the eigenvalues of $L_0$ on $I$ are
%  strictly positive, any irreducible $W^k (\fg ,f)$-module $M$ on
%  which $L_0$ is diagonalizable with spectrum bounded below, is
%  automatically a $W_k (\fg ,f)$-module, since $IM$, being a
%  proper submodule of $M$, should be $0$.
%\end{remark}

Given $v \in \fg_0$, the subspaces $\fg_+$ (resp. $\fg_{1/2}$)
are $\ad v$-invariant; let $c^{\alpha}_{\beta} (v)$, $\alpha
,\beta \in S_{+ (\hbox{  resp.  }1/2)}$ be the matrix of $\ad v$ on
$\fg_{+ (\hbox{  resp.  }1/2)}$  in the basis $\{ u_{\alpha}
\}_{\alpha \in S_{+ (\hbox{  resp.  }1/2)}}$, i.e., $[v,u_\beta]
=\sum_\alpha c^\alpha_\beta (v) u_\alpha$.  Define the following field:
\begin{eqnarray*}
  J^{\{ v \}} (z) = v(z) + v^{\ch}(z)+v^{\ne}(z)\, ,
\end{eqnarray*}
where
\begin{eqnarray*}
v^{\ch}(z)=\sum_{\alpha ,\beta \in S_+}
  c^\alpha_\beta (v) : \varphi_{\alpha} (z) \varphi^\beta (z):\, , \,\,
    v^{\ne}(z)= -\frac12 \sum_{\alpha ,\beta \in S_{1/2}} c^\alpha_\beta
      (v) :\phi_\alpha (z) \phi^\beta (z): \, .
\end{eqnarray*}

\begin{lemma}
  \label{lem:2.1}
  \begin{enumerate}\begin{enumerate}
  
\item %a
For any $v \in \fg_0$, we have
\begin{displaymath}
  d^{\rm st}_0 (J^{\{ v \}}) =0 \, .
\end{displaymath}

\item %%b
Provided that $v \in \fg^f_0$, we have:
\begin{displaymath}
  d_0 (J^{\{ v \}}) =0 \, .
\end{displaymath}

  \end{enumerate} \end{enumerate}
\end{lemma}

\begin{proof}
(b) was proved in \cite{KRW}, by making use of the
$\lambda$-bracket calculus.  The same proof gives~(a).

\end{proof}

\subsection{Computation of the Euler-Poincar\'e characters in the Ramond 
sector.}
\label{sec:2.2}

We studied in \cite{KW4} the most general $\sigma$-twisted
modules over $W^k (\fg ,f)$.  From the point of view of modular
invariance the most important case is the Ramond twisted one,
when $\sigma = \sigma_R$ is an automorphism of $\fg$, defined by
$\sigma_R = e^{2\pi i \ad x}$.

Since $\sigma_R$ is an inner automorphism of $\fg$, we can choose
a Cartan subalgebra $\fh$ of  $\fg$, contained in $\fg_0$.  Let
$\Delta \subset \fh^*$ be the set of roots, and let $\Delta^j
\subset \Delta$ be the set of roots of $\fh$ in $\fg_j$.  

A natural choice of a subset of positive roots in $\Delta$
is $\Delta_+ = \Delta^0_+\cup(\cup_{j>0}\Delta^j)$, where
$\Delta^0_+$ is a subset of positive roots of the root system $\Delta^0$.

However, we shall need another choice of a subset of positive roots,
which will be denoted by $\Delta^{\new}_+$.
First, we choose the element  $f$ 
%is a nilpotent element of principal type, we may assume that 
to be a sum of root vectors, attached to roots $\gamma_1 ,\ldots
,\gamma_s \in \Delta^{-1}$, such that $\fh^f := \{ h \in \fh
|\, \gamma_i (h) =0\, , \quad i=1,\ldots ,s \}$ is a Cartan
subalgebra of $\fg^f$.  
Next, choose an element $h_0 \in \fh^f$ such that
$\alpha (h_0)= 0$ iff $\alpha |_{\fh^f}=0$ for $\alpha \in \Delta$.
Since the restriction of any root from $\Delta^{1/2}$ to $\fh^f$ is not 
identically zero \cite{EK}, it follows that $\ad h_0 |_{\fg_{1/2}}$ has no 
zero eigenvalues. Finally, let 
$\Delta^0_0 
= \{\alpha \in \Delta^0|\,\, \alpha|_{\fh^f} = 0\}$ and let
$\Delta^0_{0,+}$ be a subset of positive roots in the root system
$\Delta^0_0 $.    
%, and choose a subset of positive roots $\Delta^0_+$ of $\Delta^0$, such that
%and $\alpha (h_0) \ge 0$ if $\alpha \in \Delta^0_+$.  
The new set of positive roots in $\Delta$ is as follows:
\begin{displaymath}
  \Delta^{\new}_+ = 
\Delta^0_{0,+} \cup 
\{ \alpha \in \Delta | \hbox{
    either  } \alpha (h_0)>0\, , \hbox{  or  }
    \alpha (h_0) =0 \hbox{  and  }\alpha (x) <0 \}\, .
\end{displaymath}

Since $h_0 \in \fh^f$, \,$\gamma_i(h_0)=0$ for all $i$, hence all 
$\gamma_i$ lie in $\Delta^{\new}_+$.
Also, since all elements from $\fg_0^f$ define symplectic endomorphisms
of $\fg_{1/2}$, \, 
$\Delta_+^{1/2}:=\Delta^{\new}_+ \cap \Delta^{1/2}$ 
(resp. $\Delta_-^{1/2}:=(-\Delta^{\new}_+) \cap \Delta^{1/2})$ 
contains exactly
half of the roots from $\Delta^{1/2}$. Thus $\Delta^{\new}_+$
satisfies the properties, required in \cite{KW3}.

Let $\displaystyle{\hat{\fg}^R = \oplus_{m \in \frac12 \ZZ} (\fg_{\bar{m}}
\otimes t^m) \oplus \CC K \oplus \CC D}$ be the $\sigma_R$-twisted
affine Lie algebra \cite{K1}, where $\displaystyle{\fg_{\bar{m}} =\oplus_{j \in
\bar{m}} \fg_j}$, $\bar{m} \in \frac12 \ZZ / \ZZ$ being the coset of $m
\in \frac12 \ZZ$, and the commutation relations are given by the
same formulas as for $\hat{\fg}$.  Note that both $\hat{\fg}^R$
and $\hat{\fg}$ lie in the affine Lie algebra $\fg [t^{1/2},
t^{-1/2}] \oplus \CC K \oplus \CC D$, and, in fact, they are
isomorphic, $\hat{\fg}$ being mapped to $\hat{\fg}^R$ by the
element $t^x$ of the group $G(\CC [t^{1/2},t^{-1/2}])$.  
Note also that they share the same Cartan
subalgebra~$\hat{\fh}$, and that $\hat{\Delta}^R = t_x (\hat{\Delta})$
is the set of roots of $\hat{\fg}^R$.

For each $\alpha \in \Delta$ define the number
\begin{displaymath}
  s_\alpha = - \alpha (x) \hbox{  if  }\alpha \in \Delta^{\new}_+
      \, , \, = 1-\alpha (x) \hbox{  if  } 
         \alpha \in -\Delta^{\new}_+ \, ,
\end{displaymath}
and introduce the following set of positive roots in the set
$\hat{\Delta}^R$ of roots of $\hat{\fg}^R$:
\begin{displaymath}
  \hat{\Delta}^R_+ = \{ (n +s_{\alpha}) K+\alpha |\, \alpha \in \Delta
     \, , \, n \in \ZZ_+ \} \cup \{ nK |\,n \in \NN \} \, .
\end{displaymath}

Let $\Delta^R$ (resp. $\Delta^R_+$) $= \{ -\alpha (x) K+\alpha |\, 
\alpha \in \Delta$ (resp. $\Delta_+^{\new}$)$\} \subset 
\hat{\Delta}^R$ (resp. $ \hat{\Delta}_+^R$). Let $\Delta^{R,j}_+
= \{\alpha \in \Delta^R_+| (\alpha|x)=j \}$.   
Let $\Delta^{R,f}$ (resp. $\Delta^{R,f}_+)=\{\alpha \in \Delta^R$ (resp. 
$\Delta^R_+)|\,\,\alpha|_{\fh^f}=0\}$.

\begin{lemma}
  \label{lem:2.2}
  \begin{enumerate}\begin{enumerate}

  \item %%a
There exists $\bar{w} \in W$, such that $\hat{\Delta}^R_+ = t_x \bar{w}
(\hat{\Delta}_+)$.

\item %%b
$\Delta^{\new}_+ = \bar{w} (\Delta_+)$.

\item %%c
If $g \in \fg^f$, then $(x |g) =0$.  Consequently, $\fh^f \subset t_x
\bar{w} (\fh )$.

\item %%d
Let $\prod^R$ be the set of simple roots for $\Delta^{R}_+$. Then
$\Delta^{R,f}_+ \cap \prod^R$ is the set of simple roots for
$\Delta^{R,f}_+$.
  \end{enumerate}\end{enumerate}

\end{lemma}

\begin{proof}
Let $\hat{\prod}^R = \{ \tilde{\beta}_i = s_{\beta_i} K +
\beta_i |\, i=0,1,\ldots ,r \}$, where $\beta_i \in \Delta$, be the
set of simple roots of $\hat{\Delta}^R_+$.  We have:  $K = \sum_i
a_i \tilde{\beta}_i$, where $a_i$ are positive integers.  Recall
also that $s_{\beta_i} = -\beta_i (x) +n_i$,  where $n_i =0$ or
$1$.

Since $(x |K) =0$,  it follows that:
\begin{displaymath}
  \sum_i a_i \beta_i =0 \hbox{   and   } \sum_i a_in_i =1 \, .
\end{displaymath}
It follows from the second equation that there exists index $i_0$
such that $a_{i_0} =1$ and $n_i = \delta_{i,i_0}$.  Hence $\{
\beta_i \}_{i \neq i_0}$ is a set of simple roots for $\Delta$
and $\tilde{\beta}_{i} =- \beta_i (x)K +\beta_i = t_x \beta_i$ for
all $i \neq i_0$.  Hence there exists $\bar{w} \in W$, such that $\{
\beta_i \}_{i \neq i_0} =\bar{w} (\prod)$, where $\prod$ is the set of
simple roots of $\Delta_+$.  Therefore $\{ \tilde{\beta}_i \}_{i \neq
  i_0} =t_x \bar{w} (\prod)$ and hence $\hat{\prod}^R = t_x \bar{w}
(\hat{\prod})$, where $\hat{\prod}$ is the set of simple roots of
$\hat{\Delta}_+$ proving~(a).

By the proof of (a), $\beta_i \in \Delta^{\new}_+$ if $i \neq i_0$.
Hence, by the construction of $\bar{w}$, (b) follows.

Since $x=[e,f]$, we have: $(x|g) = ([e,f] |g)
=(e|[f,g]) =0$ if $g \in \fg^f$, proving (c).

In order to prove (d), we need to show that if $\alpha 
\in \Delta^{R,f}_+$
and $\alpha=\beta + \gamma$, where $\beta, \gamma 
\in \Delta^{R}_+$, then $\beta , \gamma 
\in \Delta^{R,f}_+$. Evaluating at $h_0$, we have:
$0=\alpha(h_0)= \beta(h_0) + \gamma(h_0)$. But both 
summands on the right are non-negative, by definition of 
$\Delta_+^{\new}$ and $\Delta_+^{R}$. Hence $\beta(h_0)=\gamma(h_0)=0$
and the restriction of $\beta$ and $\gamma$ to $\fh^f$ is zero.
\end{proof}

We have: $\hat{\Delta}^R_+ = t_x \bar{w}(\hat{\Delta}_+)$, 
$\Delta^R = t_x(\Delta)$, $\Delta^R_+ = t_x \bar{w} (\Delta_+)$.
Let $\hat{\rho}^R =
t_x \bar{w} (\hat{\rho})$, $\fh^R = t_x (\fh)$, $D^R = t_x (D)$,
$W^R = t_x Wt^{-1}_x$, $\hat{P}^{k,R}_+ = t_x \bar{w} (\hat{P}^k_+)$,
$Pr^{k,R} = t_x \bar{w} Pr^k$, etc.  Then 
%$\hat{\Delta}^R \subset
%\hat{\fh}$ is the set of roots of $\hat{\fg}^R$,
%$\hat{\Delta}^R_+$ being the subset of positive roots introduced
%above, 
$\hat{W}^R = W^R \ltimes t_{Q^{\vee,R}}$ is the (affine) Weyl
group of $\hat{\fg}^R$, $\tilde{W}^R = W^R \ltimes t_{P^R}$ 
is the extended affine Weyl group, $\hat{\fh} = \fh^R +
\CC K + \CC D^R$ is the ``standard'' decomposition of the Cartan
subalgebra $\hat{\fh}$ of $\hat{\fg}^R$, etc.  
%It is also easy to
%see from the definition of $\Delta^R_+$ that $\Delta^{\rm new}_+ =
%\bar{w} (\Delta_+)$.

Consider the triangular decomposition $\hat{\fg}^R =\hat{\fn}^R_-
+ \hat{\fh} + \hat{\fn}^R_+$, corresponding to the set of
positive roots $\hat{\Delta}^R_+$, 
%of $\hat{\fg}^R$
and let $M$
be a highest weight $\hat{\fg}^R$-module with
highest weight vector $v_{\Lambda}$, annihilated by
$\hat{\fn}^R_+$, where $\Lambda \in \hat{\fh}^*$ has level~$k$.
By \cite{GK2}, $R_{\hat{\fg}^R} ch_M$ converges to a holomorphic
function in $Y_+$, which extends to a holomorphic 
function in the domain $Y$. 
%(this holds if $M=L(\Lambda)$ is an admissible module, or a Verma module).

The $\hat{\fg}^R$-module $M$ extends to a $\sigma_R$-twisted $V^k
(\fg)$-module by the map:
\begin{displaymath}
  e_{\alpha} \mapsto e^R_\alpha (z) : =
    \sum_{n \in (\alpha |x) +\ZZ} (e_{\alpha} t^n )z^{-n-1}\, , 
     \quad h \mapsto h^R (z) = \sum_{n \in \ZZ} (ht^n)z^{-n-1}\, , 
\end{displaymath}
where $e_{\alpha}$ is a root vector of $\fg$, attached to $\alpha
\in \Delta$, and $h \in \fh$.  Likewise we have the
$\sigma_R$-twisted $F (\fg ,f)$-module $F^R (\fg ,f)$, 
given by $\varphi_{\alpha}
\mapsto \varphi^R_{\alpha} (z) := \sum_{n \in \alpha (x)+\ZZ}(\varphi
_{\alpha} t^n) z^{-n-1}$, 
$\varphi^{\alpha} \mapsto \varphi^{\alpha,R} (z) := 
\sum_{n \in \alpha (x)+\ZZ}(\varphi^{\alpha} t^n) z^{-n-1}$, 
and $\phi_{\alpha} \mapsto \phi^R_{\alpha} (z) := 
\sum_{n \in 1/2 +\ZZ}(\phi_{\alpha} t^n) z^{-n-1}$, generated by
the vacuum vector $\vac_R$, subject to the conditions  
$\varphi_{\alpha}t^n\vac_R=\varphi^{-\alpha}t^{n-1}\vac_R=
\phi_{\alpha}t^n\vac_R
=0$ if $\alpha +nK \in \hat{\Delta}_+^R$    
\cite{KW4}.  
We thus obtain the
$\sigma_R$-twisted $\C^k (\fg ,f)$-module $\C^R (M) = M \otimes
F^R (\fg ,f)$.

Introduce the charge decomposition
\begin{displaymath}
  \C^R (M) = \bigoplus_{m \in \ZZ} \C^{R}_m (M)
\end{displaymath}
by letting charge $M=0$, charge~$\vac_R =0$, charge
~$\varphi^R_{\alpha} =1$, charge~$\varphi^{\alpha ,R}=-1$,
charge~$\phi^R_{\alpha}=0$.  The most important
$\sigma_R$-twisted fields of $C^R(M)$ are:
\begin{eqnarray*}
  d \mapsto d^R (z) = \sum_{n \in \ZZ} d^R_n z^{-n-1}\,, \quad
     L \mapsto L^R (z) = \sum_{n \in \ZZ} L^R_n z^{-n-2}\, ,
     \quad\\
     J^{\{ h \}} \mapsto J^{\{ h \}, R} (z) = \sum_{n \in \ZZ}
       J^{\{ h \} ,R}_n z^{-n-1} \qquad ( h \in \fh)\, .
\end{eqnarray*}
Since $(d^R_0 )^2 =0$, we obtain a complex $(\C^R (M) , d^R_0)$,
where $d^R_0 : \C^R_{m}(M) \to \C^R_{m-1} (M)$.  The homology of
this complex is canonically a $\sigma_R$-twisted $W^k (\fg
,f)$-module
\begin{displaymath}
  H^R (M) = \bigoplus_{m \in \ZZ} H^{R}_m (M) \, ,
\end{displaymath}
where all $H^{R}_m(M)$ are submodules (see \cite{KW4} for details).

Recall \cite{KRW} that with respect to $L_0$ the fields
$e_{\alpha}(z)$, $h(z)$, $\varphi_{\alpha}(z)$,  $\varphi^{\alpha}(z)$, 
$\phi_{\alpha} (z)$ have
conformal weights $1-\alpha (x)$, $1$, $1-\alpha (x)$, $\alpha
(x)$ and $\alpha (x) = 1/2$, respectively.  Writing the
corresponding twisted fields in the form $a(z) = \sum_n a_n
z^{-n-\Delta _{a}}$, where $\Delta_a$ is the conformal weight of
$a(z)$, we obtain:
\begin{eqnarray*}
  e^R_{\alpha} (z) &=& \sum_{n \in \ZZ} e_{\alpha ,n} 
     z^{-n-1+\alpha (x)} \, , \quad  h^R (z) = \sum_{n \in \ZZ} h_n z^{-n-1}
    \, , \quad     \varphi^R_{\alpha} (z) = \sum_{n \in \ZZ}
       \varphi_{\alpha ,n} z^{-n-1+\alpha (x)}\, ,\\[1ex]
  \varphi^{\alpha ,R} (z) &=& \sum_{n \in \ZZ}
     \varphi^{\alpha }_n z^{-n- \alpha (x)} \, , \quad
       \phi^R_{\alpha} (z) = \sum_{n \in \ZZ}\phi_{\alpha ,n}
         z^{-n-1/2} \, .
\end{eqnarray*}

\begin{lemma}
  \label{lem:2.3}
%Let $\vac_R$ be the vacuum vector of $F^R (\fg ,f)$.  The
%following are necessary and sufficient conditions for
%annihilation of the vector~
Let $v=v_{\Lambda} \otimes \vac_R \in M
\otimes F^R (\fg ,f)$. Then
%, provided that $M$ is a Verma
%$\hat{\fg}^R$-module:
%
\begin{eqnarray*}
  e_{\alpha ,n} v &=& 0 \hbox{  if  } \alpha \in \Delta^{\new}_+ \, ,
     \,   n \geq 0 \,\, (\hbox{resp.  } \alpha \in -\Delta^{\new}_+ \, ,
       \,   n>0)\, , \\[1ex]
 h_n v &=& \delta_{0,n} \Lambda(h)v \,\, \hbox {if} \,\, h \in \fh \, , n \geq 0
\, , \\[1ex]
 \varphi_{\alpha ,n} (\hbox{or  } \phi_{\alpha ,n} )v &=&
    0 \,\, \hbox{if  } \alpha \in \Delta^{\new}_+ \, , \, n \geq 0 
\,\, (\hbox{resp.  } \alpha \in -\Delta^{\new}_+ \, ,\,   n>0)\, , \\[1ex]
   \varphi^{\alpha}_n v &=& 0 \,\, \hbox{if } \alpha \in \Delta^{\new}_+\, , \, 
      n>0  \,\, (\hbox{resp.   } \alpha \in -\Delta^{\new}_+ \, , \, 
         n \geq 0)\, .
\end{eqnarray*}

\end{lemma}

\begin{proof}

It follows from the fact that $e_{\alpha ,n}v=0$, 
$\varphi_{\alpha ,n}v=0$ and $\phi_{\alpha ,n}v=0$ if $n- \alpha
(x) \geq s_{\alpha}$, and $\varphi^{\alpha}_n v=0$ if $n-
\alpha (x) \geq -s_{\alpha}$ by the construction of
$\hat{\Delta}^R_+$ and of $F^R (\fg ,f)$, see \cite{KW3}.
\end{proof}

As usual,  define the
Euler--Poincar\'e character of $H^R (M)$ by the following
formula:
\begin{displaymath}
  \ch_{H^R (M)} (\tau ,z,t) =e^{2 \pi i kt} \sum_{m \in \ZZ} (-1)^m 
     \tr_{H_m^{R}(M)} q^{L^R_0} e^{2\pi i J_0^{\{ z \}, R}} \, , 
\end{displaymath}
where $q=e^{2 \pi i \tau}$, $\tau$ lies in the complex upper half
plane $\CC^+$, $t \in \CC$ and $z \in \fh^f$.

By the Euler--Poincar\'e principle we have
\begin{displaymath}
  \ch_{H^R (M)} (\tau ,z,t) =
  \ch_{\C^R (M)} (\tau ,z,t) \hbox{  for  } \tau \in \CC^+ \, , \, 
  z \in \fh^f , \,  t \in \CC. 
\end{displaymath}
Unfortunately, the right-hand side converges only for generic $z
\in \fh$ and may diverge for all $z \in \fh^f$.

To get around  this difficulty, we use a trick similar to the one
employed in \cite{FKW}.

\begin{definition}
  \label{def:2.4}
An element $f$ of $\fg$ is called a nilpotent element of
\emph{principal type} if $f$ is a principal nilpotent element in
the centralizer of $\fh^f$ (equivalently, if $\Delta^0_0 =\emptyset$).

\end{definition}

Recall that in $\fg = \sl_{r+1}$ all nilpotent elements are of principal type,
but in general this is not the case due to existence of non-principal 
nilpotents $f$ with $\fh^f=0$.

%From now on we will assume that $f$ is of principal type. 
%In this case we can (and will) choose $h_0 \in \fh^f$ such that
%
%\begin{displaymath}
%  \alpha (h_0) > 0 \hbox{   for all   } \alpha \in \Delta^0_+
%    \cup \Delta^{1/2}_+ \, .
%\end{displaymath}

Introduce a bigrading
\begin{displaymath}
  \C^R (M) = \bigoplus_{m,n \in \frac12 \ZZ}\,  \C_{m,n} \, , 
\end{displaymath}
by letting
\begin{eqnarray*}
  \deg v_{\Lambda} \otimes \vac_R =(0,0) \, , \quad
  \deg e_{\alpha} t^n = (\alpha (x)\, , \,  - \alpha (x))\, , \quad
  \deg ht^n = (0,0)\, , \\[1ex]
  \deg \varphi_{\alpha}t^n = (\alpha (x)-1 \, , \, -\alpha (x))
    = -\deg \varphi^{\alpha}t^n \, , \quad \deg \phi_{\alpha}t^n =
    (0,0)\, ,
\end{eqnarray*}
and introduce the ascending fibration
\begin{displaymath}
  F_p = \bigoplus_{\substack{m \leq p \\n \in \frac12 \ZZ}}
   \C_{m,n} \, , \quad p \in \frac12 \ZZ \, .
\end{displaymath}
It is clear that $d^R_0 : F_p \to F_{p+1}$, hence we have the
associated graded complex $(\Gr\,  \C^R (M) = \oplus_{p \in \frac12
  \ZZ} F_p /F_{p-1/2} \, ,\, \Gr \, d^R_0)$.  It is also clear that
\begin{displaymath}
  \Gr \, d^R_0 = d^{{\rm st} ,R}_0 \, .
\end{displaymath}
We thus obtain a spectral sequence with the first term $(E_1 =
\Gr \, \C^R (M) \, , \,  d^{{\rm st} ,R}_0)$.

First of all, we need to show that this spectral sequence
converges.  For this (and for computation of characters) we need
the following commutation relations.

\begin{lemma}
  \label{lem:2.4}

  \begin{enumerate}\begin{enumerate}

  \item %%a
With respect to $\ad L^R_0$ all operators $e^R_{\alpha, n}$,
$\varphi^R_{\alpha ,n}$, $\varphi^{\alpha ,R}_n$ and
$\phi^R_{\alpha ,n}$ are eigenvectors with eigenvalue $-n$.

\item %%b
With respect to $\ad J^{\{ h \},R}_0$ with $h \in \fh$ the operators
$e^R_{\alpha ,n}$ and $\varphi^R_{\alpha ,n}$ have eigenvalue
$\alpha (h)$, and $\varphi^{\alpha ,R}_n$  has eigenvalue $-\alpha
(h)$.

\item %%c
With respect to $\ad J^{\{ h \},R}_0$ with $h \in \fh^f$ (resp. to
$\ad J^{\{ x \},R}_0$) the operators $\phi^R_{\alpha ,n}$ have
eigenvalue $\alpha (h)$ (resp.~$0$).

  \end{enumerate}\end{enumerate}

\end{lemma}

\begin{proof}
(a) follows from the fact that $L^R_0$ is the energy operator,
(b)~is (2.11)  from \cite{KRW}, the first part of~(c) is (2.12)
from \cite{KRW} and the second part holds since 
$x^{\ne}(z)=0$.
%$\sum_{\alpha , \beta \in S_{1/2}} :\phi_{\alpha} \phi^{\beta}:\, =0$.

\end{proof}

The following lemma implies the convergence of our spectral
sequence and the convergence of the Euler--Poincar\'e character
of $E_1$.

\begin{lemma}
  \label{lem:2.5}

Let $f$ be a nilpotent element of $\fg$ of principal type.  Then
we can choose $h' \in \fh^f$, such that $\alpha (h')>0$ for all
$\alpha \in \Delta^{0,\new}_+= \Delta^0 \cap \Delta_+^{\new}$.

\begin{list}{}{}

\item (a)~~
Common eigenspaces of $L^R_0$, $J^{\{ h_0\},R}_0$ and $J^{\{ h^{\prime}\},R}_0$
in $F_p$ are finite-dimensional for each $p \in \frac12 \ZZ$.

\item (b)~~
Common eigenspaces of $L^R_0$, $J^{ \{h_0 \},R }_0$, 
$J^{ \{ h^{\prime} \},R}_0$
and $J^{ \{ x \},R}_0$
in $E_1$ are finite-dimensional.

  \end{list}

\end{lemma}

\begin{proof}

The space $C^R (M)$ is obtained by applying to $v_{\Lambda}
\otimes \vac_R$ products of some number of operators of the
following form (see Lemma~\ref{lem:2.3}):

\begin{list}{}{}

\item (1)~~$e_{\alpha ,n}$ with $\alpha \in \Delta^{\new}_+$, $n<0$,

\item (2)~~$e_{\alpha ,n}$ with $\alpha \in -\Delta^{\new}_+$, $n \leq
  0$,

\item (3)~~$h_n$ with $h \in \fh$, $n<0$,

\item (4)~~$\varphi_{\alpha ,n}$ with $\alpha \in \Delta^{\new}_+$
, $n<0$,

\item (5)~~$\varphi_{\alpha ,n}$ with $\alpha \in -\Delta^{\new}_+$, $
 n \leq 0$,

\item (6)~~$\phi_{\alpha ,n}$ with $\alpha \in \Delta^{\new}_+$, $n<0$,

\item (7)~~$\phi_{\alpha ,n}$ with $\alpha \in -\Delta^{\new}_+$, $n
  \leq 0$,

\item (8)~~$\varphi^{\alpha}_n$ with $\alpha \in \Delta^{\new}_+$, $n
  \leq 0$,

\item (9)~~$\varphi^{\alpha}_n$ with $\alpha \in -\Delta^{\new}_+$, $
 n <0$.

\end{list}

By Lemma~\ref{lem:2.4}(a), only application of the operators of
the form, $e_{-\alpha ,0}$ with $\alpha \in \Delta^{\new}_+$,
$\varphi_{-\alpha ,0}$ with $\alpha \in \Delta^{\new}_+$,
$\phi_{-\alpha ,0}$ with $\alpha \in \Delta^{\new}_+$, and $\varphi^0_\alpha$ with
$\alpha \in \Delta^{\new}_+ $ may produce infinitely many states in an
eigenspace of $L^{R}_0$.  Furthermore, by Lemma~\ref{lem:2.4}(b),
application of all these operators to an eigenvector of $J^{\{
  h_0 \},R}_0$ changes the eigenvalue by a non-positive quantity.
Since $\alpha (h_0) >0$ for $\alpha \in \Delta^{1/2}_+$, we
conclude that the operators $\phi_{-\alpha ,0}$ with $\alpha \in
\Delta^{\new}_+$ change the eigenvalue of $J^{\{ h_0 \},R}_0$ by a
negative quantity.

Furthermore, application of the operators $e_{-\alpha ,0}$ with
$\alpha \in \Delta^{0,\new}_+$ change the eigenvalue of $J^{\{ h^\prime
  \},R}_0$ by a negative quantity $-\alpha (h')$.

Finally, application of the operators $e_{-\alpha ,0}$ with
$\alpha \in \Delta^{\new}_+$ either changes the eigenvalue of $J^{\{ h_0
\},R}_0$ by a negative quantity or else does not change this
eigenvalue, but changes the degree of the filtration (resp. the
eigenvalue of $J^{\{ x \},R}_0$) by a positive quantity.

This proves both statements of the lemma.

\end{proof}

Now we are in a position to compute the 
Euler--Poincar\'e character 
$\ch_{H^R (M)} (\tau ,z, t)$, where $z \in \fh^f$,
assuming that $f$ is a nilpotent element of principal type. From
now on we shall assume that $f$ is a nilpotent element of principal
type, when talking about the Euler--Poincar\'e characters.

First, note that, by Lemma~\ref{lem:2.5}(a) our spectral sequence
converges.  Hence, by Lemmas~\ref{lem:2.5}(b) and
\ref{lem:2.1}(a), we have:
\begin{displaymath}
  \ch_{H^R (M)} (\tau ,z,t) = e^{2\pi ikt}\lim_{\epsilon \to 0}
     \sum_{m \in \ZZ} (-1)^m \tr_{\C^{R}_m(M)} q^{L^R_0}
        e^{2\pi i J_0^{ \{ z + \epsilon x \} , R}} \, .
\end{displaymath}
Since $L^R_0 = L^{\fg ,R}_0 -x_0+L^{\ch ,R}_0 + L^{\ne ,R}_0$, we
obtain that the right-hand side is equal to the product of
$e^{2\pi i kt}$ and the following two expressions:
\begin{eqnarray*}
  I &=& \lim_{\epsilon \to 0} \tr_M  \, 
         q^{L^{\fg ,R}_0 -x_0} 
         e^{2\pi i (z+\epsilon x)_0} \, ,\\
  II &=& \lim_{\epsilon \to 0} \str_{F^R(\fg ,f)} 
         q^{L^{{\rm ch} ,R}_0 +L^{\ne , R}_0}
         e^{2\pi i ((z+\epsilon x)^{{\rm ch},R}_0 + 
           (z+\epsilon x)^{{\rm ne} ,R}_0)}\, ,
\end{eqnarray*}
where $\str$ stands for the supertrace.

We let $S_+ = \Delta_+ \backslash \Delta^0_+$ for short (which is
consistent with our earlier notation).

\begin{lemma}
  \label{lem:2.6}

  \begin{enumerate}\begin{enumerate}

  \item %%a
    The constants $s_{\fg}$, $s_{\ne}$ and $s_{\ch}$, introduced
    in \cite{KW4}, Proposition~3.2, for arbitrary $\sigma$, are
    given by the following formulas for $\sigma = \sigma_R$:
    \begin{displaymath}
      s_{\fg} = - \frac{k}{k+h^\vee} \sum_{\alpha \in S_+}
                \binom{s_\alpha}{2} \, , \quad 
      s_{\ne} = -\frac{1}{16}\dim \fg_{1/2}\, , \quad
      s_{\ch} = (\rho|x) -\frac{h^\vee}{2}|x|^2\, .
    \end{displaymath}

\item %%b
  The element $\gamma' \in \fh^*$, introduced in \cite{KW4},
  Corollary~3.1, is given by the following formula for  $\sigma =
  \sigma_R$:
  \begin{displaymath}
    \gamma' = \sum_{\alpha \in S_+} s_{\alpha} \alpha - \rho \, .
  \end{displaymath}
%
%Consequently, by \cite{KW4}, Corollary~3.2,
%
%\begin{displaymath}
%  \sum_{\alpha \in S_+} s_\alpha \alpha = \rho - \bar{\hat{\rho}}^R \, .
%\end{displaymath}

\end{enumerate}\end{enumerate}
\end{lemma}

\begin{proof}

Formula for $s_\fg$ follows from \cite{KW4}, formula~(3.9), using
that
\begin{displaymath}
  \binom{s_{-\alpha}}{2} = \binom{s_{\alpha}}{2}
\end{displaymath}
(which follows from $s_{\alpha} + s_{-\alpha} =1$) and $s_{\alpha}
  =0$ if $\alpha = \Delta_+^0$.

Formula for $s_{\ne}$ follows from that in \cite{KW4},
formula~(3.10), since $s_\alpha = \pm 1/2$ for $\alpha \in
\Delta_{\mp}^{1/2}$.

Since $s_\alpha = -\alpha (x)$ or $1-\alpha (x)$, each summand in
\cite{KW4}, formula~(3.11) for $s_\ch$, equals $\alpha (x)
(1-\alpha (x))/2$.  Hence $s_\ch = \frac12 \sum_{\alpha \in S_+}
\alpha (x) -\frac12 \sum_{\alpha \in S_+} \alpha (x)^2 = (\rho
|x) -\frac{h^{\vee}}{2} |x|^2$, proving~(a).

Next, by definition, we have: $\gamma' = \frac12 \sum_{\alpha \in
S_+ \cup \Delta^0_+ } (s_{\alpha} \alpha + s_{-\alpha} (-\alpha))$.  Since
$s_\alpha + s_{-\alpha}=1$ and $s_\alpha =0$ if $\alpha \in
\Delta^0_+$, we obtain: $\gamma' = \sum_{\alpha \in S_+} s_\alpha
\alpha -\frac12 \sum_{\alpha \in S_+ \cup \Delta^0_+} \alpha$, proving~(b).

\end{proof}

\begin{lemma}
  \label{lem:2.7}

  \begin{enumerate}\begin{enumerate}
\item %%a
$\hat{\rho}^R = h^\vee D -\gamma' + (\sum_{\alpha \in S_+} \alpha
(x) s_\alpha -\rho (x) + \frac{h^\vee}{2} |x|^2)K$.

\item %%b
$\frac{ (\bar{\Lambda} | \bar{\Lambda} + 2 \bar{\hat{\rho}}^R)}{2
  (k+h^\vee)} + s_{\fg} = \frac{(\Lambda | \Lambda +2
  \hat{\rho}^R)}{2(k+h^\vee)} -\Lambda (D)$ for any $\Lambda \in \hat{\fh}$.

\end{enumerate}\end{enumerate}
\end{lemma}

\begin{proof}
By \cite{KW3} , Corollary 3.2, 
%Lemma~{\ref{lem:2.6}}(b), 
we have:
\begin{displaymath}
  \hat{\rho}^R = h^\vee D -\gamma' + aK \hbox{  for some } a \in
  \CC \, .
\end{displaymath}
In order to compute $a$, recall that 
%recall that, by Lemma~{\ref{lem:2.2}}(b), 
%$\hat{\prod}^R = t_x\bar{w}
%(\hat{\prod})$ for some $\bar{w} \in W$. Hence
%
\begin{displaymath}
\hat{\rho}^R = t_x\bar{w} (\hat{\rho}) =t_x\bar{w} 
(h^\vee D +\rho)
%= h^\vee t_x  (D) - t_x\bar{w} (\rho) 
= h^\vee D + h^\vee x+ \bar{w}(\rho) -
(\frac{h^\vee}{2} |x|^2 + (x|\bar{w} (\rho))K \, .
\end{displaymath}
Comparing with the first formula for $\hat{\rho}^R$, we get:
\begin{displaymath}
  -\gamma' = \bar{w} (\rho) +h^\vee x \hbox{  and  }
  a=-\frac{h^\vee}{2} |x|^2 - (\bar{w}(\rho)|x)\, .
\end{displaymath}
Hence $a= \frac{h^\vee}{2} |x|^2 +(\gamma'|x)$.  Substituting
$\gamma'$ from Lemma~\ref{lem:2.6}(b), we obtain~(a).

We have by definition and (a):
\begin{eqnarray*}
  \Lambda &=& kD + \bar{\Lambda} + (\Lambda |D)K \, , \\
  \hat{\rho}^R &=& h^{\vee}D + \bar{\hat{\rho}}^R 
      +(\sum_{\alpha \in S_+} \alpha (x) s_{\alpha} -
      \rho (x) + \frac{h^\vee}{2} |x|^2) K \, .
\end{eqnarray*}
Hence
\begin{displaymath}
  \frac{(\Lambda | \Lambda + 2\hat{\rho}^R)}{2(k+h^\vee)} = 
  \frac{(\bar{\Lambda}|\bar{\Lambda} +2\bar{\hat{\rho}}^R)}{2(k+h^\vee)}
  +\frac{k}{k+h^\vee} \left( \sum_{\alpha \in S_+} 
      \alpha (x) s_{\alpha}-\rho (x) +\frac{h^\vee}{2}
      |x|^2 \right) + \Lambda (D) \, .
\end{displaymath}
Adding to both sides $s_{\fg}$ and using in the left-hand side
the formula for $s_{\fg}$, given by Lemma~\ref{lem:2.6}(a), we
obtain:
\begin{displaymath}
  \frac{(\Lambda | \Lambda + 2 \hat{\rho}^R)}{2(k+h^\vee)} 
    -\Lambda (D) =
  \frac{(\bar{\Lambda}|\bar{\Lambda} + 2\hat{\rho}^R)}{2(k+h^\vee)}
  + s_{\fg}+ \frac{k}{k+h^\vee} A \, ,
\end{displaymath}
where $A=\sum_{\alpha \in S_+} \left( \binom{s_\alpha}{2} +
  \alpha (x) s_{\alpha}\right) - \rho (x) + \frac{h^\vee}{2} |x|^2$.  
By definition of $s_{\rm ch}$ given by \cite{KW4},
formula~(3.11), and the third formula of Lemma~\ref{lem:2.6}(a),
we obtain that $A=0$, proving~(b).

\end{proof}

\begin{lemma}
  \label{lem:2.8}

  \begin{enumerate} \begin{enumerate}
  \item %%a
On the $\hat{\fg}^R$-module $M$ (with the highest weight
$\Lambda$) we have:
\begin{displaymath}
  L^{\fg ,R}_0 = \frac{(\Lambda | \Lambda +2 \hat{\rho}^R )}
     {2 (k+h^\vee)} I_M -D \, .
\end{displaymath}

\item %%b
For $z \in \fh^f$ we have:
\begin{eqnarray*}
\lefteqn{\hspace{-5in} II = q^{s_{\ch} + s_{\ne}} e^{\pi i \sum_{\alpha \in
         \Delta_{1/2}} s_\alpha \alpha (z) + 2\pi i (\bar{\hat{\rho}}^R (z) - \rho (z))}}\\
      \times  \prod_{\alpha \in \hat{\Delta}^R_+} 
          (1-e^{-2\pi i \alpha (z)})^{{\rm mult~~} \alpha}
          \prod^\infty_{j=1} (1-q^j)^{-r}
          \prod_{ \substack{\alpha \in \hat{\Delta}^{R,\re}_+ \\
            \alpha (x) = 0, -1/2}}
          (1 -e^{-2\pi i \alpha (z)})^{-1}\, ,
\end{eqnarray*}
%
%where $s_{\rm ne}$ and $s_{\rm ch}$ are given by Lemma~\ref{lem:2.6}(a).
  \end{enumerate} \end{enumerate}

\end{lemma}

\begin{proof}
By \cite{KW4}, Proposition~3.2 and Corollary~3.2, we have:
\begin{displaymath}
  L^{\fg ,R}_0 v_{\Lambda} = \left( \frac{1}{2{(k+h^\vee)}} 
    ( \bar{|\Lambda} |^2 
      + 2 (\bar{\Lambda} | \hat{\rho}^R)) + s_{\fg} \right)v_{\Lambda}\, .
\end{displaymath}
Hence, by Lemma~\ref{lem:2.7}:
\begin{displaymath}
  L^{\fg ,R}_0 v_\Lambda = \left( 
     \frac{(\Lambda | \Lambda + 2\hat{\rho}^R)}{2(k+h^\vee)}
       - \Lambda (D) \right) v_{\Lambda } \, .
\end{displaymath}
Since $L^{\fg ,R}_0 = -D + {\rm const}$, (a)~follows.

By Lemmas~\ref{lem:2.3}, \ref{lem:2.4} and \cite{KW4},
Proposition~3.2 and Corollary~3.3, we have:
\begin{eqnarray*}%\begin{displaymath}
\lefteqn{\hspace{-5in} II = \lim_{\epsilon \to 0} q^{s_{\rm ch} + s_{\rm ne}}
        e^{2\pi i \left( \frac12 \sum_{\alpha \in \Delta^{1/2}} 
            s_{\alpha} \alpha (z+\epsilon x) - 
            (\rho -\bar{\hat{\rho}}^R) (z+\epsilon x)\right)}
            \ch_{F^{R}} (\tau , z+\epsilon x)\, ,}\\
%\end{displaymath}
%
\noalign{where}\\
%
%\begin{eqnarray*}
  \ch_{F^R} (\tau, z +\epsilon x) 
        = \hspace{-1ex} \prod_{\substack{\alpha \in S_+ \\ n>-s_{\alpha}}}
         \hspace{-.5ex}   \left(1-q^n e^{2\pi i \alpha (z+\epsilon x)}\right)
%
%     && 
\hspace{-.5ex} \prod_{\substack{\alpha \in S_+ \\ n>-s_{-\alpha}}}
       \hspace{-.5ex} \left(1-q^n e^{-2\pi i \alpha (z+\epsilon x)}\right)
 \hspace{-.5ex} \prod_{\substack{\alpha \in \Delta^{1/2}\\ n>-s_\alpha}}
   \hspace{-.5ex}  \left(  1-q^n e^{2\pi i \alpha (z+\epsilon x)}\right)^{-1}\, .
\end{eqnarray*}
Hence
\begin{eqnarray*}
  II &=& q^{s_{\rm ch} + s_{\rm ne}} 
       e^{\pi i \sum_{\alpha \in \Delta^{1/2}} s_\alpha 
         \alpha (z) +2\pi i (\bar{\hat{\rho}}^R -\rho )(z)}\\
   & \times & \left( 
      \prod_{\substack{\alpha\in \hat{\Delta}^R_+ \\ 
          |\alpha (x)|\geq 1}} (1-e^{-\alpha})
         \prod_{\substack{\alpha \in \hat{\Delta}^R_+\\
            \alpha (x) = 1/2}} (1-e^{-\alpha})\right) (-\tau D+z)\, ,  
\end{eqnarray*}
which proves~(b).

\end{proof}

\begin{lemma}
  \label{lem:2.9}
  \begin{enumerate}\begin{enumerate}

  \item %%a
    $\sum_{\alpha \in \Delta^{1/2}} s_\alpha \alpha (z)
    =-\sum_{\alpha \in \Delta^{1/2}_+} \alpha (z)$ if $z \in
    \fh^f$.

\item %%b
All the $\fg^f_0$-modules $\fg_j$ are self-contragredient.  In
particular $\tr_{\fg_j} {\ad}  {\it a} =0$ for all $a \in \fg^f_0$ and $j
\in \frac12 \ZZ$, and $\rho (z) = \frac12 \sum_{\alpha \in
  \Delta^0_+} \alpha (z)$ for $z \in \fh^f$.

  \end{enumerate}\end{enumerate}

\end{lemma}

\begin{proof}
(a) follows from the fact that $s_{\alpha} = -1/2$ (resp. $=1/2$)
if $\alpha \in \Delta^{1/2}_+$ (resp.~$\Delta^{1/2}_-$).

In order to prove~(b), note that the $\fg^f_0$- (even
$\fg_0$-)modules $\fg_j$ and $\fg_{-j}$ are contragredient since
they are paired by the invariant form $(\, . \, | \, . \, )$.  On
the other hand, $(\ad f)^{2j}$ gives an isomorphism of the
$\fg^f_0$-modules $\fg_j$ and $\fg_{-j}$.

\end{proof}

\begin{lemma}
\label{lem:2.10}
$\{ \alpha  |_{\fh^f}    \,\,\,| \alpha \in \Delta^R_+ \, , \,
\alpha (x) =1/2 \} = \{\alpha |_{\fh^f} \, \,\,| \alpha \in \Delta^R_+ \, ,\, 
\alpha  (x)= -1/2\}$.
\end{lemma}

\begin{proof}
Recall that $\Delta^{1/2} = \Delta^{1/2}_+  \sqcup
\Delta^{1/2}_-$ and $s_{\alpha} = \mp 1/2$ if $\alpha \in
\Delta^{1/2}_\pm$.  Therefore, if $\alpha \in \Delta^{1/2}_\pm$, then
$\pm (\alpha - K/2) \in \Delta^R_+$.  On the other hand, by
Lemma~\ref{lem:2.9}(b), $\Delta^{1/2}_+ |_{\fh^f}
=-\Delta^{1/2}_- |_{\fh^f}$, which proves the lemma.

\end{proof}

Note that the $D$-operator for the $\sigma_R$-twisted affine Lie
algebra $\hat{\fg}^R$ is
\begin{displaymath}
  D^R : = t_x \bar{w} (D) = D +x- \frac{(x|x)}{2} K \, .
\end{displaymath}
Let
\begin{displaymath}
  R_{\hat{\fg}^R} = e^{\hat{\rho}^R} \prod_{\alpha \in \hat{\Delta}^R_+}
       (1-e^{-\alpha})^{{\rm mult~}\alpha}
\end{displaymath}
be the Weyl denominator for $\hat{\fg}^R$.  Using
Lemmas~\ref{lem:2.8} and \ref{lem:2.9}, we can rewrite $\ch_{H^R
  (M)}$ as follows (here and further $z \in \fh^f$):
\begin{eqnarray}\label{eq:2.2}
\hspace*{6ex}\ch_{H^R(M)} (\tau ,z,t) \!\!\!\!
  & =& \!\!\!\! q^{\frac{(\Lambda |\Lambda+ 2 \hat{\rho}^R)}{2(k+h^\vee)}
     \! \, -\frac{k}{2} (x|x)+s_{\rm ch}+s_{\rm ne}}%%% \notag\\[1ex] 
%  
%%\label{eq:2.2}
   e^{\pi i(\sum_{\alpha\in\Delta^{1/2}}s_{\alpha} \alpha
   (z) -  \sum_{\alpha \in \Delta^0_+} \alpha (z))}
    e^{-2\pi i h^\vee t}  \\ %%\tag{2.2}
 &\times & \left(\frac{R_{\hat{\fg}^R} \ch_M}
        {\prod^{\infty}_{j=1} (1-e^{-jK})^r
         \prod_{\substack{\alpha \in \hat{\Delta}^{R,\re}_+ \\
            \alpha (x) = 0,-1/2}} (1- e^{-\alpha})}
           \right) (2\pi i (-\tau D^R + z +tK)).\notag
\end{eqnarray}
This formula makes sense since 
$\alpha |_{\fh^f} \neq 0$ if $\alpha(x)=0$
because $f$ is a nilpotent element of principal type, and 
$\alpha |_{\fh^f} \neq 0$ if $|\alpha(x)|=1/2$ for any nilpotent $f$
\cite{EK}.

Formula (\ref{eq:2.2})implies the following corollary.
\begin{corollary}
  \label{cor:2.1}
The minimal eigenvalue of $L^R_0$ on $H^R (M)$ equals
\begin{displaymath}
  \frac{(\Lambda|\Lambda + 2\hat{\rho}^R)}{2(k+h^\vee)}
  - (\Lambda |D^R) 
+ s_{\ch} + s_{\ne} -\frac{k}{2}(x|x)\, .
\end{displaymath}
%
%The image of the vector $v_{\Lambda} \otimes \vac_R $ spans the 
%eigenspace of $L_0^R$,
%attached to this eigenvalue. 
All other eigenvalues are obtained from this 
by adding a positive integer.
\end{corollary}

By the general principles of conformal field theory, define the
normalized Euler--Poincar\'e character by:
\begin{displaymath}
  \chi_{H^R (M)} (\tau ,z,t) = e^{-\frac{\pi i \tau}{12} c(\fg ,f,k)}
     \ch_{H^R (M)} (\tau ,z,t) \, .
\end{displaymath}
Using (\ref{eq:2.1}) and Lemma~\ref{lem:2.6}(a), we obtain:
\begin{eqnarray*}
  - \frac{1}{24} c(\fg ,f,k) +s_{\rm ch}
     + s_{\rm ne}- \frac{k}{2} |x|^2 &=& -\frac{1}{24}
        (\dim \fg_0 + \dim \fg_{1/2})
           +\frac{|\rho |^2}{2(k+h^\vee)}\\
        &=& -\frac{1}{24} \dim \fg^f 
          + \frac{|\hat{\rho}^R |^2}{2(k+h^\vee)}\, .
\end{eqnarray*}
(Recall that $\dim \fg^f =\dim \fg_0 + \dim \fg_{1/2}$.)
Using this, we can rewrite (\ref{eq:2.2}) as follows:
\begin{equation}
  \label{eq:2.3}
  \chi_{H^R (M)} (\tau ,z ,t) = \frac{B_{\Lambda}(\tau ,z,t)}
     {\psi (\tau , z,t)} \, , 
\end{equation}
where
\begin{equation}
\label{eq:2.4}
B_\Lambda (\tau,z,t)= 
   q^{\frac{|\Lambda +\hat{\rho}^R |^2}
        {2(k+h^\vee)}} (R_{\hat{\fg}^R} \ch_M) (2\pi i (-\tau D^R +z+tK))
\end{equation}
is the numerator of the normalized character $\chi_M$ of the 
$\hat{\fg}$-module $M$ (with highest weight $\Lambda$ of level $k$)
and
\begin{eqnarray*}
   \psi (\tau ,z,t) 
     &=& e^{2\pi i h^\vee t} q^{\frac{1}{24} \dim \fg^f}
         e^{\pi i \sum_{\alpha \in \Delta^{R,0}_+ \cup
             \Delta^{R,1/2}_+} \alpha (z)}\\
    & \times & \prod^{\infty}_{j=1} (1-q^j)^r
        \prod_{\substack{\alpha \in \hat{\Delta}^{R, {\rm re}}_+\\ 
            \alpha (x) = 0,-1/2}}
              (1-e^{2\pi i \alpha (\tau D^R -z)})\, .
\end{eqnarray*}
Using Lemma \ref{lem:2.10}, 
%that $\Delta^{1/2}_+ |_{\fh^f} =-\Delta^{1/2}_- |_{\fh^f}$,
we rewrite the last formula as
\begin{eqnarray}
\label{eq:2.5}
  \psi (\tau ,z,t) &=& e^{2\pi i h^\vee t}
     q^{\frac{1}{24} \dim \fg^f}e^{\pi i
       \sum_{\alpha \in \Delta^{R,0}_+ \cup \Delta^{R,1/2}_+}\alpha (z)}\\
   &\times & \prod^\infty_{n=1} (1-q^n)^r
   \prod_{\alpha \in \Delta^{R,0}_+ \cup \Delta^{R,1/2}_+}      
    (1-q^{n-1} e^{-2\pi i \alpha (z)})(1-q^n e^{2\pi i \alpha (z)})\, .\notag
\end{eqnarray}

If $\Lambda \in Pr^k$ and $M=L(\Lambda)$, the numerator $B_{\Lambda}$ 
is explicitly
known (see Theorem \ref{th:1.2}); in \cite{KW1}, \S3 one can find
an explicit expression in terms of theta functions:
\begin{displaymath}
  B_\Lambda (\tau,z,t)=A_{ \Lambda^0+\hat{\rho}}( u\tau, 
       \bar{y}^{-1} (z+\tau \beta)\, , \, \frac1u (t+(z|\beta) +
          \frac{\tau |\beta |^2}{2})) \,, z \in \fh^f \, , 
\end{displaymath}
where $  A_\lambda (\tau ,z,t) = \sum_{w \in W} \epsilon (w) 
     \Theta_{w(\lambda)} (\tau,z,t)$ (cf.~\cite{K1}, Chapter~13).
    Dividing and multiplying the right hand side of (\ref{eq:2.3}) by 
$ A_{\hat{\rho}^R}
     (u\tau , \bar{y}^{-1} (z+\tau \beta)\, , \, \frac{1}{u}
     (t+(z|\beta) + \frac{\tau |\beta |^2}{2}))$,  we  obtain:
\begin{equation}
\label{eq:2.6}
 \chi_{H^R (L (\Lambda))} (\tau ,z,t) = \chi_{L^R (\Lambda^0)}
   ( u\tau , \bar{y}^{-1} (z+\tau \beta) , \frac{1}{u} 
     ( t+ (z |\beta) +\frac{\tau |\beta|^2}{2})) 
       \frac{C (\tau,z,t)}{\psi (\tau ,z,t)}\, , \\
\end{equation}
where
\begin{eqnarray*}
   C (\tau ,z,t) &= q^{\frac{u}{24}\dim \fg} 
     e^{2\pi i ((\rho^R |\bar{y}^{-1} (z+\tau \beta)) +
        \frac{h^\vee}{u} (t + (z|\beta) +\frac{\tau |\beta
         |^2}{2})) }\\
 &\times \prod_{\alpha \in \hat{\Delta}^R_+} 
    (1-e^{-\alpha})^{{\rm mult}\, \alpha} (u\tau, \bar{y}^{-1}(z+\tau
\beta),0)\,.  
\end{eqnarray*}
\begin{remark}
\label{rem:2.1}
(a) Since $\fg^f$ and 
$ \fg_0 + \fg_{1/2}$ are isomorphic as $\fh^f$-modules,
we have:
\begin{displaymath}
  (1-q^n)^r \prod_{\alpha \in \Delta^{R,0}_+ \cup \Delta^{R,1/2}_+} 
(1-q^n e^{-2\pi i \alpha (z)})(1-q^n e^{2\pi i \alpha (z)})= {\rm det}_{\fg^f}
(1- q^n e^{2 \pi iz}),\, \, z \in \fh^f. 
\end{displaymath}
%Choose $h_0 \in \fh^f$ to be positive on 
%$ \Delta^0_+ \cup \Delta^{1/2}_+$, and let 
%$\fg^f_+$ be the subalgebra of $\fg^f$, spanned by all eigenspaces
%of $\ad h_0$ with positive eigenvalues on $\fg^f$.
Let $\rho^f(z)=\frac{1}{2} \sum_{\alpha \in \Delta^{R,0}_+ \cup 
\Delta^{R,1/2}_+}
\alpha (z)$, $z \in \fh^f$. 
Thus, formulas (\ref{eq:2.3}) - (\ref{eq:2.5}) mean that 
the quantum Hamiltonian reduction of the
$\hat{\fg}^R$-module $M$ to the $W_k(\fg,f)$-module $H^R(M)$ does not change 
the numerator of its normalized
character, but ``reduces'' its denominator, replacing $\fg$
by $\fg^f$.

(b) The product in $C(\tau, z, t)$ is equal to 
  $\prod_{\alpha \in \hat{\Delta}^R_{\Lambda ,+}} 
(1-e^{ -\alpha})^{\mult \alpha} (\tau,z,0)$.
This follows from formula (\ref{eq:2.11}) below.  
\end{remark}

%The following proposition gives a sufficient condition of
%coincidence, up to a sign, of the Euler--Poincar\'e characters
%$\ch_{H^R (L (\Lambda))}$ for different principal admissible $\Lambda$.

%\begin{proposition}
%  \label{prop:2.1}

%Let $\Lambda$ and $\Lambda'$ be principal admissible weights,
%such that $\Lambda' = y_0 . \Lambda$ for some $y_0 \in \tilde{W}$.

%\alphaparenlist
%\begin{enumerate}
%\item %%a
%Let $\fh'$ be a subspace of $\fh$, such that there exists $w_0 \in
%\hat{W}_\Lambda$, for which $y_0 w_0 |_{\CC D+\fh'} =Id $.  Then
%$\ch_{L(\Lambda')} |_{\CC D +\fh'} = \epsilon (w_0)
%\ch_{L(\Lambda)} |_{\CC D+\fh'}$.

%\item %%b
%If there exists $w_0 \in \hat{W}^R_\Lambda$, for which $y_0w_0
%|_{\CC D^R + \fh^f} = Id$, then 
%
%\begin{displaymath}
%  \ch_{H^R (L(\Lambda'))} = \epsilon (w_0) 
%       \ch_{H^R (L (\Lambda))} \, .
%\end{displaymath}

%\end{enumerate}
%\end{proposition}

%\begin{proof}

%(a) follows from the character formula (1.3).  (b)
%follows from (2.2), using (a).

%\end{proof}

\subsection{Modular invariance, asymptotics, conditions of vanishing
and convergence of normalized Euler-Poincar\'e characters in the 
Ramond sector.}
\label{sec:2.3}~~
Formula (\ref{eq:2.5}) can be rewritten again in terms of the Dedekind
$\eta$-function $\eta (\tau) = q^{1/24} \prod^\infty_{j=1}
(1-q^j)$ and the following modular function in $\tau \in \CC ^+, s \in \CC$:
\begin{equation}
  \label{eq:2.7}
  f(\tau ,s) = e^{\frac{\pi i \tau}{6}} e^{\pi i s}
     \prod^\infty_{n=1} (1-q^n e^{2\pi i s})(1-q^{n-1}e^{-2\pi is})
\end{equation}
as follows:
\begin{equation}
  \label{eq:2.8}
  \psi (\tau ,z,t) = e^{2\pi i h^\vee t} \eta (\tau)^r
     \prod_{\alpha \in \Delta^{R,0}_+ \cup \Delta^{R,1/2}_+}
        f(\tau ,\alpha (z))\, .
\end{equation}

\begin{lemma}
  \label{lem:2.11}

  \begin{enumerate} \begin{enumerate}

  \item %%a
One has the following transformation properties:
\begin{eqnarray*}
  \eta \left(-\frac{1}{\tau}\right) = (-i \tau)^{1/2} \eta (\tau )\, ,\,
  f \left( -\frac{1}{\tau}\,, \, \frac{s}{\tau} \right) =
    -ie^{\pi i s^2/ \tau} f (\tau ,s)\, .
\end{eqnarray*}

\item %%b
As $\tau \downarrow 0$, one has:
\begin{eqnarray*}
 \eta (\tau ) \sim (-i\tau)^{-1/2} e^{-\pi i /12\tau}\, ,\, 
 f (\tau ,-a\tau) \sim 2 (\sin\pi a)  e^{-\pi i /6\tau}\, , 
  \quad a \in \CC \, .
\end{eqnarray*}

  \end{enumerate} \end{enumerate}

\end{lemma}

\begin{proof}
By the Jacobi triple product identity, one has:
\begin{eqnarray*}
  f (\tau ,s) &=& \theta (\tau ,s) /\eta (\tau ) \, , \, 
     \hbox{  where }\\
  \theta (\tau ,s) &=& e^{\pi i (\tau /4 +s)} 
     \prod^\infty_{n=1}    (1-q^n) (1-q^n e^{2\pi i s})
        (1-q^{n-1} e^{-2\pi i s})\\
     &=& \sum_{n \in \ZZ} (-1)^{n-1} e^{2\pi i s (n-1/2)}
        q^{(n-1/2)^2/2} \, .
\end{eqnarray*}
Now (a) follows from the transformation formula for
theta-functions (see e.g.~\cite{K1}, Chapter~13 for details).

(b) follows from formulas in (a) with $\tau$ replaced by $-1/\tau$.
\end{proof}

Assuming that $k \neq 0$, define the following quadratic form on
$\fh^f$:
\begin{displaymath}
  Q(z) = \frac{k+h^\vee}{k} (z|z) - \frac{1}{k}
      \sum_{\alpha \in \Delta^{R,0}_+ \cup \Delta^{R,1/2}_+}
         \alpha (z)^2 \, .
\end{displaymath}
Now we can prove a modular transformation formula for the
normalized Euler--Poincar\'e characters $\chi_{H^R (L(\Lambda
  ))}$, where $\Lambda$ runs over the set $Pr^{k,R} (\mod \CC K)$ of all
principal admissible weights of level~$k$  of the affine Lie
algebra $\hat{\fg}^R$. (The assumption that $k \neq 0$ is not restrictive 
since $Prn^{0,R}(\fg ,f) = \emptyset$ if $f \neq 0$.) 

\begin{theorem}
  \label{th:2.1}

For $\Lambda \in Pr^{k,R}$ one has:
\begin{enumerate} \item[] \begin{enumerate}

\item %%a
%%\begin{eqnarray*}
  $ \chi_{H^R (L (\Lambda))}   \left( -\frac{1}{\tau}\, , \, 
    \frac{z}{\tau} \, , \, t-\frac{Q(z)}{2\tau} \right)
     = (-i)^{\frac12 (\dim \fg -\dim \fg^f)}%\\
      \sum_{\Lambda' \in Pr^{k,R}} a (\Lambda ,\Lambda')
        \chi_{H^R (L (\Lambda'))} (\tau ,z,t) $,
%%%\end{eqnarray*}
%
where $a(\Lambda , \Lambda')$ is as defined in Theorem~\ref{th:1.3}(a).

\item %%b
$\chi_{H^R (L(\Lambda))} (\tau +1,z,t) = e^{2\pi i s^f_{\Lambda}}
\chi_{H^R (L(\Lambda))} (\tau ,z,t)$, where
$s^f_\Lambda=\frac{|\Lambda + \hat{\rho}^R|^2}{2(k+h^\vee)}-\Lambda(D^R)
- \frac{1}{24} \dim \fg^f$.

\end{enumerate}\end{enumerate}

\end{theorem}

\begin{proof}

Since $B_{\Lambda} = \chi_{L(\Lambda)} D_{\hat{\fg}^R}$ , we have from
Theorem~\ref{th:1.3}(a) and (\ref{eq:1.5}):
\begin{displaymath}
  B_\Lambda \left( - \frac{1}{\tau}\, , \, \frac{z}{\tau}\, , \, 
    t-\frac{(z|z)}{2\tau} \right) = (-i)^{|\Delta_+|}
     (-i\tau)^{r/2}\sum_{\Lambda' \in Pr^{k,R}} a (\Lambda ,\Lambda')
        B_{\Lambda'} (\tau ,z,t)\, .
\end{displaymath}
Hence
\begin{eqnarray}
  \label{eq:2.9}
  B_{\Lambda }\left( - \frac{1}{\tau}\,, \, \frac{z}{\tau}\, , \,
    t-\frac{Q(z)}{2\tau} \right) 
     &=& (-i)^{|\Delta_+|} (-i\tau)^{r/2}
          e^{\frac{\pi i}{\tau} (k+h^\vee) ((z|z)-Q (z))}\\
     && \times \sum_{\Lambda' \in Pr^{k,R}} a (\Lambda ,\Lambda')
        B_{\Lambda'} (\tau ,z,t)\, .\notag
\end{eqnarray}
Here we used also that the only dependence of $B_\Lambda$ on $t$
is the factor $e^{2\pi i (k+h^\vee)t}$.  On the other hand, using
formula~(\ref{eq:2.8}) for $\psi (\tau ,z,t)$ and
Lemma~\ref{lem:2.11}(a), we obtain:
  \begin{eqnarray}
    \label{eq:2.10}
    \psi \left( - \frac{1}{\tau} \, , \, \frac{z}{\tau}\, , \, 
      t-\frac{Q(z)}{2\tau} \right) 
        &=& (-i)^{|\Delta^{R,0}_+ \cup \Delta^{R,1/2}_+|}
            e^{-\frac{\pi i h^\vee}{\tau} Q(z)}\\
            && \times e^{\frac{\pi i}{\tau} 
              \sum_{\alpha \in \Delta^{R,0}_+ \cup \Delta^{R,1/2}_+}
               \alpha (z)^2} \psi (\tau ,z,t) \, .\notag
  \end{eqnarray}
Now (a) follows from (\ref{eq:2.3}), (\ref{eq:2.9}),
(\ref{eq:2.10}), the definition of $Q(z)$ and the fact that
$\dim \fg^f = \dim \fg_0 +\dim \fg_{1/2}
      = 2 |\Delta^0_+\cup \Delta^{1/2}_+ | +r$.

(b) follows from (\ref{eq:2.3}), (\ref{eq:2.4}), (\ref{eq:2.5}).
\end{proof}

\begin{theorem}
  \label{th:2.2}
For the principal admissible weight $\Lambda = (t_{\beta}
\bar{y}).(\Lambda^0 + (k+h^\vee -p) D^R) \in Pr^{k,R}$,  where
$\bar{y} \in W^R , \beta \in Q^{*,R}$,
$k+h^\vee = p/u $, $\Lambda^0 \in \hat{P}^{p-h^\vee, R}_+$, one has for
each $z \in \fh^f$, as $\tau \downarrow 0$:
\begin{displaymath}
  \ch_{H^R (L(\Lambda))} (\tau, -\tau z ,0) 
     \sim \epsilon (\bar{y}) u^{-r/2} a (\Lambda^0) A_{\beta}(z)
   e^{\frac{\pi i }{12\tau} (\dim \fg^f-\frac{h^\vee}{pu} \dim \fg)}\, ,
\end{displaymath}
where
\begin{displaymath} A_{\beta}(z)=\frac{\prod_{\alpha \in \Delta^R_+} 
2\sin{\frac{\pi (\alpha |z-\beta)}{u}}} {\prod_{\alpha \in \Delta^{R,0}_+
\cup\Delta^{R,1/2}_+} 2\sin \pi (\alpha |z)} 
\end{displaymath}
and $a (\Lambda^0)$ is the constant defined in Theorem~\ref{th:1.3}(b).

\end{theorem}

\begin{proof}

Using Lemma~\ref{lem:2.2}(b) and  Lemma~\ref{lem:2.11}(b) we obtain from 
(\ref{eq:2.8}) and a
similar formula for $D_{\hat{\fg}^R}$ (or (\ref{eq:1.6})):
\begin{eqnarray*}
  D_{\hat{\fg}^R } (\tau , -\tau z,0) \sim (-i\tau)^{-r/2}  
      \prod_{\alpha \in \Delta^R_+} 2 \sin \pi (\alpha |z)
         e^{-\frac{\pi i}{12 \tau} \dim \fg} \, ,\\
  \psi (\tau ,-\tau z,0) \sim (-i\tau)^{-r/2} 
      \prod_{\alpha \in \Delta^{R,0}_+ \cup \Delta^{R,1/2}_+}
          2\sin \pi (\alpha |z) e^{-\frac{\pi i}{12\tau}\dim \fg^f}\, .
\end{eqnarray*}
The asymptotics for the numerator of $\ch_{H^R (L(\Lambda))}
(\tau ,-\tau z,0)$ follow from the first of these formulas and
Theorem~\ref{th:1.3}(b).  Now Theorem~\ref{th:2.2} follows, using
the second of these formulas.

\end{proof}

%\begin{corollary}
%  \label{cor:2.2}
%If $\Lambda \in Pr^{k,R}$ and $k=-h^\vee + \frac{p}{u}$ is such that
%$\dim \fg^f < \frac{h^\vee}{pu} \dim \fg$, then $\ch_{H^R (L (\Lambda))}=0$.

%\end{corollary}
Given $\Lambda \in \hat{\fh}^*$, let 
$\Delta^R_\Lambda = \Delta^R\cap \hat{\Delta}^R_\Lambda $, and 
$\Delta^R_{\Lambda,+} = \Delta^R_+\cap \hat{\Delta}^R_\Lambda $.  

\begin{lemma}
  \label{lem:2.12}

Let $\Lambda$ be as in Theorem~\ref{th:2.2}.  Then
\begin{displaymath}
  \{ \alpha \in \Delta^R  | \,\, (\alpha |\beta )\in u \ZZ \}
     = \Delta^R_{\Lambda} \, .
\end{displaymath}

\end{lemma}

\begin{proof}

Let $\hat{\Delta}^R_{(u)} = \{ \gamma + nuK | \gamma \in \Delta^R
\, , \, n \in \ZZ \} \cup \{ un K | n \in \ZZ \, , \, n \neq 0
\}$, and note that  
\begin{equation}  
\label{eq:2.11}
  \hat{\Delta}^R_\Lambda = t_\beta (\hat{\Delta}^R_{(u)}) \, .
\end{equation} 
Let $\alpha \in \Delta^R$ be such that $(\alpha |\beta) = un$
for some $n \in \ZZ$.  Then $\alpha =t_\beta$ $(\alpha + unK) \in
t_\beta (\hat{\Delta}^R_{(u)}) = \hat{\Delta}^R_\Lambda$, which
proves that $\alpha \in \Delta^R_\Lambda$.

Conversely, if $\alpha \in \Delta^R_\Lambda$, then $\alpha \in
t_\beta (\hat{\Delta}^R_{(u)})$.  Hence $\alpha \in t_\beta
(\gamma +nu K)$ for some $n \in \ZZ$, $\gamma \in \Delta^R$, that
is, $\alpha = \gamma + (nu-(\beta |\gamma))K$.  Since $\alpha \in
\Delta^R$, it follows that $\alpha = \gamma$ and $(\beta |\gamma
) =nu$.

\end{proof}

\begin{theorem}
  \label{th:2.3}
\begin{enumerate}\begin{enumerate}
\item %%a
Let $M$ be a restricted $\hat{\fg}^R$-module, and assume that
$R_{\hat{\fg}^R} \ch_M$ converges to a holomorphic function on
$Y$.  Then $\ch_{H^R (M)}$ is identically zero iff there exists $\alpha \in
\Delta^R$, such that %(i) from (a) holds and $R_{\hat{\fg}^R}
%\ch_M =0$ on the hyperplane $\alpha =0$.

\begin{list}{}{}
\item (i)~~$\alpha |_{\fh^f} = 0$  ($\Rightarrow |\alpha (x) | \geq
  1$),

\item  (ii)~~$R_{\hat{\fg}^R} \ch_M $ vanishes on the hyperplane $\alpha =0$.
\end{list}

\item %%b
Let $\Lambda$ be as in Theorem~\ref{th:2.2}.  Then $\ch_{H^R (L(\Lambda))} =0$ 
iff there exists $\alpha \in \Delta^R$ such that:

\begin{list}{}{}
\item (i)~~$\alpha |_{\fh^f} = 0$  ($\Rightarrow |\alpha (x) | \geq
  1$),

\item  (ii)~~$(\alpha | \beta) \in u \ZZ$.
\end{list}

\item %%c
Let $\Lambda$ be an admissible weight. Then
$\ch_{H^R (L(\Lambda))}$ is not identically zero iff 
$\Delta^R_{\Lambda,+} \subset \Delta^R_+  \backslash \Delta^{R,f}$,
where $\Delta^{R,f}=\{\alpha \in \Delta^R| \,\,\, \alpha|_{\fh^f}=0\} $.
  
\end{enumerate}\end{enumerate}

\end{theorem}

\begin{proof}

(a) follows from the
character formula~(\ref{eq:2.2}). Condition~(i) implies that
$|\alpha (x) |\geq 1$ since $f$ is of principal type, hence $\alpha
|_{\fh^f} \neq 0$ for $\alpha \in \Delta^0$, and $\alpha
|_{\fh^f} \neq 0$ for $\alpha \in \Delta^{\pm 1/2}$ for all $f$ \cite{EK}.  

%It follows from Theorem~\ref{th:2.2} and from positivity of $a
%(\Lambda^0)$ that $\ch_{H^R (L (\Lambda))}=0$ iff
%
%\begin{displaymath}
%  \prod_{\alpha \in \Delta^R_+} \sin 
%      \frac{\pi (\alpha |z-\beta)}{u} \bigg{/} \prod_{\alpha \in
%  \Delta^0_+ \cup \Delta^{1/2}_+} \sin \pi (\alpha | z) = 0\, \hbox{for all} 
%\, z \in \fh^f.  
%\end{displaymath}
%
%Hence $\ch_{H^R (L (\Lambda))}=0$ iff $\sin \frac{\pi (\alpha
%  |z-\beta)}{u} |_{\fh^f} =0$ for some $\alpha \in \Delta$.  This
%is equivalent to (i) and (ii) of (a). 
%% (b)   (It is also easy to deduce (a)from (b).)

In order to prove (b), note that by Theorem~\ref{th:1.2}(a),
$R_{\hat{\fg}^R} \ch_{L (\Lambda )}$ vanishes on the hyperplane
$\alpha = 0$ iff $\alpha \in \hat{\Delta}^R_\Lambda$.  Now
Lemma~\ref{lem:2.12} and (a) imply~(b). Due to Theorem~\ref{th:1.2}(a), (c) is 
an equivalent form of (a) in the case when $M=L(
\Lambda)$,  where $\Lambda$ is an admissible weight.

\end{proof}

We call a weight  $\Lambda \in Pr^{k,R}$ {\it nondegenerate} for $W^k(\fg,f)$
if $\ch_{H^R(L(\Lambda))}\neq 0$, and denote the set of all such weights by 
$Prn^{k,R}(\fg, f)$.

%\begin{remark}
%  \label{rem:2.2}

%Theorem~\ref{th:2.3}(b) is consistent with the asymptotics given
%by Theorem~\ref{th:2.2}.

%\end{remark}

Given $\Lambda \in \hat{\fh}^*$ of level~$k$, denote by
$h_\Lambda$ the eigenvalue of $L^R_0$ on $v_\Lambda \otimes
\vac$.  Then we have the eigenvalue decomposition with respect to
$L^R_0$:
\begin{displaymath}
  H^R (L (\Lambda)) = \bigoplus_{j \in h_\Lambda + \ZZ_+}
     H^R (L (\Lambda))_j \, , 
\end{displaymath}
hence the decomposition of the Euler--Poincar\'e characters:
\begin{displaymath}
  \ch_{H^R (L (\Lambda))} =  e^{2\pi i kt} 
     \sum_{j \in h_\Lambda + \ZZ_+} q^j 
          \varphi_{\Lambda  ,j} (z) \, , 
\end{displaymath}
where $\ch_{H^R (L (\Lambda))_j} = e^{2\pi i kt} q^j
\varphi_{\Lambda ,j} (z)$.   It follows from
Theorem~\ref{th:1.2}(a) and formula~(\ref{eq:2.2}) that, if
$\Lambda$ is an admissible weight for $\hat{\fg}^R$, all
functions $\varphi_{\Lambda ,j} (z)$ are meromorphic on $\fh^f$,
and for $j=h_\Lambda$ we have $(z \in \fh^f)$:
\begin{displaymath}
  \varphi_{\Lambda ,h_{\Lambda}} (z) = 
    \frac{\sum_{w \in W_\Lambda^R} \epsilon (w) e^{2\pi i 
          (w (\Lambda  +\rho^R)-\rho^R |z)}}
         {\prod_{\alpha \in \Delta^{R,0}_+ \cup \Delta^{R,1/2}_+}
           (1-e^{-2\pi i (\alpha |z)})} \, .
\end{displaymath}

\begin{definition}
  \label{def:2.2}

We say that $\ch_{H^R (L (\Lambda))}$ is \emph{almost convergent}
if  $\lim_{z \to 0}\varphi_{\Lambda , h_{\Lambda}} (z)|_{\fh^f}$ exists
and is non-zero. 
%on $\fh^f$ and has a removable singularity at $z=0$.  
(The first condition is necessary for the convergence 
of $\ch_{H^R (L  (\Lambda))} (\tau ,0,0)$, and the second condition is
sufficient for its non-vanishing.)
\end{definition}

\begin{theorem}
  \label{th:2.4}

Let $\Lambda$ be an admissible weight for $\hat{\fg}^R$.
%, and let $\Delta^R_{\Lambda ,+} = \hat{\Delta}^R_\Lambda \cap \Delta^R_+
%= \{ \alpha \in \Delta^R_+ |\,\, (\Lambda  + \rho^R | 
%\alpha^\vee)\in \NN \}$.  
Then \\[1ex]
%
%%%%\begin{enumerate} \item[] \begin{enumerate}
%%%%\item %%a
\hspace*{6ex} (a)~~$\ch_{H^R (L (\Lambda))}$ is almost convergent iff
\begin{equation}
\label{eq:2.12}
 \{ \alpha |_{\fh^f} \,\,\,|\,  \alpha \in \Delta^R_{\Lambda,+} \}
    = \{ m_\alpha \alpha |_{\fh^f}\,\,\, |\, 
        \alpha \in \Delta^{R,0}_+ \cup \Delta^{R,1/2}_+ \} 
        \hbox{   (counting multiplicities), }
\end{equation}
\hspace*{10ex} for some non-zero $m_\alpha \in \QQ$ (it is easy to show that 
$m_\alpha > 0$).\

%\item %%b
% $ \varphi_{\Lambda ,h_{\Lambda}}(z)$ 
%does not vanish identically iff the left-hand
%side of (\ref{eq:2.12}) is contained in the right-hand side.

%%%%\item %%b
\hspace*{2ex} (b)~~$ \varphi_{\Lambda ,h_{\Lambda}}(z)$ is not identically zero iff
 $ \ch_{H^R(L(\Lambda))}$ is not identically zero.
 %
%%%%\end{enumerate} \item[] \end{enumerate}
\end{theorem}

\vspace{3ex}

\begin{proof}
We rewrite the formula for $\varphi_{\Lambda ,h_{\Lambda}} (z)$ as follows:
\begin{displaymath}
  \varphi_{\Lambda ,h_{\Lambda}} (z)  =
     \frac{\sum_{w \in W^R_\Lambda} \epsilon (w) 
          e^{2\pi i (w (\Lambda +\rho^R) - \rho^R |z)}}
        {\prod_{\alpha \in \Delta^R_{\Lambda ,+}}
          (1-e^{-2\pi i (\alpha |z)})} \,\,   
    \frac{\prod_{\alpha \in \Delta^R_{\Lambda ,+}} (1-e^{-2\pi i (\alpha |z)})}
          { \prod_{\alpha \in \Delta^{R,0}_+\cup \Delta^{R,1/2}_+}
            (1-e^{-2\pi i (\alpha | z)})}\, . 
\end{displaymath}
Since $W^R_{\Lambda}$ contains all reflections in $\alpha \in
\Delta^R_{\Lambda ,+}$ the first factor is holomorphic in $z \in
\fh^f$, and, by the usual argument, its limit as $z \to 0$ is
equal to $\prod_{\alpha \in \Delta^R_{\Lambda ,+}} (\Lambda +
\rho^R |\alpha) / (\rho^R |\alpha) \neq 0$. (a) follows.
(b) follows due to  Theorem \ref{th:2.3}(c).

\end{proof}

\begin{corollary}
  \label{cor:2.2}
A necessary condition of almost convergence of $\ch_{H^R
  (L(\Lambda))}$, where $\Lambda$ is an admissible weight, is:
\begin{displaymath}
  |\Delta^R_{\Lambda} | = | \Delta^0 \cup \Delta^{1/2}| \, .
\end{displaymath}
\end{corollary}

\vspace{4ex}

\subsection{Characters of the principal $W$-algebras.}~~
\label{sec:2.4}
In this section we consider in detail the case when $f$ is a
principal nilpotent element of $\fg$.  Recall that these are
elements of an adjoint orbit, whose closure contains all nilpotent
elements of $\fg$.  Recall that in this case $x=\rho^\vee$,
$\fh^f =0$, $\fg_0 =\fh$, hence $\Delta^0=\emptyset$, and
$\fg_{1/2}=0$, hence $\Delta^{1/2} =\emptyset$.  Hence $h_0 =0$
and $\Delta^{\rm new}_+ = \bar{w}_0 (\Delta_+)$, where
$\bar{w}_0$ is the longest element of $W$.  Thus, the quantum
Hamiltonian reduction, considered in this paper, coincides with
the `` $-$''-reduction of \cite{KRW}.  The corresponding
$W$-algebra, called \emph{principal}, will be denoted by $W^k
(\fg)$, and its simple quotient by $W_k (\fg)$.

By results of Arakawa \cite{A2}, a non-zero $W^k (\fg)$-module
$H^R (L (\Lambda))$ is irreducible and coincides with $H^{0,R} (L
(\Lambda))$. Hence the Euler--Poincar\'e character $\ch_{H^R (L
  (\Lambda))}$ is the character of this module.  Hence $\ch_{H^R
  (L (\Lambda))} =0$ iff the $W^k (\fg)$-module $H^R (L
(\Lambda))$ is zero.  It follows from Proposition~\ref{prop:1.2}
and Theorem~\ref{th:2.3}(c) that for a principal admissible
weight $\Lambda$ of $\hat{\fg}^R$ we have:
\begin{displaymath}
  H^R (L (\Lambda)) =0 \hbox{   iff   } (\Lambda |\alpha)\in\ZZ
     \hbox{   for some   } \alpha \in \Delta^\vee \, .
\end{displaymath}
The remaining principal admissible weights $\Lambda$ are  called
non-degenerate and the corresponding $W^k (\fg)$-modules $H^R (L
(\Lambda))$ are irreducible.  By Proposition~\ref{prop:1.3}, the
set of such $\Lambda$, denoted by $Prn^{k,R}$, is non-empty iff
$k$ satisfies conditions~(\ref{eq:0.2}):
\begin{displaymath}
  k+h^\vee = p /u \, , \hbox{   where   } p,u \in \NN \, , \, 
    (p,u) =1 \, , \, (\ell ,u)=1 \, , \,  
       p \geq h^\vee \, , \, u \geq h \, .
\end{displaymath}

It follows from \cite{FKW} that the irreducible $W^k
(\fg)$-modules $H^R (L (\Lambda))$ and $H^R (L (\Lambda'))$ with
$\Lambda ,\Lambda' \in Prn^{k,R} \cong Prn^k$ are isomorphic iff
$\varphi (\Lambda ) = \varphi (\Lambda')$, where
\begin{displaymath}
  \varphi : Prn^k \to I_{p,u} = (\hat{P}^{p-h^{\vee}}_+ \times
  \hat{P}^{\vee u-h}_+)/\tilde{W}_+ 
\end{displaymath}
is the surjective map, defined in Proposition~\ref{prop:1.3}.

Formula (\ref{eq:2.1}) for the central charge of these $W^k
(\fg)$-modules becomes:
\begin{displaymath}
  c(k,\fg) =r-\frac{12 |u\rho -p\rho^\vee |^2}{p u} \, .
\end{displaymath}

The normalized character of the irreducible $W^k (\fg)$-module
$H^R (L (\Lambda))$, parameterized by $\varphi (\Lambda ) =
(\lambda ,\mu) \in I_{p,u}$, is given by the following formula
\cite{FKW}:
\begin{displaymath}
  \chi_{\lambda ,\mu} (\tau) = \eta (\tau)^{-r}\sum_{w \in \hat{W}}
     \epsilon (w) q^{\frac{pu}{2} | 
         \frac{w (\lambda + \hat{\rho})}{p} -
            \frac{\mu + \hat{\rho}^{\vee}}{u}|^2} \, , 
\end{displaymath}
the minimal eigenvalue of $L^R_0$ being 
\begin{displaymath}
  h_{\lambda ,\mu} =\frac{1}{2pu} (| u (\bar{\lambda} +\rho) -
    p (\bar{\mu} +\rho^\vee)|^2 - |u \rho -p\rho^\vee|^2)\, .
\end{displaymath}

It follows from Theorem~\ref{th:2.2} that, as $\tau \downarrow
0$,
\begin{displaymath}
  \chi_{\lambda ,\mu} (\tau) \sim (up)^{-r/2}
    |P/Q^\vee |^{-1/2} \prod_{\alpha \in \Delta_+}4\sin
    \frac{(\lambda +\rho |\alpha)}{p} \sin 
      \frac{(\mu +\rho^\vee|\alpha)}{u} \, .
\end{displaymath}
A formula for modular transformations can be found in \cite{FKW}.

\vspace{1ex}

\textbf{Conjecture PA.}~~If $\Lambda \in Prn^{k,R}$, then the $W^k
(\fg)$-module $H^R (L (\Lambda))$ is actually a $W_k
(\fg)$-module (defined, up to isomorphism, by $\varphi (\Lambda)$), 
and these are all irreducible $W_k (\fg)$-modules.

\vspace{2ex}

\textbf{Conjecture PB.}~~The vertex algebra $W_k (\fg)$ is
semisimple iff $k$ satisfies (\ref{eq:0.2}).

\vspace{2ex}

\subsection{Conjectures.}~~
\label{sec:2.5}
Consider the space of meromorphic functions in the domain $Y$ 
%$(\tau, z,t)$
%where $\tau \in \CC$, $z \in \fh$, $t \in \CC$ 
with the
following action  of  $SL_2 (\ZZ)$:
\begin{displaymath}
  (g \cdot \chi) (\tau ,z,t) = \chi \left( 
    \frac{a\tau +b}{c\tau +d} \, , \, \frac{z}{c\tau +d}\, , \,
       t-\frac{c(z|z)}{2(c\tau +d)}\right) \, , \, 
     g= \binom{a \,\, b}{c \,\, d} \in SL_2 (\ZZ)\, .
\end{displaymath}
We call a function $\chi$ from this space {\it modular invariant}
%Given such a function $\chi$, denote by $\sum_\chi$ the
if the $\CC$-span of the set of all functions $\{ g \cdot \chi \}_{g \in SL_2
(\ZZ)}$ is finite-dimensional.  
It follows from Proposition~\ref{prop:1.1} and the
results of 
\cite{KW1}, Theorems 3.6 and 3.7, and \cite{KW2}, Remark 4.3(a), that 
%that if $kD$ is an admissible weight, then
the function $\chi_{L(kD)}$ is modular invariant if $k$ is of the form
(\ref{eq:0.1}).
%
%\begin{equation}
%  \label{eq:2.13}
%  k=\frac{p}{u} -h^\vee \, , \hbox{   where   } p\, , \, 
%  u \in \NN \, , \, (p,u)=1 \, , \, (\ell ,u) =1 \, , \, 
%  p \geq h^\vee\, .
%\end{equation}
%
Moreover, if $k$ is of the form 
(\ref{eq:0.1}) with $(u,\ell)=1$, then
$\chi_{L(\lambda)}$ is modular invariant for all $\lambda \in Pr^k$, simply
because $|Pr^k \mod \CC K|< \infty$ and the $\CC$-span of 
$\{\chi_{L(\lambda)}\ \, |\lambda \in Pr^k\}$ is $SL_2
(\ZZ)$-invariant (cf.~Theorems~\ref{th:1.1}--\ref{th:1.3}).

\vspace{2ex}

\textbf{Conjecture A.}~~If $\chi_{L(kD)}$ is modular invariant, 
%and $k+h^\vee \neq 0$, 
then $k$ is of the form (\ref{eq:0.1}).

\vspace{2ex}

It has been established recently \cite{GK2} that the character of any 
highest weight 
$\hat{\fg}$-module of non-critical level $k \neq -h^\vee$ is a meromorphic 
function in 
the domain $Y$. It follows from \cite{GK1} that for non-critical $k$,
$\chi_{L(kD)}$ 
can be modular invariant only for $k=-h^\vee + p/u \ell$, where
$(p,u)=1, u \geq 1, p \geq 2$. Thus, the first unknown case is $\fg =sl_3$,
$k=-1$.

\vspace{2ex}

\textbf{Conjecture B.}

\alphaparenlist

\begin{enumerate}

\item %%a
If $M$ is an irreducible highest weight $\hat{\fg}^R$-module 
of level $k$, then $H^R(M)=H^{R}_0(M)$, and this is either an 
irreducible $\sigma_R$-twisted $W^k(\fg,f)$-module, or zero.

%\item %%b
%Let $L (\Lambda)$ be an admissible $\hat{\fg}^R$-module, such
%that $\ch_{H^R (L(\Lambda))}$ is almost convergent, and let $j_\Lambda$
%be the number of negative $m_\alpha$ in (\ref{eq:2.12}).  Then
%\romanlistii
%\begin{enumerate}
%\item %%(i) 
%$H^{j,R} (L (\Lambda))$ is non-zero only if $0 \leq j \leq
%j_\Lambda$,
%\item %%(ii)
%$(-1)^{j_\Lambda} \ch_{H^R (L (\Lambda))}$ is either zero, or
%coincides with the character of an irreducible $\sigma_R$-twisted
%$W^k (\fg ,f)$-module $H^R (L (\Lambda'))$ 
%for some admissible $\Lambda'$ of level $k$.
%\end{enumerate}

%\item %%b
%$H:M \to H^R (M)$ is an exact functor from the category of
%restricted level $k$ $\hat{\fg}^R$-modules to the category of
%$\sigma_R$-twisted $W^k (\fg ,f)$-modules. 

\item %%b
Let $W^{R,f}=\{w \in \hat{W}^R|\,\, w|_{\CC D^R +\fh^f}=Id \}$.
Suppose that $H^R(L(\Lambda)) 
\neq 0$. Then the $\sigma_R$-twisted $W^k(\fg,f)$-modules $H^R(L(\Lambda))$ 
and $H^R(L(\Lambda'))$
are isomorphic iff $\Lambda'=y.\Lambda$ for some $y \in W^{R,f}$. 

\end{enumerate}

%Let $L(\Lambda)$ be an irreducible highest 
%weight $\hat{\fg}^R$-module of level $k$. 

%% \neq -h^\vee$
%Then $H^R(L(\Lambda))$ 
%%$(-1)^{j_\Lambda}\ch_{H^R(L(\Lambda))}$ 
%is either $0$, or 
%%the character of an 
%an irreducible $\sigma_R$-twisted $W^k (\fg,f)$-module.
%%, and $h_\Lambda$ is the lowest eigenvalue of $L_0^R$ in this module. 
%Furthermore, $H^{j,R}(L(\Lambda))$ is non-zero
%only if $j = j_\Lambda$, where $j_\Lambda$ is the number of negative 
%$m_{\alpha}$ in (\ref{eq:2.12}). 

\vspace{2ex}
In the case when $f$ is a principal (resp. ``minimal'' nilpotent), 
Conjecture ~B(a) was stated in 
\cite{FKW} and \cite{KRW}, and was proved in 
\cite{A2} and \cite{A1} respectively. 

The specific choice of $\hat{\Delta}^R_+$
(depending on our specific choice of $\Delta_+^{\new}$) is very important.
Indeed, if the property, given by \ref{lem:2.2}(d) (guaranteed by our choice
of $\Delta_+^{\new}$), doesn't hold, then Conjecture B(a) fails, for example,
when $\fg=sp_4$ and $f$ is a root 
vector, attached to a short root. 
%In this case $j_\Lambda=1$ for some admissible $\Lambda$.

It follows from Conjecture~B(a) and Theorem~\ref{th:2.4} that for an admissible
$\Lambda$, 
$\ch_{H^R(L(\Lambda))}(\tau,0,0)$ 
converges for all $\tau \in \CC^+$
and does not vanish iff $\ch_{H^R(L(\Lambda))}$ is almost convergent.
 
\begin{remark}
  \label{rem:2.2}
%  Conjecture~B is consistent with 
%Theorems~\ref{th:2.3} and
%  \ref{th:2.4}.  Indeed, suppose that 
%$H^{\odd ,R} (L(\Lambda))=0$.  
%Then all coefficients of the series ${\rm
%    ch}_{H^R (L (\Lambda))} (\tau ,z,t)$ are non-negative
%  integers, hence all terms of this series are non-negative real
%  if $t$, $\tau$ and $z$ are purely imaginary.  Then
%  Theorems~\ref{th:2.3} and \ref{th:2.4} follow immediately from
%  Theorem~\ref{th:2.2}.  
It follows from
  Theorem~\ref{th:2.2} that if $\Lambda \in Pr^{k,R}$ and
  $k=-h^\vee +\frac{p}{u}$ is such that $pu \dim \fg^f < h^\vee \dim
  \fg$, then $H^R (L (\Lambda))=0$.  Moreover, if  $pu \dim\fg^f =
  h^\vee \dim \fg$ and $\Lambda \in Prn^{k,R}(\fg ,f)$,
  then $0 < \dim H^R (L (\Lambda))<\infty$, hence $c (\fg ,f,k)=0$.
%, and also we get from (\ref{eq:2.6}) that $p=h^\vee$. 
Thus, we obtain the following 
  generalization of the ``strange'' formula:
\begin{displaymath}
  | \rho - \frac{p}{u}x |^2 = \frac{1}{12}\, \frac{p}{u}
    \left(\dim \fg_0 -\frac12 \dim \fg_{1/2}\right)\, .
\end{displaymath}
This formula holds for each exceptional pair $f$, $u$ 
(see Definition \ref{def:2.3} below), and an integer $p \geq h^\vee$, 
$(u,p)=1$, such that $pu \dim \fg^f =h^\vee \dim \fg$.
\end{remark}

%Another corollary of Conjecture~B is that $\ch_{H^R(L(\Lambda))}$
%is almost convergent iff $H^R (L (\Lambda))$ is non-zero and all
%eigenspaces of $L_0^R$ in it are finite-dimensional.

\vspace{2ex}

\textbf{Conjecture C.}~~The set $Prn^{k,R}(\fg ,f)$ is non-empty iff k 
is of the form
(\ref{eq:0.1}), where $u$ satisfies 
\begin{equation}
\label{eq:2.13}
   u>(\theta|x)\, .
\end{equation}
%If $M$ is a principal admissible
%$\hat{\fg}^R$-module of level~$k$, then 
%If $k$ satisfies (\ref{eq:0.1}) and (\ref{eq:2.13}),
%, the $W^k(\fg,f)$-module $H^R (L(\Lambda))$, where 
%then for each
%$\Lambda \in Prn^{k,R}(\fg,f)$ there exists an irreducible 
%$\sigma_R$- twisted 
%$W_k (\fg,f)$-module with the normalized character
%$(-1)^{j_\Lambda}\chi_{H^R(L(\Lambda))}$, which completely 
%determines this module
%. These are all irreducible $\sigma_R$- twisted 
%$W_k (\fg,f)$-modules, up to isomorphism.
%, provided that $k$ is of the form (\ref{eq:2.13}).
%Also the normalized characters $\chi_{H^R(L(\Lambda))}(\tau,z,t)$ uniquely 
%determine the irreducible  $W_k (\fg ,f)$-modules $H^R(L(\Lambda))$.

\begin{definition}
  \label{def:2.3}

A pair $(k,f)$, where $k \in \CC$ and $f$ is a nilpotent element
of $\fg$, is called exceptional if the Euler--Poincar\'e
character $\ch_{H^R (L (\Lambda))}$ of the $W^k (\fg ,f)$-module
$H^R (L (\Lambda))$ is almost convergent for some principal
admissible $\hat{\fg}^R$-module of level~$k$, and is either 0 or is almost 
convergent for all principal admissible $\hat{\fg}^R$-modules of level $k$.
In this case $f$ is called an exceptional nilpotent of $\fg$
and its adjoint orbit is called an exceptional nilpotent orbit,
and $k$ is called an exceptional level for $f$ and its denominator $u$
an exceptional denominator.

\end{definition}

%The necessary and sufficient conditions of almost convergence of
%$\ch_{H^R (L (\Lambda))}$ are given by Theorem~\ref{th:2.4}.
%Note that if Conjecture~B holds, then all eigenspaces of $L^R_0$
%in $H^R (L (\Lambda))$ are finite-dimensional iff $\ch_{H^R
%  (L(\Lambda))}$ is almost convergent, and this happens iff the
%specialized normalized character 
%$tr_{H^R (L (\Lambda))} q^{L^R_0- c (\fg ,f,k)/24}= 
%\chi_{H^R (L(\Lambda))} (\tau ,0,0)$
%converges in the upper half-plane.

Recall that convergence of characters $tr_M q^{L_0 -c/24}$
of all modules $M$ over a vertex algebra~$V$ is a necessary
condition of rationality of $V$.  
We conclude from the above discussion that, provided that
Conjectures ~B and ~C hold, the vertex algebra 
$W_k (\fg ,f)$ is rational iff the
pair $(k,f)$ is exceptional.

\vspace{2ex}

\textbf{Conjecture D.}~~The vertex algebra $W_k (\fg ,f)$ is semisimple 
iff the pair $(k,f)$ is exceptional.   
%Moreover, in
%this case the normalized characters $\chi_{H^R(M)} (\tau ,z,t)$ of all
%irreducible  $W_k (\fg ,f)$-modules $H^R(M)$ converge to holomorphic
%functions in the domain  $\CC^+ \times \fh^f \times \CC$, and the
%coefficient $A_{\beta}(z)$ of their asymptotics, as 
%$\tau \downarrow 0$, is holomorphic and does not vanish at $z \in
%\fh^f$.  

\vspace{2ex}

\begin{remark}
  \label{rem:2.3}

Note that $W^k (\fg ,0)$ (resp.~$W_k (\fg,0)$) is 
isomorphic to the universal (resp.~simple) affine vertex algebra
$V^k (\fg)$ (resp.~$V_k (\fg)$), and $H^R (M) \cong M$ in this
case $(f=0)$.  Hence the pair $(k,0)$ is exceptional iff $k \in
\ZZ_+$.  Since $Pr^k=\hat{P}^k_+$ in this case and all $V_k
(\fg)$-modules with $k \in \ZZ_+$ are $\{ L
(\Lambda)\}_{\Lambda\in \hat{P}^k_+ \mod \CC K}$ \cite{FZ}, we
conclude that Conjecture~C (resp.~D)  holds if $f=0$ and $k \in
\ZZ_+$ (resp.~$f=0$).  However, it is unknown whether
Conjecture~C holds if $k$ is of the form (\ref{eq:0.1}) with $u$ satisfying  
(\ref{eq:2.13}) and $k \notin \ZZ_+$, except for the case 
$\fg = s\ell_2$ \cite{AM}.
E.~Frenkel pointed out that the first part of Conjecture~C
follows from the special case $f=0$.

\end{remark}

\begin{remark}
  \label{rem:2.4}
In the case when $f$ is a principal nilpotent of a simply laced
simple Lie algebra, the normalized characters $\chi_{\lambda
  ,\mu} (\tau)$ of $W^k (\fg ,f)$ coincide with those of the
centralizer of $V_k (\fg)$ in the vertex algebra $V_1 (\fg)
\otimes V_{k-1} (\fg)$ \cite{KW2}.  Conjecture~C in the case of
simply laced $\fg$ and principal nilpotent $f$ would follow if
the latter vertex algebra were isomorphic to $W_k (\fg ,f)$.
However, apparently this isomorphism is an open problem.
\end{remark}

\vspace{2ex}

\subsection{Description of exceptional pairs, assuming 
positivity.}~~
\label{sec:2.6}  
The following theorem provides an easy way to check almost
convergence under the assumption that all the coefficients of
$\ch_{H^R(L(\Lambda))}$ are non-negative
(cf.~Conjecture~B(a)).

\begin{theorem}
  \label{th:2.5}
Let $\Lambda$ be an admissible weight for $\hat{\fg}^R$ and
assume that all the coefficients of $\ch_{H^R(L(\Lambda))}$ 
are non-negative.
%$H^{j,R} (L (\Lambda))=0$ for $j$ odd.
Then $\ch_{H^R (L (\Lambda))}$ is almost convergent iff
\begin{equation}
  \label{eq:2.14}
  \Delta^R_\Lambda \subset \Delta^R \backslash \Delta^{R,f} 
     \hbox{  and   } |\Delta^R_\Lambda | 
         = | \Delta^0 \cup \Delta^{1/2}| \, .
\end{equation}
Consequently, if these conditions hold, then the eigenspace of
$L^R_0$ with minimal eigenvalue is finite-dimensional.

\end{theorem}

\begin{proof}
In view of Theorems~\ref{th:2.3} and \ref{th:2.4}, and
Corollary~\ref{cor:2.2}, it remains to prove that conditions
(\ref{eq:2.14}) imply almost convergence.  In view of the proof
of Theorem~\ref{th:2.4}, it suffices to prove the following
simple lemma on rational functions.

\begin{lemma}
  \label{lem:2.13}

Let $R (z_1 , \ldots , z_n) =
  f (z_1, \ldots ,z_n) / 
     \prod_{(k_1,\ldots ,k_n) \in \ZZ^n_+ \backslash \{ 0 \}}
       (1-z^{k_1}_1 \ldots z^{k_n}_n)^{m(k_1,\ldots ,k_n)}$ ,
where $f$ is a polynomial, $m (k_1,\ldots ,k_n) \in \ZZ_+$ and
all but a finite number of them is~$0$.
Suppose that all coefficients of the Taylor expansion of $R (z_1
,\ldots ,z_n)$ at $z_1 = \cdots = z_n =0$ are non-negative, and
$R(z^{s_1} ,\ldots ,z^{s_n})$ has a removable singularity at
$z=1$, for some positive integers $s_i$.  Then $R (z_1,\ldots
,z_n)$ is a polynomial.

\end{lemma}

This lemma follows from a similar lemma in one variable.
\end{proof}

\begin{lemma}
  \label{lem:2.14}
Let $R (z) = f(z) /\prod^N_{j=1} (1-z^j)^{m_j}$, where $f(z)$ is a
polynomial in $z$ and $m_j \in \ZZ_+$.  Suppose that all
coefficients of the Taylor expansion of $R(z)$ at $z=0$ are
non-negative and that $R(z)$ has a removable singularity at
$z=1$.  Then $R(z)$ is a polynomial.

\end{lemma}

\begin{proof}
Let $M=\sum_j m_j$.  We have:  $f(z) = (1-z)\,{}^Mf_1 (z)$, and
$\prod^N_{j=1} (1-z^j)^{m_j} = (1-z)\, {}^Mf_2 (z)$, where 
$f_1(z)$ and $f_2(z)$
are polynomials and the coefficients of $f_2(z)$ are
non-negative.  Hence $f_1 (z) = R(z) f_2 (z)$.  Since all
coefficients of the Taylor expansion at $z=0$ of $R(z)$ and
$f_2(z)$ are non-negative and $f_1 (z)$ is a polynomial, it
follows that $R(z)$ is a polynomial.

\end{proof}

Using Theorem~\ref{th:2.5}, we can give a description of
exceptional pairs $(k,f)$ in terms of the Lie algebra $\fg$ and
its adjoint group~$G$.  For this we need some lemmas.

\begin{lemma}
  \label{lem:2.15}

Let $\Lambda$ be an admissible weight of $\hat{\fg}$ of level
$k=v/u$, where $u,v \in \ZZ$, $u>0, (u,v)=1$.  Suppose that

\begin{list}{}{}
\item (i)~~$(u,\ell)=1$,

\item (ii)~~for each $\alpha \in \Delta^\vee$ there exists $m \in
  \ZZ$, such that $mK+\alpha \in \hat{\Delta}_\Lambda^{\vee,\re}$.  
\end{list}
Then $\Lambda$ is a principal admissible weight.

\end{lemma}

\begin{proof}
It follows from the classification of admissible weights in
\cite{KW1} that if $\fg$ is simply laced (i.e.,~$\ell =1$), then
condition~(ii) implies that $\Lambda$ is principal admissible.

In the non-simply laced cases  condition~(ii) leaves only the
following possibilities for $\hat{\prod}^\vee_\Lambda$ to be
non-isomorphic to $\hat{\prod}^\vee$ \cite{KW1} (we use the
numeration of simple roots of $\hat{\fg}$ from the tables of
\cite{K1}, Chapter~4):
\begin{displaymath}
  \hat{\prod}^\vee_\Lambda = \sigma_{\alpha^\vee_0} 
    \hat{\prod}^\vee  \hbox{  if  } \fg = B_r \, , \, 
  C_r \, , \, F_4 \hbox{  or  } G_2 \, ; \, 
  \hat{\prod}^\vee_\Lambda =\sigma_{\alpha^\vee_4} 
    \hat{\prod}^\vee  \hbox{  if  } \fg =F_4 \, .
\end{displaymath}
But the denominator $u$ of $k$ can be computed by the formula
$\sum^r_{i=0} a^\vee_i \gamma_i = uK$, where
$\hat{\prod}^\vee_\Lambda = \{ \gamma_0 , \ldots , \gamma_r \}$,
and we see by a case-wise inspection that in all cases $u$ is
divisible by~$\ell$, a contradiction with~(i).

\end{proof}

\begin{lemma}
  \label{lem:2.16}

Let $\Lambda \in \fh$ and let $u$ be a positive integer, satisfying the
following conditions:

\begin{list}{}{}

\item (i)~~$(\Lambda + \rho |\alpha) \notin -\ZZ_+$ for all
  $\alpha \in \Delta^\vee_+$,

\item (ii)~~$(\Lambda |\alpha) \in \frac1u \ZZ$ for all $\alpha
  \in \Delta^\vee$,

\item (iii)~~$(u,\ell) =1$.

\end{list}

Then for a sufficiently large integer~$p$, coprime to~$u$,
$\hat{\Lambda} = (\frac{p}{u}-h^\vee) D+\Lambda$ is a principal
admissible weight (of level $k=-h^\vee +p/u$).

\end{lemma}

\begin{proof}

Let $p$ be a positive integer, coprime to $u$, such that $p/u >
(\Lambda + \rho|\alpha)$ for all $\alpha \in \Delta^\vee$.  Then,
clearly, $\hat{\Lambda}$ is an admissible weight.  Furthermore,
$(\hat{\Lambda}+\hat{\rho} | j\ell K+\alpha) = j\ell p/u +
(\Lambda +\rho |\alpha)$ and $(\Lambda + \rho |\alpha) \in
\frac1u \ZZ$ for each $\alpha \in \Delta^\vee$.  Since
$(u,\ell p)=1$, for each $\alpha \in \Delta^\vee$ there  exist
$m_\alpha$, such that $m_\alpha \ell p/u  + (\Lambda +\rho
|\alpha) \in \ZZ$, hence $(\hat{\Lambda}+\hat{\rho} | m_\alpha \ell K
+\alpha) \in \ZZ$.  Therefore, by Lemma~\ref{lem:2.15},
$\hat{\Lambda}$ is a principal admissible weight.

\end{proof}

\begin{lemma}
  \label{lem:2.17}

Let $\Lambda \in \fh$ be such that $e^{2\pi i\Lambda}$ is an
element of order $u$ in the adjoint group~$G$ of the Lie algebra
$\fg$, where $u$ is a positive integer, coprime to $\ell$.  Then

\begin{list}{}{}
\item (a)~~$\Delta^\vee_\Lambda =\{ \alpha^\vee := 2\alpha
  /(\alpha |\alpha) |\,\, \alpha \in \Delta_\Lambda \}$, where
  $\Delta_\Lambda =\{ \alpha \in \Delta | (\Lambda | \alpha) \in
  \ZZ \}$.

\item (b)~~For a sufficiently large integer $p$, coprime to $u$,
  there exists a principle admissible weight $\hat{\Lambda}$ of
  level $k=-h^\vee + p/u$, such that $\hat{\Lambda} |_{\fh}$ is
  conjugate to $\Lambda$ by the Weyl group $W_\Lambda$ of $Z_\fg
  (e^{2\pi i\Lambda})$ (the centralizer of $e^{2\pi i \Lambda}$
  in $\fg$)  and $Z_\fg (e^{2\pi i \hat{\Lambda}|_{\fh}}) =
  Z_{\fg} (e^{2\pi i \Lambda})$.
\end{list}

\end{lemma}

\begin{proof}

Since $e^{2\pi i \Lambda}$ has order $u$, we have:  $\alpha
\in \Delta_\Lambda$ implies $(\Lambda |\alpha) \in \ZZ$; $\alpha
\in \Delta \backslash \Delta_\Lambda$ implies $(\Lambda |\alpha)
\in \frac1u \ZZ \backslash \ZZ$.  Note also that $\alpha^\vee = \alpha$
 (resp. $\ell \alpha$) if $\alpha$ is a long (resp. short) root
 and $(\Lambda |\alpha) \in \frac1u \ZZ \backslash \ZZ$ iff $\ell
 (\Lambda |\alpha) \in \frac1u \ZZ \backslash \ZZ$, since
 $(u,\ell)=1$.  These remarks prove (a).

In order to prove (b), note that $Z_{\fg} (e^{2\pi i \Lambda}) =\fh
+ \sum_{\alpha \in \Delta_\Lambda} \fg_\alpha$ , and $\Lambda$ is
an integral weight of the semisimple part of $Z_{\fg} (e^{2\pi i
  \Lambda})$ (due to (a)).  Hence there exists $w \in W_\Lambda$,
such that $\Lambda' = w (\Lambda)$ has the property that
$(\Lambda' |\alpha) \in \ZZ_+$ for all $\alpha \in
\Delta_{\Lambda} \cap \Delta_+$.  Hence we have: 
  $(\Lambda' + \rho |\alpha) \in \NN$ for all $\alpha 
\in \Delta^\vee_\Lambda \cap \Delta^\vee_+$, and
$(\Lambda' |\alpha )\in \frac1u \ZZ \backslash \ZZ$
for all $\alpha \in \Delta^\vee \backslash
              \Delta^\vee_\Lambda$.
Then, by Lemma~\ref{lem:2.16}, there exists a positive integer
$p$, coprime to $u$, such that $\hat{\Lambda}' := (-h^\vee + p/u)
D +\Lambda'$ is a principal admissible weight of level $k=-h^\vee
+ p/u$.  Since $\Delta_{\Lambda'} =w(\Delta_\Lambda) =
\Delta_\Lambda$, because $w \in W_\Lambda$, we conclude that
$Z_{\fg} (e^{2\pi i \Lambda'})=\fh + \sum_{\alpha \in
  \Delta_{\Lambda'}} \fg_\alpha = Z_{\fg} (e^{2\pi i \Lambda})$, proving~(b).

\end{proof}

Given a nilpotent element $f$ of the Lie algebra $\fg$ and a
positive integer $u$, coprime to~$\ell$, let
\begin{displaymath}
  S_{u,f} =\{ s \in G | s^u=1 \hbox{  and  } Z_{\fg} (s) \cap
     Z_{\fg} (\fh^f)=\fh \}\, .
\end{displaymath}

\begin{theorem}
  \label{th:2.6}

Assume that all the coefficients of $ch_{H^{R} (L (\Lambda))}$
are non-negative for all
principal admissible weights for $\hat{\fg}^R$ of level
$k=-h^\vee + p/u$, where $u \geq 1$, $p \geq h^\vee$, $(u,p\ell)
=1$.  Let $f$ be a nilpotent element of $\fg$ of principal type.
Then the pair $(k,f)$ is exceptional iff:

\begin{list}{}{}
\item (i)~~$\dim Z_{\fg} (s) \geq \dim \fg^f$ for all $s \in
  S_{u,f}$;

\item (ii)~~$\dim Z_{\fg} (s) = \dim \fg^f$ for some $s\in S_{u,f}$.
\end{list}

\end{theorem}

\begin{proof}

Note that condition~(c) of Theorem~\ref{th:2.3} of non-vanishing
of $\ch_{H^R (L(\Lambda))}$ is equivalent to the condition that
the element $s=e^{2\pi i \Lambda |_\fh}$ lies in $S_{u,f}$ and
satisfies (i).  Also, by Theorem~\ref{th:2.5}, $\ch_{H^R (L
  (\Lambda))}$ is almost convergent iff $s$ lies in $S_{u,f}$ and
satisfies (ii).  Hence conditions (i) and (ii) are necessary for
the pair $(k,f)$ to be exceptional.  Due to Lemma~\ref{lem:2.17},
these conditions are also sufficient if $p$ is large enough.  But it 
is clear from Theorems~\ref{th:2.3} and \ref{th:2.4} that the
pairs $(-h^\vee + p/u , f)$ are exceptional for all $p \geq
h^\vee$ iff one of them is exceptional for some $p \geq h^\vee$.

\end{proof}

\section{Exceptional pairs for $\fg = s\ell_n$.}
\label{sec:3}

\subsection{Sheets in $\fg = s \ell_n$.}
\label{sec:3.1}~~Recall that the adjoint nilpotent orbits of
$\fg =\sl_n$  are parameterized by partitions of $n$, and that
the closure of the nilpotent orbit, corresponding to the
partition $m_1 \geq m_2 \geq \cdots$, contains the nilpotent
orbit, corresponding to the partition $n_1 \geq n_2 \geq \cdots
$, iff $m_1 \geq n_1$, $m_1 +m_2 \geq n_1 +n_2, \ldots $
\cite{CM}.  Note also that all nilpotents of $\sl_n$ are of
principal type.

In order to classify exceptional pairs we use
the theory of sheets.  Recall that a sheet in a simple Lie
algebra $\fg$ is an irreducible component of the algebraic
variety in $\fg$, consisting of all adjoint orbits of fixed
dimension.  In the case of $\fg = s\ell_n$ the description of
sheets is especially simple.

\begin{proposition}\cite{Kr}
  \label{prop:3.1}
Let $f$ be a nilpotent element of $s\ell_n$, let  $m_1 \geq m_2 \geq \cdots \geq m_s >0$ be the
corresponding partition of $n$, let $m=m_1$ and let $m'_1 \geq m'_2 \geq \cdots
\geq m'_m >0$ be the dual partition.

\alphaparenlist
\begin{enumerate}
\item %%a
The element $f$ is contained in a unique sheet, which we
denote by $\Sh_f$.

\item %%b
All the semisimple elements in $\Sh_f$ are those diagonalizable
matrices in $\sl_n$, which have $m$ distinct eigenvalues of
multiplicities $m'_1, m'_2, \ldots ,m'_m$.  We denote the set
of all semisimple elements in $\Sh_f$ by $\Sh^0_f$.

\item %%c
The rank of the semisimple Lie algebra $[\fg^h,\fg^h]$, where $h
\in \Sh^0_f$ and $\fg^h$ is the centralizer of $h$, is equal to $n-m$.

\end{enumerate}

\end{proposition}

\begin{proposition}
    \label{prop:3.2}

Let $\fh$ be the set of all diagonal matrices in $\fg= \sl_n$, let
$f$ be as in Proposition~\ref{prop:3.1} and assume that the
centralizer $\fh^f$ of $f$ in $\fh$ is the Cartan subalgebra of
the reductive part of $\fg^f$.  Let
\begin{displaymath}
  \Omega_f = \{ h \in \fh | \, \alpha (h) \neq 0 
  \hbox{   for any  root   } \alpha \in \fh^* \hbox{  of  } \sl_n\, , \hbox{  such that  } 
        \alpha |_{\fh^f} =0 \} \, .
\end{displaymath}

\alphaparenlist
\begin{enumerate}
\item %%a
If $h \in \Omega_f$, then rank $[\fg^h,\fg^h] \leq n-m$ and $\dim
\fg^h \leq \dim \fg^f$.

\item %%b
If $h \in \Omega_f$ and rank $[\fg^h,\fg^h] =n-m$, then $h \in
\Sh_f$.  Moreover, $\Sh^0_f = W  (\Omega_f \cap \Sh_f)$, where
$W$ is the Weyl group.

\end{enumerate}
\end{proposition}

\begin{proof}
We fill the boxes of the Young diagram of the partition $m_1 \geq
m_2 \geq \cdots \geq m_s$ by the eigenvalues of $h$ (in
$\CC^n$).  Then $h \in \Omega_f$ iff the eigenvalues in each row
are distinct.  Moreover, $h \in \Omega_f \cap \Sh_f$ iff, in
addition, all eigenvalues in each column are equal.  This
proves~(b), due to Proposition~\ref{prop:3.1}(c).  Now (a)
follows since making eigenvalues of $h$ in a column unequal and
keeping $h$ in $\Omega_f$ can only decrease the rank of
$[\fg^h,\fg^h]$ and the dimension of $\fg^h$.

\end{proof}

Let $\Delta \subset \fh^*$ be the set of roots of $\fg = \sl_n$,
and let $f \in \fg$ be a nilpotent element as in
Proposition~\ref{prop:3.1}.  Let $\Delta^f = \{ \alpha \in \Delta
|\,\,\,  \alpha |_{\fh^f} =0 \}$.  We call $\Phi \subset \Delta$ a root
subsystem if $\alpha \in \Phi$ implies $-\alpha \in \Phi$ and
$\alpha ,\beta \in \Phi$, $\alpha + \beta \in \Delta$ implies
$\alpha +\beta \in \Phi$.  The dimension of the $\CC$-span of
$\Phi$ in $\fh^*$ is called the rank of $\Phi$ .  Note that for
each $h \in \fh$ the set of roots of $\fg^h$ is a root subsystem,
and its rank equals rank $[\fg^h,\fg^h]$; moreover, all root
subsystems of $\fg = \sl_n$ are thus obtained.  
Hence the above propositions can
be translated in the language of root subsystems.  Given a positive
integer $m$, such that $m\leq n$, denote by $f_m$ the nilpotent
element, corresponding to the partition of $n$ of the form $m=m=
\cdots = m>s \geq 0$.

\begin{proposition}
    \label{prop:3.3}

\alphaparenlist
\begin{enumerate}{}

\item %%a
If $\Phi \subset \Delta \backslash \Delta^f$ is a root subsystem,
then
\begin{equation}
  \label{eq:3.1}
  {\rm rank} \, \Phi \leq n-m \hbox{  and  } |\Phi| \leq
  |\Delta^0 \cup \Delta^{1/2}|\, .
\end{equation}

\item %%b
There exists a root subsystem $\Phi \subset \Delta \backslash
\Delta^f$, such that in (\ref{eq:3.1}) one has equalities.  In
this case $\Phi |_{\fh^f} = (\Delta^0 \cup \Delta^{1/2}) |_{\fh^f}$.

\item %%c
There exists a root subsystem $\Phi \subset \Delta \backslash
\Delta^f$ such that rank $\Phi = n-m$, but $|\Phi |< |\Delta^0
\cup \Delta^{1/2}|$ iff  $f \neq f_m$.
%%%the partition, corresponding to $f$ is no
%%%of the form $m_1=\cdots = m_{s_1} >m_s \geq 0$.

\item %%d
If  $f=f_m$ and $\Phi \subset \Delta
\backslash \Delta^f$ is a root subsystem of rank~$n-m$, then
$|\Phi | = | \Delta^0 \cup \Delta^{1/2}|$.
% and $\Phi |_{\fh^f} =(\Delta^0 \cup \Delta^{1/2})|_{\fh^f}$.

\end{enumerate}
\end{proposition}

\begin{proof}

%(a) follows from Proposition~\ref{prop:3.2} since
%$  \dim \fg^f = \dim  \fg_0 + \dim \fg_{1/2}$ and $\dim \fg_0 =
%|\Delta^0 | + n-1$ and $\Phi$ is, up to $W$-conjugation, the set
%of roots of $\fg^h$ with $h \in \Omega_f$, hence rank~$[\fg^h
%,\fg^h]=$ rank~$\Phi \dim \fg^h = |\Phi | + n-1$. 

Recall that
\begin{displaymath}
  \dim \fg^f = \dim \fg_0 + \dim \fg_{1/2} \, , \, 
  \dim \fg_0 = |\Delta^0 |+n-1 \, , \, 
  \dim \fg_{1/2} = |\Delta^{1/2}| \, .
\end{displaymath}
Note that a root subsystem $\Phi \subset \Delta \backslash \Delta^f$ is, 
up to $W$-conjugation, the set of roots of
$\fg^h$ with $h \in \Omega_f$.  Hence, by
Proposition~\ref{prop:3.2}(a), rank~$\Phi \leq n-m$ and $|\Phi |
+n-1 \leq |\Delta^0 \cup \Delta^{1/2}| +n-1$, proving~(a).  By
Proposition~\ref{prop:3.2}(b), for $h \in \Sh^0_f$ one has
equalities in (\ref{eq:3.1}), proving the first part of ~(b).
It follows from Proposition~\ref{prop:3.2}(b) that $\Sh^0_f$ 
consists of semisimple elements $h$, for which $\fh^f$ is conjugate
to a Cartan subalgebra of $[\fg^h,\fg^h]$ and the set of roots with
respect to $\fh^{f*}$ in $[\fg^h,\fg^h]$ is the same. Since $\Sh^0_f$ is 
dense in $\Sh_f$, we obtain the second part of (b).

In order to prove (c), denote by $N_m$ the set of nilpotent
matrices~$X$, such that $X^m =0$, but $X^{m-1} \neq 0$.  It is
the same as to say that the first part of the partition,
corresponding to $X$, equals~$m$. Note that the adjoint orbit of
$f_m$ is dense in $N_m$.  If $f \neq f_m$, then there exists an
 element $f' \in N_m$, such that the closure of the adjoint
orbit of $f'$ contains $f$.  If we choose $f'$ in the canonical
Jordan form, then $\fh^f \supset \fh^{f'}$, hence $\Omega_f
\supset \Omega_{f'}$.  Taking $h \in \Sh^0_{f'} \subset
\Sh^0_f$, the root system $\Phi$ of $\fg^h$ will satisfy all
requirements of~(c).

If $f$ is of the form, described in~(c), then all orbits
from $N_m$ lie in the closure of the orbit of $f$ and all
semisimple $h$ such that rank~$[\fg^h ,\fg^h] =n-m$ lie in the
sheets of the nilpotents, contained in $N_m$.  Hence there is no
$h \in \fh$, such that rank~$[\fg^h ,\fg^h]=n-m$, but $\dim \fg^f
>\dim \fg^h$, proving~(c). (d) follows from (a) and (c).

\end{proof}

\subsection{The main theorems for $\sl_n$.}
\label{sec:3.2}

\begin{lemma}
  \label{lem:3.1}

Let $f$ be a nilpotent element of $\sl_n$, 
and let $m$ be the maximal part of the partition of~$n$,
corresponding to $f$.  Let $\Lambda \in Pr^{k,R}$ be such that
$\ch_{H^R (L(\Lambda))}$ does not vanish, where
$k=-n+\frac{p}{u}$, $p,u \in \NN$, $(p,u)=1$, $p \geq n$.  Then
$m\leq u$.

\end{lemma}

\begin{proof}
By definition, the root system $\hat{\Delta}^R_\Lambda$ is of
type $A^{(1)}_{n-1}$, hence
\begin{displaymath}
  \hat{\prod}^R_\Lambda = \{ u_0 K+\gamma_0 ,\ldots ,u_{n-1}
     K+\gamma_{n-1} \},
\end{displaymath}
where
\begin{equation}
  \label{eq:3.2}
  \sum^{n-1}_{i=0} u_i = u \, , \quad u_i \in \ZZ_+ \, .
\end{equation}
Note that rank $\Delta^R_{\Lambda} = | \, \{ i |\, \, 0 \leq i \leq n-1 \, ,
\, u_i =0 \}|$.  From (\ref{eq:3.2}) we obtain 
\begin{equation}
  \label{eq:3.3}
  {\rm rank  }\, \Delta^R_{\Lambda} \geq n-u \, .
\end{equation}
Since, by assumption, $\ch_{H^R (L(\Lambda))} \neq 0$, by
Theorem~\ref{th:2.3}(c), $\Delta^R_{\Lambda} \subset \Delta^R
\backslash \Delta^{R,f}$.  Hence, by
Proposition~\ref{prop:3.2}(a), rank $\Delta^R_{\Lambda} \leq n-m$.  The
lemma now follows from (\ref{eq:3.3}).
\end{proof}

\begin{theorem}
  \label{th:3.1}
  
  Given a positive integer $m \leq n$, denote, as above, by $f_m$ the
  nilpotent element of $\fg = \sl_n$, corresponding to the
  partition of the form $m=m=\cdots = m>s \geq 0$.   Let
\begin{displaymath}
  k=k_{p,m} = \frac{p}{m} -n \hbox{  where  } p \in \ZZ \, , \,
      p \geq n \hbox{  and  } (p,m)=1 \, .
\end{displaymath}
Then
\alphaparenlist
\begin{enumerate}{}{}

  \item %%a
    $(k,f_m)$ is an exceptional pair, that is the following two
    properties hold:   

\romanlistii
\begin{enumerate}
\item %%i
$\ch_{H^R (L (\Lambda))}$ either vanishes or
    is almost convergent for each $\Lambda \in Pr^{k,R}$;

\item %%ii
there exists $\Lambda \in Pr^{k,R}$, such that $\ch_{H^R
  (L(\Lambda))}$ does not vanish.
\end{enumerate}

\item %%b
If $k \neq k_{p,m}$ and $m<n$, then $(k,f_m)$ is not an
exceptional pair.
  \end{enumerate}

\end{theorem}

\begin{proof}

Let $\Lambda \in Pr^{k,R}$ be a principal admissible weight of
level $k=\frac{p}{m}-n$, such that $\ch_{H^R (L(\Lambda))}$ does
not vanish.  It follows from (\ref{eq:3.3}) with $u=m$, that
%
%Note that $\hat{\prod}^R_\Lambda = \{ m_o K+\gamma_0 \, , \, s_1
%K+\gamma_1,\ldots \, , \, m_{n-1} K+\gamma_{n-1}\}$, where
%$\gamma_i \in \Delta^R$, $m_iK+\gamma_i \in \hat{\Delta}^R_+$,
%$\sum^{n-1}_{i=1} (m_i K + \gamma_i)= mK$.  It follows that:
%%
%\begin{displaymath}
%  \sum^{n-1}_{i=0} m_i =m \, , \quad \sum^{n-1}_{i=0} \gamma_i=0\, .
%\end{displaymath}
%%
%Let $I_0= \{ i |0 \leq i \leq n \, , \, m_i=0 \}$.  Since all
%$m_i \in \ZZ_+$, we obtain:
%
\begin{displaymath}
{\rm rank}\, \Delta^R_{\Lambda} \geq n-m \, .
\end{displaymath}
Since, by our assumption, $\ch_{H^R (L(\Lambda))} \neq 0$, 
we have: $\Delta^R_\Lambda \subset \Delta^R \backslash \Delta^{R,f}$, 
by Theorem~\ref{th:2.3}(c).  
Hence, by
Proposition~\ref{prop:3.2}(a), ${\rm rank }\, \Delta_{\Lambda}^R
\leq n-m$, and we conclude that rank $\Delta^R_\Lambda = n-m$.
But then by Proposition~\ref{prop:3.3}(b), we obtain that
$\Delta^R_\Lambda |_{\fh^f} = (\Delta^0 \cup
\Delta^{1/2})|_{\fh^f}$.  By Theorem~\ref{th:2.4}, we conclude
that $\ch_{H^R (L (\Lambda))}$ is almost convergent, proving~(i).

Due to Theorem~\ref{th:2.3}(c) and Propositions \ref{prop:3.2}(b)
and \ref{prop:3.3}(d), (ii) holds as well, proving~(a).

Claim (b) follows from the following two statements:

\romanparenlist
\begin{enumerate}
\item %%i
if $u >m$, then there exists $\Lambda \in Pr^{k,R}$, such that
$\Delta^R_\Lambda \subset \Delta^R \backslash \Delta^{R,f}$,
$|\Delta^R_\Lambda |<|\Delta^0 \cup \Delta^{1/2}|$ (hence, by
Corollary~\ref{cor:2.2}, $\ch_{H^R (L(\Lambda))}$ is not almost
convergent);

\item %%ii
if $u<m$, then $\Delta^R_{\Lambda} \cap \Delta^{R,f} \neq
\emptyset$ for any $\Lambda \in Pr^{k,R}$ (hence by
Theorem~\ref{th:2.3}(c), $\ch_{H^R (L(\Lambda))} =0$). 

\end{enumerate}

In order to prove (i), let $k'=k_{p.m}$, so that $(k',f_m)$ is an
exceptional pair.  Then, by Theorem~\ref{th:3.1}, there exists
$\Lambda' \in Pr^{k',R}$ for which $\ch_{H^R (L (\Lambda'))}$
almost converges.  As before, we write $\Lambda'$ in the
following form:  $\Lambda' = (t_{\beta'} \bar{y}).  (\Lambda^0 -
(m-1)\frac{p}{m} D^R)$, where  $\bar{y} \in W^R$ is such that
$\bar{y} (\alpha^R_i) = \gamma_i$, $(\beta' |\gamma_i) =-m'_i$
for $i=1,\ldots ,n-1$ and $\Lambda^0 \in \hat{P}^{p-n}_+$, 
so that we have $\hat{\prod}_{\Lambda'}^R =\{ m'_i K
+ \gamma_i \}_{i=0,\ldots ,n-1}$.  Let $I_{\Lambda'} = \{ i | \, 0
\leq i \leq n-1, m'_i =0 \}$.  Then $\{ \gamma_i \}_{i \in
I_{\Lambda'}}$ is the set of simple roots of $\Delta^R_{\Lambda'}
\cap \Delta^R_+$.  Recall that the almost convergence  of $\ch_{H^R
  (L (\Lambda'))}$ implies:
\begin{displaymath}
  \Delta^R_{\Lambda'} \cap \Delta^{R,f} = \emptyset \, , \, 
|\Delta^R_{\Lambda'} | = |\Delta^0 \cup \Delta^{1/2} | \, .
\end{displaymath}

Now fix $i_0 \in I_{\Lambda'}$ and define $m_i$ for $0 \leq i
\leq n-1$ and $\beta \in \sum^{n-1}_{i=1} \RR\, \alpha^R_i$  by the
following relations:
\begin{eqnarray*}
  m_i = m'_i \hbox{  if  } i \neq i_0 \, , \, 
     m_{i_0} = u-m \, ; \\
     (\beta |\gamma_i) = -m_i \hbox{  for  } i=1,\ldots ,n-1\, .
\end{eqnarray*}
Then $\Lambda = (t_\beta \bar{y}). (\Lambda^0 -(u-1) \frac{p}{u}
D^R) \in Pr^{k,R}$ with $\hat{\prod}^R_\Lambda = \{ m_i K +
\gamma_i \}_{i=0,\ldots ,n-1}$.  Since $m_{i_0} >0$,
$\prod^R_\Lambda = \prod^R_{\Lambda'} \backslash \{
\gamma_{i_0}\}$, hence 
\begin{displaymath}
  \Delta^R_\Lambda \subset \Delta^R_{\Lambda'}\subset \Delta^R \backslash
     \Delta^{R,f} \hbox{  and  } |\Delta^R_\Lambda | <
       | \Delta^R_{\Lambda'} | = | \Delta^0 \cup \Delta^{1/2}|\, ,
\end{displaymath}
proving (i).

In order to prove (ii), let $\Lambda \in Pr^{k,R}$, where
$\hat{\prod}^R_\Lambda = \{ m_i K+\gamma_i \}_{i=0,\ldots
  ,n-1}$.  Then $u=\sum^{n-1}_{i=0} m_i <m$ (by our assumption),
hence ${\rm rank }\, \Delta^R_\Lambda = | \{ i |\,  m_i = 0 \}|
\geq n-m +1 >n-m$.  Therefore, by Proposition~\ref{prop:3.3}(a),
$\Delta^R_\Lambda \cap \Delta^{R,f} \neq \emptyset$, proving~(ii).

\end{proof}

\begin{theorem}
  \label{th:3.2}

Let $f$ be a nilpotent element of $\sl_n$, different from any of
the nilpotent elements $f_m$, $1 \leq m\leq n$.  Then $(k,f)$ is
not an exceptional pair for any~$k$.

\end{theorem}

\begin{proof}

Let $m$ be the largest part of the partition, corresponding to
$f$, and suppose that $(k,f)$ is an exceptional pair, where $k=-n
+ \frac{p}{u}$, $p,u \in \NN$, $(p,u) =1$, $p \geq n$.  Let
$\Lambda \in Pr^{k,R}$ be such that $\ch_{H^R (L(\Lambda))}$ does
not vanish.  Then, by Lemma~\ref{lem:3.1},
\begin{equation}
  \label{eq:3.4}
  m \leq u \, .
\end{equation}

By Proposition \ref{prop:3.3}(c), there exists a root subsystem
$\Phi \subset \Delta^R \backslash \Delta^{R,f}$, such that
\begin{displaymath}
  {\rm rank} \, \Phi  = n-m \, , \quad 
     |\Phi | < |\Delta^0  \cup \Delta^{1/2}| \, .
\end{displaymath}
Let $\prod_\Phi = \{ \beta_1 ,\ldots ,\beta_{n-m}\}$, be the set
of simple roots of $\Phi \cap \Delta^R_+$, and extend
$\prod_\Phi$ to a set of simple roots $\prod' = \{ \gamma_1
,\ldots ,\gamma_{n-1}\}$ of $\Delta^R$.  Let $\gamma_0 =
-\sum^{n-1}_{i=1} \gamma_i$, and define $u_i \in \ZZ_+$ as
follows (here we use (\ref{eq:3.4})): $u_0 = u-m+1$, $u_i =1$
if $\gamma_i \in  \prod^\prime \backslash \prod_\Phi$, $u_i =0$
if $\gamma_i \in \prod_\Phi$.
Let $\hat{\prod}_\Phi = \{ u_i K + \gamma_i|i=0,\ldots ,n-1 \}
\subset \hat{\Delta}^R_+$, and let $\Lambda$ be a principal
admissible weight of level~$k$, such that $\hat{\prod}^R_\Lambda
=\hat{\prod}_\Phi$.  Note that $\Delta^R_\Lambda =\Phi$.  Hence
$|\Delta^R_\Lambda | < |\Delta^0 \cup \Delta^{1/2}|$, and
therefore, by Corollary~\ref{cor:2.2}, $\ch_{H^R (L
  (\Lambda))}$ is not almost convergent.
\end{proof}

\subsection{Example of $\fg = \sl_3$, $f=$ \emph{minimal nilpotent.}}
\label{sec:3.3}~~
In this case $x=\frac12 \theta$, where $\theta = \alpha_1 +
\alpha_2$ is the highest root, $\fg_0 =\fh$, hence
$\Delta^0=\emptyset$, and $\Delta^{1/2} =\{\alpha_1,\alpha_2
\}$.  We choose $h_0 = \alpha_1 -\alpha_2$; then $\fh^f =\CC
h_0$, $\Delta^{1/2}_+ = \{ \alpha_1 \}$, $\Delta^{1/2}_- =
\{\alpha_2 \}$, $\Delta^{\rm new}_+ = \{ \alpha_1 ,-\alpha_2
,-\alpha_1-\alpha_2 \} = \bar{w} (\Delta_+)$,  where $\bar{w}
=r_2r_1$.  Therefore,
\begin{align*}
  \Delta^R_+ &= t_x (\Delta^{\rm new}_+) = \{ -\frac12 K+\alpha_1
   \, , \, \frac12 K-\alpha_2 \, , \, K-\alpha_1-\alpha_2 \} \, ,\\
   {\prod}^R &= t_x \bar{w} (\{ \alpha_1 ,\alpha_2 \}) =\{
     \alpha^R_1 :=K-\alpha_1-\alpha_2 \, , \, \alpha^R_2:=
       -\frac12 K + \alpha_1 \} \,, \\
   \hat{\prod}^R &= t_x \bar{w} (\{ \alpha_0,\alpha_1,\alpha_2\})
     = \{ \alpha^R_0 :=\frac12 K +\alpha_2 \, , \, \alpha^R_1
        \, , \, \alpha^R_2 \} \, .
\end{align*}
We also have:
\begin{displaymath}
  \fh^f = \CC (\alpha^R_1+2\alpha^R_2)\, , \, \Delta^{R,f}_+
    = \{ \alpha^R_1 \} \, , \, \Delta^{R,0}_+ =\emptyset \, , \, 
      \Delta^{R,1/2}_+ = \{ \alpha^R_2 \} \, .
\end{displaymath}

Recall that the corresponding to $f$ exceptional denominator is
$u=2$.  Then (cf.~Section~\ref{sec:1.2}):
\begin{displaymath}
  \hat{S}^R_{(2)} = \{ 2K -\alpha^R_1-\alpha^R_2
    \, , \, \alpha^R_1 \, , \, \alpha^R_2 \} \, , 
\end{displaymath}
and the corresponding set of roots is
\begin{displaymath}
  \hat{\Delta}^R_{(2)}= \{ 2nK + \alpha | \alpha \in \Delta^R \,  ,\, 
    n \in \ZZ  \} \cup \{ 2nK | n \in \ZZ\backslash \{ 0 \} \}\, .
\end{displaymath}
Consider all possible subsets $t_\beta \bar{y} (\hat{S}_{(2)})$
where $\bar{y} \in W^R = t_x Wt^{-1}_x$, $\beta \in \ZZ
\Lambda^R_1+\ZZ \Lambda^R_2$, $\Lambda^R_i \in \fh^R =t_x (\fh)
$, $(\Lambda^R_i | \alpha^R_j) = \delta_{ij}$, satisfying the
conditions
\begin{displaymath}
  t_\beta \bar{y} (\hat{S}^R_{(2)}) \subset \hat{\Delta}^R_+
  \hbox{  and   } t_\beta \bar{y} (\hat{\Delta}^R_{(u)})\cap
  (\Delta^{R,f}=\{ \alpha^R_1 \})=\emptyset \, .
\end{displaymath}
It is easy to see that there are two possibilities for such
subsets:
\begin{displaymath}
  \hat{\prod}' = t_{-\Lambda^R_1} (\hat{S}_{(2)}^R) =\{
    K-\alpha^R_1 -\alpha^R_2 \, , \, K+\alpha^R_1 \, , \, \alpha^R_2 \}
    \hbox{  and  } \hat{\prod}''=r_{\alpha^R_1}
       (\hat{\prod}')\, .
\end{displaymath}
Since $r_{\alpha^R_1} |_{\fh^f}=1$, it follows from 
formula~(\ref{eq:2.3}) that $\chi_{H^R (L(\Lambda))} = \chi_{H^R
  (L (\Lambda''))}$ if $\Lambda'$ and $\Lambda''$ are admissible
weights, satisfying $\Lambda'' +\hat{\rho}^R =r_{\alpha^R_1}
(\Lambda' + \hat{\rho}^R)$.  Thus, it suffices to consider only
the principal admissible weights $\Lambda$, such that
$\hat{\prod}_\Lambda^R = \hat{\prod}' =t_{-\Lambda^R_1}
(\hat{S}_{(2)}^R)$.

Since $u=2$, $p$ should be an odd integer $\geq h^\vee =3$, and
$k=-3+p/2$.  All principal admissible weights $\Lambda$ of
level~$k$ with $\hat{\prod}_\Lambda^R = \hat{\prod}'$ are
\begin{displaymath}
  \Lambda = t_{-\Lambda^R_1}\, . \quad (\Lambda^0 -\frac{p}{2} D^R)
    = \Lambda^0-\frac{p}{2} (\Lambda^R_1+D^R) \mod \CC K \, , \, 
    \hbox{  where  } \Lambda^0 \in \hat{P}^{p-3,R}_+ \, .
\end{displaymath}

Since $\fh^f = \CC\Lambda^R_2$, we can write an arbitrary element
of $\fh^f$ as $z\Lambda^R_2$, $z \in \CC$.
Since $\dim \fg^f =4$ and $h^\vee=3$, we obtain from
(\ref{eq:2.5}):
\begin{displaymath}
  \psi (\tau ,z\Lambda^R_2 , t) = e^{6\pi i t} \eta (\tau)    
\theta (\tau ,z)\, .
\end{displaymath}
%
%Hence, by (\ref{eq:2.3}):
%%
%\begin{equation}
%  \label{eq:3.5}
%  \chi_{H^R (L (\Lambda))} = \frac{B_\Lambda (\tau ,z \Lambda^R_2,t)}
%     {e^{6\pi i t} \eta (\tau) \theta (\tau ,z)}\, ,
%\end{equation}
%%
%where $\Lambda \in M_p := \{ \Lambda^0-\frac{p}{2} \Lambda^R_1 |
%\Lambda^0 \in P^{p-3,R}_+$ and $B_\Lambda$ is given by Remark
%\ref{rem:2.1}.
%
Note that for $\Lambda \in M_p := \{ \Lambda^0-\frac{p}{2} (\Lambda^R_1 +D^R) |
\Lambda^0 \in \hat{P}^{p-3,R}_+ \}$ we have:
\begin{displaymath}
  \hat{\Delta}^R_\Lambda = \{ 2nK \pm \alpha^R_2 \, , \, 
  (2n-1)K \pm \alpha^R_1 \, , \, (2n-1)K \pm 
  (\alpha^R_1 + \alpha^R_2) | n \in \ZZ \} \cup \{
  2nK |n \in \ZZ \backslash 0 \}
\end{displaymath}
(it is a root system with basis $\hat{\prod}^R_\Lambda$), hence
\begin{displaymath}
  \hat{\Delta}^R_{\Lambda ,+}|_{\fh^f} =\{ (n-1)K+
     \alpha^R_2 |_{\fh^f} \, , \, nK-\alpha^R_2|_{\fh^f}|\,\,
         n \in \NN\} \cup \{ nK |n \in \NN \} \, .
\end{displaymath}
%
%Hence the numerator in (\ref{eq:3.5}) is divisible by
%$\prod^\infty_{n=1}
%(1-e^{-((n-1))K+\alpha^R_2})(1-e^{-(nK-\alpha^R_2)})$ in the
%domain $\CC^+ \times \fh^f \times \CC$.  Therefore all
%singularities in this domain of the denominator in (\ref{eq:3.5})
%are canceled by this factor and all functions $\{ \chi_{H^R (L
%  (\Lambda))}\}_{\Lambda \in M_p}$ are holomorphic in this
%domain.
%
Since for principal admissible weights in $M_p$ we have
$\bar{y}=1$, $\beta =-\Lambda^R_1$, a straightforward calculation
gives that in (\ref{eq:2.6}) we have
\begin{displaymath}
  C (\tau ,z\Lambda^R_2 ,t) /\psi (\tau ,z\Lambda^R_2,t)
  =e^{-3\pi i t} \, .
\end{displaymath}
Hence (\ref{eq:2.6}) for $\Lambda \in M_p$ becomes:
\begin{equation}
  \label{eq:3.5}
  \chi_{H^R (L (\Lambda))} (\tau ,z\Lambda^R_2 ,t) =e^{-3\pi it}
    \chi_{L^R (\Lambda^0)} (2\tau ,z\Lambda^R_2 -\tau
    \Lambda^R_1 , \, \frac12 ( t+ \frac{\tau -z}{3})) \, .
\end{equation}

Since $\dim \fg=8$, $\dim \fg^f =4$ and $\ch_{H^R (L (\Lambda))}$
for $\Lambda \in Pr^{k,R}$ does not vanish only when $\hat{\prod}_{\Lambda}^R$ 
is $\hat{\prod}'$ or $\hat{\prod}''$, and the latter two give equal
contributions. Thus, Theorem~\ref{th:2.1} gives the following
transformation formula for  $\Lambda \in M_p$:
\begin{displaymath}
  \chi_{H^R (L (\Lambda))} \left(-\frac{1}{\tau} , \frac{z\Lambda^R_2}{\tau}\, ,\, 
     t-\frac{Q_p(z)}{2\tau} \right)= -2\sum_{\Lambda' \in M_p}
     a (\Lambda ,\Lambda') \chi_{H^R (L (\Lambda'))} 
     (\tau ,z \Lambda^R_2,t )\, , 
\end{displaymath}
where $Q_p (z) = \frac{2p-6}{3p-18}z^2$.

Finally, we compute asymptotics, using Theorem~\ref{th:2.2}.  It
is easy to see that $A_\beta (z\Lambda^R_2)=2$ for $\beta =
- \Lambda^R_1$.  Hence we have for $\Lambda \in M_p$, as $\tau
\downarrow 0$:
\begin{displaymath}
  \ch_{H^R (L (\Lambda))} (\tau ,-\tau z \Lambda^R_2,0) =
  a (\Lambda^0) e^{\frac{\pi i}{12 \tau} 4
    \left(1-\frac3p\right)}\, .
\end{displaymath}

\begin{remark}
  \label{rem:3.1}
For $\fg=s\ell_n$ and its exceptional pair $(k=-n + \frac{p}{u}\,
, \, f=f_u)$, where $u \leq n$, the ``extra factor'' in the
character formula (\ref{eq:2.6}) is independent of $z$, and is
given by the following formula:
\begin{displaymath}
  \frac{C (\tau ,z,t)}{\psi (\tau ,z,t)} = \pm a (t) q^b
    \left( \prod^\infty_{j=1}
       \frac{1-q^{uj}}{1-q^j}\right)^{s'-1} M (q) \, , 
\end{displaymath}
where $\pm = (-1)^{j_{\Lambda}}$, $a(t) = e^{2\pi i n (u^{-1}-1)t}$,
$b=(s-1)(u-s-1)(su - s^2 +u)/24u$, $s$ is the remainder of the
division of $n$ by $u$, $s'= {\rm min} \{ s,u-s \}$, and
\begin{displaymath}
  M(q) = \prod^{s'}_{i=1} \left( \prod^\infty_{j=1}
       (1-q^{(j-1)u+i}) (1-q^{ju-i})\right)^{s'-i} \, .
\end{displaymath}
In particular, the ``extra factor'' is equal to $\pm a (t)$ if
$s'=1$, and to $\pm a (t) \eta (\tau ) /\eta (u \tau)$ if $s'=0$.

\end{remark}

\subsection{Conjectures.}
\label{sec:3.4}~~Let $\fg$ be a simple Lie algebra, and denote by
$E_0$ the set of all non-principal exceptional nilpotent orbits of $\fg$.  Let
$h$ be the Coxeter number of~$\fg$, and let $I_0$ denote the set
of integers~$j$, such that $1 \leq j <h$ and $(j, \ell) =1$,
where $\ell (= 1,2 \hbox{ or }3)$ is the lacety of~$\fg$.

\vspace{2ex}

\textbf{Conjecture E.}~~There exists an order preserving
map $\varphi : E_0 \to I_0$, such that all exceptional
pairs $(k,f)$, for which $f$ is not principal, are as follows:
$f $ lies in an orbit from $E_0$, $k=-h^\vee +
\frac{p}{\varphi (f)}$, where $p \in \ZZ$, $p \geq h^\vee$ and
$(p,\varphi (f))=1$.

\vspace{2ex}

Recall that all adjoint nilpotent orbits of $\fg = so_N$, $N$ odd
(resp.~$sp_N$, $N$~even) are intersections of nilpotent orbits of
$\sl_N$ with $\fg$, if they are non-empty, and these
intersections are non-empty iff all even (resp.~odd) parts of the
corresponding partition of $N$ have even multiplicity.  In the
case $\fg = so_N$, $N$~even, the answer is the same, as for $\fg
= so_N$, $N$~odd, except that in cases when all parts of the
partition are even, the intersection consists of two orbits,
permuted by an outer automorphism of $so_N$, which we shall
identify.  Furthermore, a nilpotent orbit of $\fg = so_N$ or
$sp_N$ is of principal type if either the multiplicities of all
parts of the partition of $N$ are even, or else one has
respectively:

\begin{list}{}{}
\item $\fg = sp_N$, and exactly one even part has odd multiplicity,

\item $\fg = so_N$, $N$ odd, and exactly one odd part has odd
  multiplicity,

\item $\fg=so_N$, $N$ even, and exactly two distinct parts, one of
  which is~$1$, have odd multiplicity.

\end{list}

\textbf{Conjecture F.}~~Let $\fg$ be $so_N$ or $sp_N$.  Let $m$
be a positive integer, and denote by $N_m$ the set of matrices $\{
X \in {\rm Mat}_N \cap \fg | X^m =0 \}$, and let $\O_m$ be the
adjoint orbit, open in $N_m$ (such an orbit is unique if $\fg
\neq so_{4n}$ or $m$ is odd).  Denote by $F_0$ the set of all
non-principal adjoint orbits $\O_m$ of principal type with $m$ odd.  
Then $F_0 \subset E_0$ and $\varphi (\O_m) = m$ for $\O_m
\in F_0$.

\vspace{2ex}

\textbf{Conjecture G.}~~

\alphaparenlist
\begin{enumerate}
\item %%a
If $\fg = so_{2n+1}$ or $sp_{4n+2}$, then $E_0 = F_0$.

\item %%b
If $\fg = sp_{4n}$, then $E_0 = \{F_0 \, , \, \O_{2n,2n}\} $,
where $\O_{2n,2n}$ denotes the nilpotent orbit, corresponding to
the partition $4n=2n+2n$, and $\varphi (\O_{2n,2n}) = 2n+1$.

\item %%c
If $\fg = so_{2n}$, then $F_0$ consists of all exceptional
nilpotent orbits $\O$ with $\varphi (\O)$ odd.

\end{enumerate}

We checked conjectures~E, F and G for $N \leq 13$.

\begin{examples}
\alphaparenlist
\begin{enumerate}
\item %%a
  $\fg = so_8$.  Then $ E_0 = \{ F_0 \, , \, \O_{3,2,2,1} \}$,
  where $\O_{3,2,2,1}$ is the nilpotent orbit, corresponding to
  the partition $8=3+2+2+1$, and $\varphi (\O_{3,2,2,1}) =2$.
  (Note that $\O_{3,2,2,1}$ is not open in $N_3$.)

\item %%b
  $\fg = so_{10}$.  Then $ E_0 = \{ F_0\, , \, 
  \, \O_{3,2,2,1,1,1}\}$, and $\varphi
  (O_{3,2,2,1,1,1})=2$.  (Note that $\O_{3,2,2,1,1,1}$ is not open in $N_3$.)

\item %%c
$\fg = so_{12}$.  Then $E_0 = \{ F_0, \O_{5,3,3,1}  
,\O_{3,2,2,2,2,1},\O_{3,2,2,1,1,1,1,1}\}$,
and $ \varphi (\O_{5,3,3,1})=4$, \break $\varphi (\O_{3,2,2,2,2,1})= 
\varphi (\O_{3,2,2,1,1,1,1,1})=2$.   
(Note that $O_{5,3,3,1}$ is
not open in $N_5$ and 
$\O_{3,2,2,2,2,1}$, $\O_{3,2,2,1,1,1,1,1}$ 
are not open in $N_3$.)

\item %%d
$\fg = G_2$.  Then the only non-zero non-principal exceptional
nilpotent orbit is the orbit~$\O$ of a root vector $e_\alpha$,
where $\alpha$ is a short root, and $\varphi (\O) =2$.
This is the simplest case when the extra factor in (\ref{eq:2.6})
does depend on $z$: 
\begin{displaymath}
  \frac{C (\tau ,z,t)}{\psi (\tau ,z,t)} = e^{-4\pi i t}
  \frac{f(\tau,z)}{f (2\tau ,2z)}\, ,
\end{displaymath}
where $f(\tau ,z)$ is defined by (\ref{eq:2.7}).

\end{enumerate}

\end{examples}

\vspace{2ex}

%\textbf{Conjecture H.}~~The vertex algebras $W_k(\fg, f)$ with $(k,f)$
%an exceptional pair are semisimple. 
%vertex algebras among all the
%vertex algebras $W_k (\fg, f)$.

\subsection{Corrections to \cite{KRW} and \cite{KW4}.}
\label{sec:3.5}

\arabiclist

\begin{enumerate}
\item %%1
The last sentence of Proposition~3.3 in \cite{KW4} should be
changed as follows (cf.~Theorem~\ref{th:2.3} of the present
paper):

Then $\ch_{H^{\tw}(M)} (\tau ,h)$ is not identically zero iff the
character of the $\hat{\fg}^{\tw}$-module~$M$ has a pole at all
hyperplanes $\alpha =0$, where $\alpha$ are positive even real
roots, satisfying the following three properties:

Similar change should be made in Theorem~3.2 of \cite{KRW}.
\item %%2
The first factor on the right of formula~(3.14) in \cite{KW4}
should be replaced by the following expression:
\begin{eqnarray*}
 e^{2\pi i\tau( 
       \frac{(\bar{\Lambda} | \bar{\Lambda} + \bar{\hat{\rho}}^{tw})}
         {2 (k+h^\vee)} + s_{\fg} + s_{\rm ch}
    + s_{\rm ne} + (\Lambda | \Lambda_0))}
   e^{\pi i( \sum_{\alpha \in S_{1/2}} (-1)^{p(\alpha)}
  s_\alpha \alpha (h) - 2 \sum_{\alpha \in S_+} (-1)^{p(\alpha)}
     s_\alpha \alpha (h))}\, .
\end{eqnarray*}

\item %%3
The convergence  of characters argument in the proof of
Theorem~3.1 in \cite{KRW} should be changed as in the proof of
Section~\ref{sec:2.2} of the present paper.  In particular, one
should add there the assumption that $f$ is of principal type.

\end{enumerate}

\subsection{Appendix:  on representations of $W^{\rm fin} (\fg ,f)$.}
\label{sec:3.6}

~~The associative algebra $W^{\rm fin} (\fg ,f)$ is obtained by
quantum Hamiltonian reduction as follows \cite{P}, \cite{GG}.
Let $\fg_{\geq 1} = \oplus_{j \geq 1} \fg_j$ and let $\chi :
\fg_{\geq 1} \to \CC$ be a homomorphism, defined by $\chi (a) =
(f|a)$.  Extend $\chi$ to a homomorphism $\chi : U (\fg_{\geq 1})
\to \CC$ and let $I_\chi \subset U (\fg_{\geq 1})$ be the kernel
of $\chi$.  The subspace $I_\chi$ is invariant with respect to
the adjoint action of $\fg_+$, hence the left ideal $U (\fg)
I_\chi$ of $U (\fg)$ is $\ad \fg_+$-invariant as well, and we let
\begin{displaymath}
  W^{\rm fin} (\fg ,f) = (U(\fg)/U (\fg)  I_\chi)^{\ad \fg_+}\, .
\end{displaymath}
It is easy to check that the product on $U (\fg)$ induces a
well-defined product on $W^{\rm fin} (\fg ,f)$. We thus obtain the
\emph{finite $W$-algebra}, associated to the pair $(\fg ,f)$.

For the purpose of representation theory it is more convenient to
use an equivalent definition of  $W^{\rm fin} (\fg ,f)$, similar
to that of $W^k (\fg ,f)$ (equivalence of these two definitions
was proved in the appendix to \cite{DK}).  Let ${\rm Cl} (\fg ,f)$
denote the Clifford--Weil algebra on the vector superspace
$A_{\ch} \oplus A_{\ne}$  with the bilinear form $(\,.\, |\, .\,) \oplus
\langle \, . \, | \, . \, \rangle$.  Let $\C (\fg ,f) = U (\fg )
\otimes {\rm Cl} (\fg ,f)$, and introduce the following odd
element of $\C (\fg ,f)$:
\begin{displaymath}
  d= \sum_{\alpha \in S_+} (u_\alpha + (f |u_\alpha)) 
     \varphi^\alpha + \sum_{\alpha \in S_{1/2}} \varphi^\alpha
     \phi_\alpha - \frac12 \sum_{\alpha ,\beta , \gamma \in S_+}
       c^\gamma_{\alpha ,\beta} \varphi_\gamma \varphi^\alpha
         \varphi^\beta \, .
\end{displaymath}
(This element is obtained from $d(z)$ in Section~\ref{sec:2.1} by
dropping $z$ and the signs of normally ordered product.)  Then
$d^2 = \frac12 [d,d] =0$, and we have \cite{DK}:
\begin{displaymath}
   W^{\rm fin} (\fg ,f) = H (\C (\fg ,f), \ad \, d)\, .
\end{displaymath}
As in \cite{KW3} or \cite{DK}, one proves that this homology is
concentrated in $0$\st{th} degree with respect to the charge
decomposition, defined by 
\begin{equation}
  \label{eq:3.6}
  {\rm charge}\,  u_\alpha = {\rm charge}\, \phi_\alpha = 0 \, , \, 
   {\rm charge}\,  \varphi_\alpha = -{\rm charge}\, 
        \varphi^\alpha = 1 \, . 
\end{equation}

Recall the construction of $\Delta^{\rm new}_+ \subset \Delta$
in Section~\ref{sec:2.2}, and let
\begin{displaymath}
  R^\pm_j = \{  \alpha \in \pm \Delta^{\rm new}_+ | \alpha (x)=j \}
    \, , \quad R^+_{>0} = \cup_{j>0} R^+_j \, , \quad 
       R^-_{>0} = \cup_{j>0} R^-_j \, .
\end{displaymath}
Let $S (\fg ,f)$ denote the irreducible ${\rm Cl} (\fg
,f)$-module, generated by the even vector $\vac$, subject to the
conditions:
\begin{displaymath}
  \phi_\alpha \vac = 0 \hbox{ \,\, if \,\,} \alpha \in R^+_{1/2}\,,
      \quad \varphi_\alpha \vac = 0 \hbox{\, \, if \,\,}
         \alpha \in R^+_{>0} \, \quad \varphi^\alpha \vac =0 
            \hbox{\,\, if \,\,} \alpha \in R^-_{>0}\, .
\end{displaymath}

As in Section~\ref{sec:2.2}, given a $\fg$-module $M$, we can
construct the complex
\begin{displaymath}
  (\C (M) = M \otimes S (\fg ,f),d)\, .
\end{displaymath}
It is a $\ZZ$-graded $\C (\fg ,f)$-module $\C (M) = \oplus_{j \in
\ZZ}\,  \C {}_j(M)$, where this charge decomposition extends
(\ref{eq:3.6}) by letting charge $M =0$, charge $\vac =0$.  We
thus obtain for each $j \in \ZZ$ a functor $H_j$ from the
category of $\fg$-modules to the category of $ W^{\rm fin}(\fg
,f)$-modules, given by:
\begin{displaymath}
  H_j (M) : = H_j (\C (M) , d)\, .
\end{displaymath}
As in Section~\ref{sec:2.2}, we prove the following proposition.

\begin{proposition}
  \label{prop:3.4}

Assume that $f$ is a nilpotent element of principal type.  Let
$M$ be a highest weight $\fg$-module with respect to $\Delta^{\rm
  new}_+$, so that $\ch_M (h) = \frac{n_M (h)}{\prod_{\alpha \in
    \Delta^{\rm new}_+} (1-e^{-\alpha (h)})}$, $h \in \fh$.  Then
the Euler--Poincar\'e character $\ch_{H(M)}$ of the $W^{\rm fin} (\fg
,f)$-module $H(M)=\oplus_{j \in \ZZ} H_j (M)$ is given by:
\begin{equation}
  \label{eq:3.7}
  \ch_{H(M)} (h) = \frac{n_M (h)}{\prod_{\alpha \in R^+_0 \cup
      R^+_{1/2}} (1-e^{-\alpha (h)})}\, , \quad h \in \fh^f\, .
\end{equation}

\end{proposition}

We rewrite formula (\ref{eq:3.7}), using the Kazhdan--Lusztig
formula for $n_{L (w.\Lambda)}$, where $w \in W$ and $\Lambda
+\rho$ is an integral regular anti-dominant weight:
\begin{equation}
  \label{eq:3.8}
  n_{L (w.\Lambda)} = \sum_{y \in W} \epsilon (y) \epsilon (w)
     P_{y,w} (1) e^{y.\Lambda}\, , 
\end{equation}
where $P_{y,w} (q)$ are the Kazhdan--Lusztig polynomials.

Since $f$ is a nilpotent element of principal type it can be
written as a sum of root vectors, attached to roots
$\gamma_1,\ldots ,\gamma_s$, where $\gamma_i-\gamma_j \notin
\Delta \cup \{ 0 \}$ for $i \neq j$, and $\gamma_i |_{\fh^f}
=0$.  Denote by $W^f$ the subgroup of $W$, generated by
reflections in the $\gamma_i$, $i=1,\ldots ,s$.

Since $e^{s.\lambda} |_{\fh^f} = e^\lambda |_{\fh^f}$ for $s \in W^f$,
we obtain from (\ref{eq:3.8}):
\begin{align*}
  n_{L (w.\Lambda)} |_{\fh^f} 
     &= \sum_{\substack{y \in W^f
           \backslash W\\ s \in W^f}} \epsilon (sy) \epsilon (w)
            P_{sy,w} (1) e^{(sy).\Lambda} |_{\fh^f}\\[1ex]
     &= \sum_{y \in W^f \backslash W} \epsilon (y) \epsilon (w)
        \sum_{s \in W^f} \epsilon (s) P_{sy,w} (1) e^{y.\Lambda}
          |_{\fh^f}\, .
\end{align*}

We thus obtain the following proposition.

\begin{proposition}
  \label{prop:3.5}

Assume that $f$ is a nilpotent element of principal type.  Let
$\Lambda + \rho$ be an integral regular anti-dominant weight and
$L (w.\Lambda)$ an irreducible highest weight $\fg$-module with
respect to $\Delta^{\rm new}_+$.  Then
\begin{equation}
  \label{eq:3.9}
  \ch_{H (L (w.\Lambda))} (h) =
      \frac{\sum_{y \in W^f \backslash W} \epsilon (y) \epsilon
        (w) \tilde{P}_{y,w} (1) e^{(y.\Lambda)(h)}}
      {\prod_{\alpha \in R^+_0 \cup R^+_{1/2}} (1-e^{-\alpha (h)})},
      \, \quad h \in \fh^f \, , 
\end{equation}
where $\tilde{P}_{y,w} (q) = \sum_{s \in W^f} \epsilon (s)
P_{sy,w} (q)$.

\end{proposition}

The following conjecture is a ``finite'' analogue of
Conjecture~B.

\vspace{2ex}

\textbf{Conjecture H.}

\alphaparenlist
\begin{enumerate}
\item %%a
If $M$ is an irreducible highest weight $\fg$-module (with
respect to $\Delta^{\rm new}_+$), then $H(M) = H_0(M)$, and this
is either an irreducible $W^{\rm fin} (\fg ,f)$-module, or zero.

%\item %%b
%The functor $H^0$ is exact.

\item %%b
Suppose that $f$ is of principal type and $H (L (\Lambda)) \neq
 0$.  Then the  
$W^{\rm fin} (\fg ,f)$-modules 
$H (L (\Lambda))$ and
$H (L (\Lambda'))$ are isomorphic iff $\Lambda' = y .\Lambda$
where $y \in W^f$.

\item %%c
All finite-dimensional irreducible $W^{\rm fin} (\fg ,f)$-modules are
contained among the $H_0 (L (\Lambda))$.

  \end{enumerate}

In the case when $f$ is a principal nilpotent in a Levi subalgebra, 
it was conjectured in \cite{DV} that the right-hand side of formula 
(\ref{eq:3.9}) gives the character of an irreducible highest weight
$W^{\rm fin} (\fg ,f)$-module. Due to Proposition \ref{prop:3.5},
this conjecture (in the more general case of a nilpotent element 
of principal type) follows from Conjecture H(a).

%%%%******

\end{document}